\documentclass[lettersize,journal]{IEEEtran}

\usepackage{newpxtext,newpxmath}

\let\coloneqq\relax
\let\eqqcolon\relax

\usepackage[latin1]{inputenc}
\usepackage{amsthm}
\usepackage{amssymb}
\usepackage{amsmath}
\usepackage[normalem]{ulem}
\usepackage{bbm}
\usepackage{nonfloat}
\usepackage[pdftex, backref=true]{hyperref}
\usepackage{braket}
\usepackage{dsfont}
\usepackage{mathdots}
\usepackage{mathtools}
\usepackage{enumerate}
\usepackage[shortlabels]{enumitem}
\usepackage{csquotes}
\usepackage{stmaryrd}
\usepackage{graphicx}
\usepackage{stackengine}
\usepackage{scalerel}
\usepackage{array}
\usepackage{makecell}
\newcolumntype{x}[1]{>{\centering\arraybackslash}p{#1}}
\usepackage{tikz}
\usepackage{pgfplots}
\usetikzlibrary{shapes.geometric, shapes.misc, positioning, arrows, arrows.meta, decorations.pathreplacing, decorations.pathmorphing, patterns, angles, quotes, calc}
\usepackage{booktabs}
\usepackage{xfrac}
\usepackage{siunitx}
\usepackage{centernot}
\usepackage{comment}
\usepackage{chngcntr}
\usepackage{caption}

\newcommand{\cN}{\mathcal{N}}
\newcommand{\cH}{{\pazocal{H}}}

\newcommand{\cT}{{\pazocal{T}}}
\newcommand{\cB}{{\pazocal{B}}}
\newcommand{\cD}{{\pazocal{D}}}

\newtheorem{thm}{Theorem}
\newtheorem*{thm*}{Theorem}
\newtheorem{prop}[thm]{Proposition}
\newtheorem*{prop*}{Proposition}
\newtheorem{lemma}[thm]{Lemma}
\newtheorem*{lemma*}{Lemma}

\newtheorem*{cor*}{Corollary}

\newtheorem*{cj*}{Conjecture}

\newtheorem*{Def*}{Definition}

\newtheorem*{question*}{Question}

\newtheorem*{problem*}{Problem}

\makeatletter
\def\thmhead@plain#1#2#3{%
  \thmname{#1}\thmnumber{\@ifnotempty{#1}{ }\@upn{#2}}%
  \thmnote{ {\the\thm@notefont#3}}}
\let\thmhead\thmhead@plain
\makeatother

\theoremstyle{definition}
\newtheorem{rem}[thm]{Remark}

\newtheorem{manuallemmainner}{Lemma}
\newenvironment{manuallemma}[1]{%
  \renewcommand\themanuallemmainner{#1}%
  \manuallemmainner \it
}{\endmanuallemmainner}

\newcommand{\bb}{\begin{equation}\begin{aligned}\hspace{0pt}}
\newcommand{\bbb}{\begin{equation*}\begin{aligned}}
\newcommand{\ee}{\end{aligned}\end{equation}}
\newcommand{\eee}{\end{aligned}\end{equation*}}
\newcommand*{\coloneqq}{\mathrel{\vcenter{\baselineskip0.5ex \lineskiplimit0pt \hbox{\scriptsize.}\hbox{\scriptsize.}}} =}
\newcommand*{\eqqcolon}{= \mathrel{\vcenter{\baselineskip0.5ex \lineskiplimit0pt \hbox{\scriptsize.}\hbox{\scriptsize.}}}}

\newcommand{\texteq}[1]{\stackrel{\mathclap{\scriptsize \mbox{#1}}}{=}}
\newcommand{\textleq}[1]{\stackrel{\mathclap{\scriptsize \mbox{#1}}}{\leq}}

\newcommand{\ketbra}[1]{\ket{#1}\!\!\bra{#1}}
\newcommand{\Ketbra}[1]{\Ket{#1}\!\!\Bra{#1}}
\newcommand{\ketbraa}[2]{\ket{#1}\!\!\bra{#2}}

\newcommand{\sumno}{\sum\nolimits}

\newcommand{\id}{\mathds{1}}
\newcommand{\R}{\mathds{R}}
\newcommand{\N}{\mathds{N}}

\newcommand{\C}{\mathds{C}}

\newcommand{\CC}{\mathbb{C}}

\DeclareMathOperator{\Tr}{Tr}

\DeclareMathOperator{\Span}{span}
\DeclareMathAlphabet{\pazocal}{OMS}{zplm}{m}{n}

\DeclareMathOperator{\spec}{sp}

\DeclareMathOperator{\tr}{Tr}

\DeclareMathOperator{\dom}{dom}

\newcommand{\HH}{\pazocal{H}}
\newcommand{\T}{\pazocal{T}}

\newcommand{\D}{\pazocal{D}}

\newcommand{\MM}{\mathcal{M}}

\newcommand{\lsmatrix}{\left(\begin{smallmatrix}}
\newcommand{\rsmatrix}{\end{smallmatrix}\right)}

\stackMath
\newcommand\xxrightarrow[2][]{\mathrel{%
  \setbox2=\hbox{\stackon{\scriptstyle#1}{\scriptstyle#2}}%
  \stackunder[5pt]{%
    \xrightarrow{\makebox[\dimexpr\wd2\relax]{$\scriptstyle#2$}}%
  }{%
   \scriptstyle#1\,%
  }%
}}

\newcommand{\tends}[2]{\xxrightarrow[\! #2 \!]{\mathrm{#1}}}

\stackMath

\makeatletter
\newcommand*\rel@kern[1]{\kern#1\dimexpr\macc@kerna}
\newcommand*\widebar[1]{%
  \begingroup
  \def\mathaccent##1##2{%
    \rel@kern{0.8}%
    \overline{\rel@kern{-0.8}\macc@nucleus\rel@kern{0.2}}%
    \rel@kern{-0.2}%
  }%
  \macc@depth\@ne
  \let\math@bgroup\@empty \let\math@egroup\macc@set@skewchar
  \mathsurround\z@ \frozen@everymath{\mathgroup\macc@group\relax}%
  \macc@set@skewchar\relax
  \let\mathaccentV\macc@nested@a
  \macc@nested@a\relax111{#1}%
  \endgroup
}

\tikzset{meter/.append style={draw, inner sep=10, rectangle, font=\vphantom{A}, minimum width=30, line width=.8, path picture={\draw[black] ([shift={(.1,.3)}]path picture bounding box.south west) to[bend left=50] ([shift={(-.1,.3)}]path picture bounding box.south east);\draw[black,-latex] ([shift={(0,.1)}]path picture bounding box.south) -- ([shift={(.3,-.1)}]path picture bounding box.north);}}}
\tikzset{roundnode/.append style={circle, draw=black, fill=gray!20, thick, minimum size=10mm}}
\tikzset{squarenode/.style={rectangle, draw=black, fill=none, thick, minimum size=10mm}}

\definecolor{Blues5seq1}{RGB}{239,243,255}
\definecolor{Blues5seq2}{RGB}{189,215,231}
\definecolor{Blues5seq3}{RGB}{107,174,214}
\definecolor{Blues5seq4}{RGB}{49,130,189}
\definecolor{Blues5seq5}{RGB}{8,81,156}

\definecolor{Greens5seq1}{RGB}{237,248,233}
\definecolor{Greens5seq2}{RGB}{186,228,179}
\definecolor{Greens5seq3}{RGB}{116,196,118}
\definecolor{Greens5seq4}{RGB}{49,163,84}
\definecolor{Greens5seq5}{RGB}{0,109,44}

\definecolor{Reds5seq1}{RGB}{254,229,217}
\definecolor{Reds5seq2}{RGB}{252,174,145}
\definecolor{Reds5seq3}{RGB}{251,106,74}
\definecolor{Reds5seq4}{RGB}{222,45,38}
\definecolor{Reds5seq5}{RGB}{165,15,21}

\usepackage{algorithmic}
\usepackage[caption=false,font=normalsize,labelfont=sf,textfont=sf]{subfig}
\usepackage{textcomp}
\usepackage{stfloats}
\usepackage{url}
\usepackage{verbatim}
\usepackage{graphicx}
\def\BibTeX{{\rm B\kern-.05em{\sc i\kern-.025em b}\kern-.08em
    T\kern-.1667em\lower.7ex\hbox{E}\kern-.125emX}}
\usepackage{balance}

\allowdisplaybreaks

\usepackage{tikz}

\newcommand{\gu}[1]{U_{\!#1}^{\vphantom{\dag}}}
\newcommand{\gud}[1]{U_{\!#1}^{\dag}}

\begin{document}

\title{Classical shadow tomography for continuous variables quantum systems}

\author{Simon~Becker, Nilanjana~Datta, Ludovico~Lami, and Cambyse~Rouz\'{e}
\thanks{Simon Becker is with the Department of Mathematics, ETH Z\"{u}rich, R\"amistrasse 101, 8092 Z\"urich, Switzerland. Email: simon.becker@math.ethz.ch}
%\affiliation{Department of Mathematics, R\"amistrasse 101, 8092 Z\"urich, Switzerland.}
\thanks{Nilanjana Datta is with the Department of Applied Mathematics and Theoretical Physics, Centre for Mathematical Sciences, University of Cambridge, Cambridge CB3 0WA, United Kingdom. Email: n.datta@damtp.cam.ac.uk}
%\affiliation{Department of Applied Mathematics and Theoretical Physics, Centre for Mathematical Sciences, University of Cambridge, Cambridge CB3 0WA, United Kingdom}
\thanks{Ludovico Lami is with the Korteweg--de Vries Institute for Mathematics, the Institute for Theoretical Physics, and QuSoft, University of Amsterdam, Science Park 123, 1098 XG Amsterdam. Part of this work was conducted while he was at the Institute for Theoretical Physics, University of Ulm, Albert-Einstein-Allee 11, D-89069 Ulm, Germany. Email: ludovico.lami@gmail.com}
%\affiliation{Institut f\"{u}r Theoretische Physik und IQST, Universit\"{a}t Ulm, Albert-Einstein-Allee 11, D-89069 Ulm, Germany}
%\affiliation{QuSoft, Science Park 123, 1098 XG Amsterdam, the Netherlands}
%\affiliation{Korteweg--de Vries Institute for Mathematics, University of Amsterdam, Science Park 105-107, 1098 XG Amsterdam, the Netherlands}
%\affiliation{Institute for Theoretical Physics, University of Amsterdam, Science Park 904, 1098 XH Amsterdam, the Netherlands}
\thanks{Cambyse Rouz\'{e} is with the Zentrum Mathematik, Technische Universit\"{a}t M\"{u}nchen, 85748 Garching, Germany. Email: rouzecambyse@gmail.com}}
%\affiliation{Zentrum Mathematik, Technische Universit\"{a}t M\"{u}nchen, 85748 Garching, Germany}

\maketitle

\begin{abstract}
In this article we develop a continuous variable (CV) shadow tomography scheme with wide ranging applications in quantum optics. Our work is motivated by the increasing experimental and technological relevance of CV systems in quantum information, quantum communication, quantum sensing, quantum simulations, quantum computing and error correction. We introduce two experimentally realisable schemes for obtaining classical shadows of CV (possibly non-Gaussian) quantum states using only randomised Gaussian unitaries and easily implementable Gaussian measurements such as homodyne and heterodyne detection. For both schemes, we show that 
$N=\mathcal{O}\big(\operatorname{poly}\big(\frac{1}{\epsilon},\log\big(\frac{1}{\delta}\big),M_n^{r+\alpha},\log(m)\big)\big)$
samples of an unknown $m$-mode state $\rho$ suffice to learn the expected value of any $r$-local polynomial in the canonical observables of degree $\alpha$, both with high probability $1-\delta$ and accuracy $\epsilon$, as long as the state $\rho$ has moments of order $n>\alpha$ bounded by $M_n$. By simultaneously truncating states and operators in energy and phase space, we are able to overcome new mathematical challenges that arise due to the infinite-dimensionality of CV systems. %{\strike{are absent in finite dimensions}}. 
We also provide a scheme to learn nonlinear functionals of the state, such as entropies over any small number of modes, by leveraging recent energy-constrained entropic continuity bounds. Finally, we provide numerical evidence of the efficiency of our protocols in the case of CV states of relevance in quantum information theory, including ground states of quadratic Hamiltonians of many-body systems and cat qubit states. We expect our scheme to provide good recovery in learning relevant states of 2D materials and photonic crystals.
\end{abstract}

\section{Introduction}
Obtaining classical descriptions of states of quantum-mechanical systems is a fundamental ingredient of quantum computing. It is useful for storing and transmitting quantum information and essential for verification and benchmarking of quantum devices. 
However, the underlying quantum %mechanics governing
{nature of} such systems provides a huge hurdle in obtaining such a classical description: to learn anything about a quantum state one needs to measure it, but measurements in quantum mechanics are %probabilistic and
inherently destructive {and furthermore probabilistic, entailing that} %. Hence, 
individual measurement outcomes only give limited information about the state of the system. % and the measurement itself disturbs the measured state. 
Consequently, in order to obtain a classical representation of a quantum state, one requires multiple identical copies of the state on which successive, appropriate (possibly adaptive) single-copy measurements can be performed. 

This is what is done in the traditional method of learning an unknown quantum state, known as {\em{quantum state tomography}}. It is the process of inferring a quantum state by using suitable measurements on many identical copies of the state. There is, however, a huge practical limitation in using quantum state tomography for many-body quantum systems. This is due to the so-called ``curse of dimensionality'': the number of parameters needed to fully specify the state of a quantum-mechanical system grows exponentially with the system size. The exponential number of measurements needed to infer these parameters makes quantum state tomography infeasible for large systems. Consequently, full quantum state tomography has only been realised in systems with few components, in particular, in a system of ten qubits, which too required millions of measurements. In a nutshell, obtaining a classical description of a $d$-dimensional quantum mixed state $\rho$, given many copies of it, via quantum state tomography, can be shown to require $\Omega(d^2)$ copies of the state. (More recently, it has been shown~\cite{ODonnell2016,Haah2017} that $\mathcal{O}(d^2)$ copies also suffice.) However, this number grows exponentially with the number $n$ of qubits (since $d=2^n$), and {the problem becomes rapidly intractable}.

In 2018~Aaronson~\cite{Aaronson2018} pointed out that for certain concrete tasks, obtaining a complete classical characterisation of the quantum state is unnecessary. Instead it is often sufficient to accurately predict many useful properties of the state. This led him to propose a novel task called {\em{shadow tomography}}, the aim of which is not to learn a complete description of the unknown quantum state but instead to simultaneously estimate the outcome probabilities associated with a list of $M$ two-outcome measurements, $E_1, \ldots, E_M$, performed on the state, up to a desired accuracy (say, $\varepsilon$). For an unknown $d$-dimensional quantum mixed state $\rho$, this requires the prediction of $M$ expectation values, $\Tr[E_i\rho]$ $i=1,2,\ldots, M$, to within an additive error $\varepsilon$. Aaronson showed that, remarkably, the number of copies of the quantum state (i.e.~the {\em{sample size}}) needed to make these predictions, scales {\em{polynomially}} in the system size. Moreover, to predict the $M$ different expectation values, only $\mathcal{O}({\rm{polylog}}( M))$ number of copies of the state are needed. In spite of this advantage with respect to the sample size, implementation of shadow tomography is impractical because it requires very expensive {\em{quantum processing}}, including exponentially long quantum circuits that act collectively across all copies of the unknown quantum state stored in a quantum memory, as well as a lot of storage and post-processing to make the desired predictions. Aaronson's work was followed by a spate of papers, with the best result on the sample complexity of shadow tomography being obtained in~\cite{aaronson2019gentle,Brandao2019}.

In 2020 Huang et al.~\cite{Huang2020} improved on the work of Aaronson, by providing an efficient and experimentally feasible method to learn an unknown quantum state from just a {\em{few simple, single-copy measurements}}. The measurement  outcomes are used to construct a minimal classical representation of the state, called its {\em{classical shadow}}, which can be efficiently stored on a classical computer. This can thereafter be used to predict many linear (and possibly polynomial) properties of the quantum state. For example, the properties could be the expectation values of a list of $M$ observables in the given quantum state, as was the goal of Aaronson~\cite{Aaronson2018}. However, in contrast with the shadow tomography method proposed by Aaronson, Huang et al.~\cite{Huang2020} have a strict divide between the quantum and classical parts of their protocol: after obtaining a classical shadow of the quantum state, all processing necessary to predict its properties are done via {\em{classical computations}}. They proved that, {\em{under certain conditions}}, it is possible to predict expectation values of a list of $M$ observables for an unknown quantum state with a small constant error, with high success probability, by using $\log(M)$ number of copies of the state. They called this method {\em{classical shadow tomography}}. This novel method provides a tractable and rigorous procedure to obtain succinct classical descriptions of quantum many-body states, using which many useful properties of these states can be predicted. 

Classical shadow tomography was originally developed for locally finite-dimensional systems. In contrast, recent years have seen a fast growth of the range of applications of infinite-dimensional, continuous variable (CV) quantum systems, e.g.~collections of electromagnetic modes travelling along an optical fibre or massive harmonic oscillators, in all areas of quantum information~\cite{Braunstein-review}. Notable applications include {quantum communication~\cite{holwer, Wolf2007, TGW, PLOB, MMMM, Rosati2018, exact-solution},} quantum sensing~\cite{aasi2013enhanced, zhang2018noon, meyer2001experimental, mccormick2019quantum}, quantum simulations~\cite{flurin2017observing}, quantum computing and error correction~\cite{Gottesman2001, mirrahimi2014dynamically, Ofek2016, michael2016new, Mazyar2019}, and have been enabled by a steady development of non-classical sources of radiation~\cite{ourjoumtsev2006generating, kurochkin2014distillation, huang2015optical, reimer2016generation, eichler2011observation, zhong2013squeezing, eichler2014quantum, toyli2016resonance}.

Hence, CV systems are of enormous technological and experimental relevance. This comes hand in hand with a pressing need for fast and efficient quantum state tomography of CV systems. The aim of this paper is precisely to devise a rigorous procedure for obtaining classical shadows of states of CV quantum systems, thus developing an efficient and experimentally feasible procedure to benchmark CV quantum technologies.
%. While extending the novel method of classical shadow tomography to obtain useful classical representations of such systems constitutes an important research problem, the present paper is, to our knowledge, the first contribution to it. 

The traditional way of performing tomography of CV systems consists of measuring a functional characterizing the state, typically the characteristic function~\cite{wallentowitz1995reconstruction,zheng2000motional,fluhmann2020direct}, the Wigner~\cite{smithey1993measurement} or the Husimi Q-function~\cite{altepeter20044} from which the quantum state of the system can be reconstructed, either by inverse linear transformations or by statistical inference techniques~\cite{lvovsky2009continuous,babichev2004homodyne,fluhmann2020direct,casanova2012quantum,gerritsma2011quantum,gerritsma2010quantum,zahringer2010realization,johnson2015sensing}. This approach has also been experimentally tested in various settings~\cite{smithey1993measurement,lvovsky2009continuous,kirchmair2013observation,wang2009decoherence,vlastakis2013deterministically,bertet2002direct,deleglise2008reconstruction,hofheinz2009synthesizing,wang2016schrodinger,shen2016optimized,leibfried1996experimental,Ofek2016,ding2017quantum,lv2017reconstruction}. However, as for their discrete analogues, a naive full tomography of a CV quantum system in terms of a quasi-probability distribution is highly inefficient. To remedy this issue, more advanced tomographic schemes were proposed which involve a displacement of the state in phase space followed by parity or multiple photon measurements~\cite{hofheinz2009synthesizing,shen2016optimized,leibfried1996experimental,bertet2002direct,zhang2018noon,lutterbach1997method,vlastakis2013deterministically,guerlin2007progressive}. These measurements are then processed to reconstruct the quasi-probability distributions or their corresponding states in the Fock basis. On the down side, these schemes exhibit a trade-off between the number of measurement points in phase space and the number of operator expectation values measured at each point. More recently, a more efficient method of reconstruction of quasi-probability distributions via Lagrange interpolation was proposed in~\cite{landon2018quantitative}. However, to the best of our knowledge, a rigorous analysis of the sample and computational complexity associated to each of these methods is missing. Moreover, the latter were mostly applied to systems of a small number of modes. Here, instead, we propose a new scheme for building classical shadows of multi-mode continuous variables quantum systems. Our proposal comes with rigorous complexity bounds.
\medskip

\noindent
{\em{Related works:}} Results analogous to ours recently appeared in a concurrent and independent work by Gandhari et al.~\cite{Gandhari2022}. In it too a framework generalizing the qubit-based classical shadow tomography protocol~\cite{Huang2020} to CV systems was developed. 
A key step in their method is to express the density matrix of the reconstructed state in terms of so-called pattern functions. The latter were originally introduced in the context of optical homodyne tomography by D'Ariano et al.~\cite{D'Ariano94} and have been used extensively thereafter in quantum tomography of CV systems (see e.g.~\cite{Artiles2005} and references therein). The authors of~\cite{Gandhari2022} obtained bounds on the sample complexity for estimating quantum states for their protocol by exploiting known bounds on pattern functions~\cite{Artiles2005}. The framework of~\cite{Gandhari2022} is equivalent to ours in the settings of homodyne and heterodyne detection (see Section~\ref{compare} for an explanation). Even though pattern functions do not arise explicitly in our work, they are implicit in our results. This can be seen by a comparison of our results with known expressions~\cite{Richter2000} for the Fourier transform of pattern functions in terms of Laguerre polynomials.

Other recent works on learning CV quantum states (and quantum processes) include~\cite{Rosati2022, Gu2022}.

\subsection{CV classical shadow tomography} \label{subsec:shadow_tomography}

Let us start with a brief summary of our extension of the protocol of Huang et al.~\cite{Huang2020} for obtaining a classical shadow of the quantum state of a continuous variables quantum system. Assume that multiple (say, $N$) identical copies of an unknown quantum state, $\rho$, are available, and one has (i)~an ensemble $\cal{U}$ of unitary operators and (ii)~a quantum measurement described by the set of measurement operators $\{M_x\}_x$, satisfying $\sum_x M_x^\dag M_x = I$, such that elements $M_x U$, $U\in\mathcal{U}$, describe a tomographically complete set of measurements. In Huang et al.~\cite{Huang2020} $\rho$ was considered to be an $n$-qubit state, and $\{M_x\}_x$ was a measurement in the computational basis, in which case $x \in \{0,1\}^n$. In a practical scheme, each ensemble ${\cal{U}}$ should be realisable as an efficient quantum circuit, and also have a succinct classical description.

In analogy with the discrete setting, our proposal for a {\em{CV classical shadow tomography}} contains the following three main ingredients:
\begin{enumerate}[(i)]
\item An $m$-mode CV system in an unknown state $\rho$.
\item A random variable $S$ taking values in the group $\operatorname{Sp}(2m)$ of $2m\times 2m$ symplectic matrices, and the associated unitaries $U_S$. For homodyne measurements we consider e.g.\ random variables distributed according to the Haar measure on $\operatorname{Sp} \cap \operatorname{SO}.$ 
\item A Gaussian positive operator valued measurement  $\big\{ \D(x) \psi \D(-x)\, \frac{d^{2m}x}{(2\pi)} \big\}_{x\in\mathbb{R}^{2m}}$, where $\psi\coloneqq \ketbra{\psi}$ is a fixed pure Gaussian state and $\D(x)$ denotes the unitary displacement operator of phase-space parameter $x\in\mathbb{R}^{2m}$. In practice, we will consider two types of Gaussian measurements, namely homodyne detection along a random direction in phase space and heterodyne detection.
\end{enumerate}
A {\em{CV classical shadow}} of the quantum state $\rho$ is then created by using a randomised protocol that involves repeatedly performing the following simple steps:
\begin{enumerate}[(a)]
\item A symplectic matrix $S$ is selected randomly from $\operatorname{Sp}(2m)$ and applied to a copy of $\rho$, resulting in the unitary transformation $\rho \mapsto U_S\rho U_S^\dagger$. 
\item The Gaussian measurement of effect operators $M_x$ is performed on the output of the previous step, yielding the post-measurement state $\widetilde{\rho}_x = \psi_x\coloneqq \D(x) \psi \D(-x)$, when the measurement outcome is $x$. One can attempt to partly undo the effect of the unitary by counter-rotating $\widetilde{\rho}_x$, implementing the transformation $\widetilde{\rho}_x \mapsto U_S^\dag \widetilde{\rho}_x U_S$. Naturally, due to the measurement such transformation will not yield back the original state $\rho$ which is in general mixed.
\item Averaging over $S,X$, where $X = \mathbb R^{2m}$ is the set of measurement outcomes equipped the probability measure derived from the measurement outcomes, the counter-rotated state yields a quantum channel $\mathcal{M}$ that depends on the measure $\mu$:
\bb \label{noise_M}
\mathcal{M} (\rho) \coloneqq  {\mathbb{E}}_{S,X}\left[\gud{S}\, \widetilde{\rho}_X \gu{S} \right] .
\ee
For our choice of Gaussian measurements, the effective quantum channel $\mathcal{M}$ is a simple linear bosonic channel whose action can be represented as a random displacement in phase space. 
\item Heuristically, one would like to define the {\em{classical shadow}} of $\rho$ as the random operator
\bb \label{shadow_state}
\hat{\rho}_{S,x} \coloneqq&\ \MM^{-1}\! \left( \gud{S}  \widetilde{\rho}_x \gu{S} \right) .
%=&\ \MM^{-1} \left( \frac{U_S^\dag M_x U_S\rho\, U_S^\dag M_x^\dag U_S}{\Tr \left(M_x U_S\rho\, U_S^\dag M_x^\dag\right)} \right)\,.
\ee
%This
Once the nature of the protocol, and in particular the Gaussian measurement $\big\{ \D(x) \psi \D(-x)\, \frac{d^{2m}x}{(2\pi)} \big\}_{x\in\mathbb{R}^{2m}}$, has been specified, the classical shadow is simply a function of the particular realisations $S,x$ obtained in this round of the protocol. Importantly, it is \emph{not} a function of the unknown state $\rho$. The information on $\rho$ is at this point stored \emph{only} in the probability distribution associated with the measurement outcome $x$. As a matter of fact, $\hat{\rho}_{S,x}$ is an unbiased estimator of $\rho$: %recovers {\color{blue}(yields ?)}
\bb \label{averageshadowstate}
\mathbb{E}_{S,X}\big[\hat{\rho}_{S,X}\big] &=\mathcal{M}^{-1}\left(\mathbb{E}_{S,X}\big[\gud{S}\,\widetilde{\rho}_X \gu{S} \big]\right)\\
&=\big(\mathcal{M}^{-1}\!\!\circ\!\mathcal{M}\big)(\rho) = \rho\,.
\ee
\end{enumerate}
This is the final output of one iteration of the protocol. A purely classical description of the operator $\hat{\rho}_{S,x}$ --- which, we stress again, does not depend on the unknown state $\rho$ --- is stored in a classical memory for future processing. Note that such a description can be produced given $S$ and $x$ alone, by simply computing the operator defined by~\eqref{shadow_state}.

Let us summarise the whole procedure. The measurement yields $\widetilde{\rho}_x$ as an approximation for the state $U_S\rho\, U_S^\dag$; then one attempts to find an approximation for the original state, $\rho$, by counter-rotating $\widetilde{\rho}_x$ to obtain ${U_S^\dag} \widetilde{\rho}_x {U_S}$. Finally, by applying $\MM^{-1}$ one eliminates the effect of the average noise (represented by the quantum channel $\MM$) that the state undergoes in the protocol. This yields the classical shadow
$\hat{\rho}_{S,X}$ %whose corresponding (random) parameters $S$ and $X$ are to be
whose description is stored as a string of classical data in a classical memory for future processing. Repeating the above protocol on $N$ independent identical copies of $\rho$ yields a string of classical shadows $\{\hat{\rho}^{(1)}, \ldots, \hat{\rho}^{(N)}\}$, where for simplicity we introduced the shorthand notation $\hat{\rho}^{(i)}\coloneqq \hat{\rho}_{S_i, x_i}$. This string can be further processed %which we can further process
for instance by computing the empirical average
\begin{align*}
\hat{\sigma}^{(N)} \coloneqq \frac{1}{N} \sum_{i=1}^N \hat{\rho}^{(i)}\, .
\end{align*}
The advantage of our CV quantum tomography protocol is most clearly visible from the simple structure of the map $\mathcal{M}$ and its inverse at the level of characteristic functions: for any trace class operator $Z$ with characteristic function $\chi_Z:u\in\mathbb{R}^{2m}\mapsto \tr(Z \D(u))$ and $|\psi\rangle \coloneqq  U_T^\dagger |0\rangle$, for any $T \in \operatorname{Sp}(2m)$,
\begin{align}
&\chi_{\MM(Z)}  = \chi_Z\, f_{\mu, T}\, ,\quad \text{ where } \\
&f_{\mu, T}(u) \coloneqq\ \int \!\! d\mu(S)\, %e^{-\frac12\, x^\intercal (TS)^{^\intercal} (TS)\, u}
e^{-\frac12\, \|TS u\|^2} \label{fmuT}
\end{align}
depends on the Gaussian state $\ket{\psi}$. Therefore, at least formally the classical shadows $\hat{\rho}^{(i)}$ can be equivalently reconstructed by their characteristic functions, which take the form
\begin{equation}\label{characteristicfunctionshadow}
    \chi_{\hat{\rho}^{(i)}} = \chi_{U_{S_i}^\dagger \widetilde{\rho}_{x_i}U_{S_i}^{\vphantom{\dag}}}\,f_{\mu,T}^{-1}\,.
\end{equation}
\subsection{Moment constraints} When trying to implement the above strategy, one faces however two issues that are specific to the continuous variables setting. %First, the quantum channel $\mathcal{M}$ is not invertible as a bounded operator, which, as we will see, implies that the classical shadow $\hat{\rho}_{S,X}$ is in general an unbounded operator, and in particular shows that the functions defined in~\eqref{characteristicfunctionshadow} cannot be associated to a quantum state.
First, as we will see, the quantum channel $\mathcal{M}$ is in general not surjective on the space of trace class operators, which implies that the classical shadow $\hat{\rho}_{S,X}$ will typically not be a trace class --- and, for that matter, not even a bounded --- operator. In other words, the right-hand side of~\eqref{characteristicfunctionshadow}, although well defined as a function on $\R^{2m}$, is not the characteristic function of a quantum state. Second, the parameters $S$ and $X$ that need to be stored in the classical memory at each round of the protocol belong to continuous and unbounded sets.

%To remedy these problems, we will consider the induced action of $\mathcal{M}$ at the level of the characteristic function of the state $\rho$ and will further need to assume that the latter has controlled moments of low degree.
To overcome the first problem, we focus on a suitable characterisation of the classical shadow $\hat{\rho}_{S,X}$ that is well defined, namely, its characteristic function $\chi_{\hat{\rho}_{S,X}}$, and construct the operator $\hat{\rho}_{S,X}$ itself as defined only on a restricted domain. To obtain effective estimates of $\rho$ from the above scheme, we will further need to assume that $\rho$ has controlled moments of low degree. Such an assumption allows us to show that the projection of the state $\rho$ onto the finite subspace of Hilbert--Schmidt operators supported on the space of low energy Fock states is sufficient for obtaining a good enough approximation of the state via the Gaussian shadow tomography protocol.

Let us make these considerations more precise. We consider the maps %More precisely, we consider the maps 
\bb
\mathcal P_{M}(T) \coloneqq  \sum_{\mathbf{n}_1,\mathbf{n}_2 \in \{0,..,M\}^m}  \langle \mathbf{n}_1|T|\mathbf{n}_2\rangle\, \ketbraa{\mathbf{n}_1}{\mathbf{n}_2}\,,
\ee
where $\ket{\mathbf{n}}$ is a multivariate Fock state, and $\mathbf{n}\in\{0,..., M\}^m$ is a multi-index. Using a simple approximation scheme, we also approximate the Schwartz operators $\ketbraa{\mathbf{n}_1}{\mathbf{n}_2}$ by Schwartz operators $\widetilde Z_{\mathbf{n}_1\mathbf{n}_2}$ with smooth compactly supported characteristic functions and define an auxiliary map
\begin{equation}
 \widetilde{\mathcal{P}}_{M}(T)\coloneqq \sum_{\mathbf{n}_1,\mathbf{n}_2 \in \{0,..,M\}^m}  \Tr(\widetilde  Z_{\mathbf{n}_2\mathbf{n}_1}T) \, \ketbraa{\mathbf{n}_1}{\mathbf{n}_2}\,.
 \end{equation}
Our goal is to approximate possibly unbounded observables of an $m$-mode CV quantum system. For this, we introduce a norm which captures approximations of such operators: fix an arbitrary positive integer $n$; then given $0<\alpha<n$ and two states $\rho,\rho'$ with $\tr(\rho (I+N_m)^n),\tr(\rho' (I+N_m)^n)<\infty$, where $N_m$ stands for the $m$-mode number operator, we denote $H_m\coloneqq I+N_m$ and
\begin{align*}
\|X\|^{(\alpha)}_1 \coloneqq \left\|H_m^{\alpha/2}X H_m^{\alpha/2}\right\|_1\,
\end{align*}
for any trace-class $X$ or which the right-hand side is finite. 
Then the following bounds hold which put the above approximation scheme on a rigorous footing:
\begin{prop}\label{approxmomentconstraints}
\label{prop:EC2}
Let $N_m$ be the number operator over $m$ modes and let $\rho$ be an $m$-mode state such that $\tr(N_m^n\rho)\eqqcolon  mE<\infty$. 
Then, for any $\eta^2\ge 2M^2$,
\begin{align*}
&\|\rho-\mathcal{P}_{M}(\rho)\|_1^{(\alpha)}\le 2(1+M)^{\frac{\alpha-n}{2}}E\, , \\
&\|\rho-\widetilde{\mathcal{P}}_{M}(\rho)\|^{(\alpha)}_1 \le 2(1+M)^{\frac{\alpha-n}{2}} E \\
&\hspace{18.6ex} +(3mM)^{6mM+\alpha} e^{-\frac{\eta^2m}{2}} \eta^{2Mm}.
\end{align*}
\end{prop}

\section{Main results}

\subsection{CV classical shadows via homodyne detection}
We first consider the scenario in which one performs a homodyne detection along a random direction in phase space. More precisely, $m$ independent random matrices $S_1,\dots , S_m$ are distributed uniformly (according to the Haar measure) on the intersection $\operatorname{Sp}(2) \cap \operatorname{SO}(2)$ between the symplectic and the special orthogonal group, with corresponding angles $\theta_1,\cdots, \theta_m\in [-\pi,\pi]$. The homodyne measurement is then performed along the position axis and yields the classical outcome $x=(x_1,...,x_m)\in\mathbb{R}^{2m}$. As in the qubit setting, we now restrict ourselves to an arbitrary subset $A$ of $|A|\le r$ modes. In that case, the characteristic functions of the reduced shadows defined in~\eqref{characteristicfunctionshadow} are distributions of the form
\begin{align*}
\chi_{\hat{\rho}_{A}^{(i)}}(u_A)=\prod_{j\in A}\sqrt{2}\pi \|u_j\|\, \delta \left( (S_j u_j)_2 \right) e^{- i u_j^\intercal \Omega S_j x_j}\,,
\end{align*}
for any $u_A\coloneqq \{u_j\}_{j\in A}\in\mathbb{R}^{2|A|}$, where $\delta(x)$ is the Dirac distribution at $x\in\mathbb{R}$, where $(S_ju_j)_2\coloneqq \sum_{k} (S_j)_{2k} (u_j)_k$, and where $\Omega$ stands for the canonical symplectic form as defined in Equation~\eqref{Omega}. As foreseen in the previous paragraph, this characteristic function is not square integrable, and therefore cannot be associated to a quantum state. Indeed, we recall that, by Plancherel's theorem, any trace class operator, and thereby any quantum state, gives rise to a square integrable characteristic function, % by Plancherel's theorem, 
and for any two such operators $T_1,T_2$,
\begin{align}\label{Plancherel0}
    \tr[T_1^\dagger T_2]=\int \frac{d^{2|A|}x}{(2\pi)^{|A|}}\, \chi_{T_1}(x)^* \chi_{T_2}(x)\,\,\equiv \braket{\chi_{T_1},\,\chi_{T_2}}\,.
\end{align}
Instead, we construct a random matrix $\hat{\rho}^{(i)}_A(M)$ in the set
$\mathbb{M}_{(M+1)^{ |A|}}(\mathbb{C})$ of matrices of size $(M+1)^{ |A|}$ by simply extending~\eqref{Plancherel0}:
\begin{align*}
   \hat{\rho}^{(i)}_{A}(M)\coloneqq \sum_{\mathbf{n}_1,\mathbf{n}_2\in\{0,...,M\}^A}\,\langle \chi_{\ket{\mathbf{n}_1}\!\bra{\mathbf{n}_2}}, \chi_{\hat{\rho}^{(i)}_A}\rangle\,\ketbraa{\mathbf{n}_1}{\mathbf{n}_2}\,.
\end{align*}
By~\eqref{averageshadowstate}, we know that the expected matrix $\mathbb{E}[\hat{\rho}^{(i)}_{A}(M)]$ coincides with the unknown state $\rho$. Moreover, by Bernstein's matrix concentration inequality, we have that, with high probability the empirical average $\sigma_A^{(N)}(M)\coloneqq \frac{1}{N}\sum_{i=1}^N\hat{\rho}^{(i)}_{A}(M)$ will well-approximate $\mathcal{P}_M(\rho_A)$ on all sets $A$ of size $|A|\le r$ for $N=\mathcal{O}\big(\operatorname{poly}\big(e^r,\,\log(m)\big)\big)$. Combining this with the approximation bounds under moment constraints derived in Proposition~\eqref{approxmomentconstraints} leads us to our first main result:
\begin{thm}[(CV classical shadows via homodyne detection)] 
With the notation introduced above, given $0\le \alpha <n$ and assuming that $E^{(n)}_r\coloneqq \max_{|A|\le r}\tr(\rho_A H_r^n)<\infty$ for all $A$ of size $|A|\le r$, we have that for integer $M=\Big\lceil \Big(\frac{4E_r^{(n)}}{\epsilon}\Big)^{\frac{2}{n-\alpha}}\Big\rceil$, $N=\mathcal{O}\Big(\operatorname{poly}\Big(\frac{1}{\epsilon^2},\,M^{r+\alpha},\log\big(\frac{1}{\delta}\big),\log(m)\Big)\Big)$ and any region $A$ of size $|A|\le r$,
\begin{equation}
      \left\|\sigma^{(N)}_A(M)-\rho_A\right\|^{(\alpha)}_1\le \epsilon
\end{equation}
 with probability at least $1-\delta$. Similarly, for $N=\mathcal{O}\Big(\operatorname{poly}\Big(\frac{1}{\epsilon^2},\,M^{r+\alpha},\log\big(\frac{1}{\delta}\big),\log(L)\Big)\Big)$ we have that for any set of $L$ observables $O_j$ on regions $A_j$ of size at most $r$ and with $\|H_{r}^{-\frac{\alpha}{2}}O_jH_{r}^{-\frac{\alpha}{2}}\|_\infty\le 1$, 
\begin{align*}
    \max_j \big|\tr[O_j\,(\sigma_A^{(N)}(M)-\rho_{A_j})]\big|\le \epsilon
\end{align*}
with probability at least $1-\delta$.

\end{thm}

\subsection{CV classical shadows via heterodyne detection}
We also consider the case of a heterodyne detection
\bb
\left\{ \frac{1}{(2\pi)^{m/2}} \ketbra{x}\right\}_{x\in \R^{2m}}\, ,
\ee
i.e.\ $T=I$, and all unitaries employed are passive, i.e.\ such that $\big[ \gu{S}, \frac12 R^\intercal R\big] =0$. In that case, we show that the function defined in~\eqref{fmuT} takes the simpler form \[f_{\mu,T}(u) = e^{-\frac12 \|u\|^2}.\]
 Therefore, the classical shadow will have improper characteristic function
\begin{align*}
\chi_{\hat{\rho}^{(i)}_A}(u_A) = e^{  \frac14 \|u_A\|^2}\prod_{j\in A}e^{-iu_j^\intercal \Omega x_j}\,.
\end{align*}
In this case, we consider the matrices
\begin{align*}
    \hat{\rho}_A^{(i)}(M)\coloneqq \sum_{\mathbf{n}_1,\mathbf{n}_2\in\{0,...,M\}^A}\,\langle \chi_{\widetilde{Z}_{\mathbf{n}_1\mathbf{n}_2}}, \chi_{\hat{\rho}^{(i)}_{A}}\rangle\,\ketbraa{\mathbf{n}_1}{\mathbf{n}_2}\,.
\end{align*}

From here, repeating essentially the same argument as in the case of homodyne detection, defining the matrix $\widetilde{\sigma}_A^{(N)}(M) \coloneqq \frac{1}{N}\sum_{i=1}^N\hat{\rho}^{(i)}_{A}(M)$, we arrive at our second main result:

\begin{thm}[(CV classical shadows via heterodyne detection)] 
With the above notation, given $0\le \alpha <n$ and assuming that $E^{(n)}_r\coloneqq \max_{|A|\le r}\tr(\rho_A H_r^n)<\infty$ for all $A$ of size $|A|\le r$, we have that for $N=\mathcal{O}\Big(\operatorname{poly}\Big(\frac{1}{\epsilon^2},\,(E_r^{(n)})^{r+\alpha},\log\big(\frac{1}{\delta}\big),\log(m)\Big)\Big)$ and any region $A$ of size $|A|\le r$,
\begin{align*}
      \|\widetilde{\sigma}^{(N)}_A(M)-\rho_A\|^{(\alpha)}_1\le \epsilon
\end{align*}
 with probability at least $1-\delta$. Similarly, for $N=\mathcal{O}\Big(\operatorname{poly}\Big(\frac{1}{\epsilon^2},\,(E_r^{(n)})^{r+\alpha},\log\big(\frac{1}{\delta}\big),\log(L)\Big)\Big)$ we have that for any set of $L$ observables $O_j$ on regions $A_j$ of size at most $r$ and with $\|H_{r}^{-\frac{\alpha}{2}}O_jH_{r}^{-\frac{\alpha}{2}}\|_\infty\le 1$, 
\begin{align*}
    \max_j \big|\tr[O_j\,(\widetilde{\sigma}_A^{(N)}(M)-\rho_{A_j})]\big|\le \epsilon
\end{align*}
with probability at least $1-\delta$. 

\end{thm}

\section{Notation and basic notions}

\subsection{Operators and norms}

Given a separable Hilbert space $\cH$, we denote by $\cB(\cH)$ the space of bounded linear operators on $\cH$, and by $\cT_p(\cH)$ the \textit{Schatten $p$-class}, which is the Banach subspace of $\cB(\cH)$ formed by all bounded linear operators whose Schatten $p$-norm, defined as $\|X\|_{p}=\left(\Tr|X|^p\right)^{1/p}$,  is finite. Henceforth, we refer to $\cT_1(\cH)$ as the set of \textit{trace class} operators. The set of quantum states (or density matrices), i.e.\ positive semi-definite operators $\rho \in \cT_1(\cH)$ of unit trace, is denoted by $\cD(\cH)$. The Schatten $1$-norm, $\|\cdot\|_1$, is the {trace norm}, and the corresponding induced distance (e.g.\ between quantum states) is the {trace distance}. Note that the Schatten $2$-norm, $\|\cdot\|_2$, coincides with the \textit{Hilbert--Schmidt norm}. We also recall that a quantum channel with input system $A$ and output system $B$ is any completely positive, trace-preserving (CPTP) linear map $\cN:\cT_1(\cH_A)\to\cT_1(\cH_B)$, where $\cH_A, \cH_B$ are the Hilbert spaces corresponding to $A,B$, respectively. 

\medskip

If $\rho$ is a quantum state with spectral decomposition $\rho=\sum_i p_i \ketbra{\phi_i}$, and $A$ is a positive semi-definite operator, the \textit{expected value} of $A$ on $\rho$ is defined as
\bb
\Tr[\rho A]\coloneqq \sum_{i:\, p_i>0} p_i \big\|A^{1/2}|{\phi_i}\rangle\big\|^2 \in \mathbb{R}_+\cup \{+\infty\}\, ;
\label{expected positive}
\ee
here we use the convention that $\Tr[\rho A]=+\infty$ if the above series diverges or if there exists an index $i$ for which $p_i>0$ and $\ket{\phi_i}\notin \dom\left(A^{1/2}\right)$. This definition can be extended to a generic densely defined self-adjoint operator $A$ on $\cH$, by considering its decomposition $A=A_+-A_-$ into positive and negative parts, with $A_\pm$ being positive semi-definite operators with mutually orthogonal supports. The operator $A$ is said to have a \textit{finite expected value on $\rho$} if (i)~$\ket{\phi_i}\in \dom\big(A_+^{1/2}\big)\cap \dom\big(A_-^{1/2}\big)$ for all $i$ for which $p_i>0$, and (ii)~the two series $\sum_i p_i \big\|A_\pm^{1/2} \ket{\phi_i}\big\|^2$ both converge. In this case, the following quantity is called the \textit{expected value} of $A$ on $\rho$:
\bb
\Tr[\rho A]\coloneqq \sum_{i:\, p_i>0} p_i \big\|A_+^{1/2} \ket{\phi_i}\big\|^2 - \sum_{i:\, p_i>0} p_i \big\|A_-^{1/2} \ket{\phi_i}\big\|^2\,.
\label{expected}
\ee
Obviously, for a pair of operators $A,B$ satisfying $A\geq B$, we have that $\Tr[\rho A]\geq \Tr[\rho B]$.

\medskip

Let $A$ be an (unbounded) operator $A$ on some Banach space ${{X}}$, with domain $\dom(A)$. Such an operator is called closed if its {\em{graph}}, that is $\left\{(\ket{x},A\ket{x}); \ket{x} \in \dom(A) \right\} \subset X \times X,$ is closed. The spectrum of a closed operator $A$ is defined as the set~\cite[Definition~9.16]{HALL}
\[\spec(A)\coloneqq \left\{ \lambda \in \mathbb C:\, \lambda I-A \text{ is not continuously invertible} \right\}.\]
Henceforth, we often suppress the identity operator $I$ in the expression $(\lambda I - A)$ for notational simplicity.
% Here, a closed operator $B$ is said to be not bijective if there exists no bounded operator $C$ with the property that: (i)~for all $\ket{\psi}\in \cH$, one has that $K\ket{\psi}\in \dom(B)$, and moreover $B K\ket{\psi}=\ket{\psi}$; and (ii)~for all $\ket{\psi}\in \dom(B)$, it holds that $KB\ket{\psi}=\ket{\psi}$. 
We remind the reader that the spectrum of a self-adjoint positive operator is a closed subset of the positive real half-line~\cite[Proposition~9.20]{HALL}. Given a possibly unbounded operator $X$, $L_X$ stands for the left multiplication by $X$: $L_X(Y)=XY$, whereas $R_X$ stands for right multiplication by $X$: $R_X(Y)=YX$, whenever these products are well-defined.

\subsection{Continuous variable quantum systems}

A CV system with $m$ modes is defined on the Hilbert space $\cH_m\coloneqq L^2(\R^m)$, equipped with the multi-mode Fock basis $\{|\mathbf{n}\rangle\}_{\mathbf{n}\in\mathbb{N}^m}$ of eigenvectors of the number operator $N_m$:
\begin{align}\label{def:numberop}
    N_m|\mathbf{n}\rangle=\Big(\sumno_{j=1}^mn_j\Big)\,|\mathbf{n}\rangle\,.
\end{align}
We denote the canonical operators on each mode as $X_j,P_j$ ($j=1,\ldots,m$). Define the formal vector
\bb\nonumber
R\coloneqq \begin{pmatrix} X_1,\ldots, X_n,P_1, \ldots, P_n \end{pmatrix}^\intercal 
\ee
and the symplectic form
\bb
\Omega   \coloneqq \begin{pmatrix} 0 & \id \\ -\id & 0 \end{pmatrix} 
\label{Omega}
\ee
(where all blocks are $m \times m$ matrices), in terms of which the canonical commutation relations read (at least when evaluated on Schwartz functions)
\bb
[ R_j, R_k ] = i\Omega_{jk}\, .
\label{CCR}
\ee
We also introduce also the annihilation and creation operators $a_j^\dag, a_j$ ($j=1,\ldots, m$), defined by
\bb
a_j \coloneqq \frac{X_j + i P_j}{\sqrt2}\, ,\qquad a_j^\dag \coloneqq \frac{X_j - i P_j}{\sqrt2}\, .
\ee
In terms of these operators, the single-mode Fock states can be constructed as
\bb
\ket{n}\coloneqq \frac{(a^\dag)^n}{\sqrt{n!}}\ket{0}\, ,
\label{Fock_creation}
\ee
with similar formulae holding for the multi-mode case. The canonical commutation relations can also be written as $[a_j, a_k^\dag] = \delta_{jk}$.

With a slight abuse of notation, we will often denote by the same symbol $\Omega$ the symplectic form for different sets of modes. The quantum covariance matrix $V[\rho]$ and mean vector $t[\rho]$ associated with a generic state $\rho$ are defined by
\bb \label{firstsecondmoments}
t[\rho]_j \coloneqq \Tr \left(\rho R_j\right)\, ,\quad V[\rho]_{jk} \coloneqq \Tr \left(\rho \left\{R_j-t_j,\, R_k-t_k\right\}\right)\, ,
\ee
provided that these expressions are well defined. For an arbitrary $x\in \R^{2m}$, we define the associated displacement operator by
\bb
\D(x) \coloneqq e^{- i x^\intercal \Omega R}\, .
\label{D}
\ee
Note that $\D(x)^\dag = \D(x)^{-1} = \D(-x)$. By writing $x = x' \oplus x''$, where $x' \in \R^m$ groups together the first $m$ components of $x$ and $x''\in \R^m$ the last $m$, one can also introduce the complex vector
\bb
\alpha(x) \coloneqq \frac{1}{\sqrt2} \left( x' + i x''\right) ,
\label{alpha}
\ee
in terms of which we have that~\cite[Eq.~(3.3.30)--(3.3.31)]{BARNETT-RADMORE}
\begin{align}
\D(x) &= \exp \left[\sumno_{j=1}^m \left( \alpha_j(x) a_j^\dag - \alpha_j(x)^* a_j\right) \right] \label{D_complex} \\
&= e^{-\frac14 \|x\|^2} e^{\sum_j \alpha_j(x) a_j^\dag} e^{- \sum_j \alpha_j(x)^* a_j} , \label{D_split_complex}
\end{align}
where $\|x\|^2 \coloneqq \sum_j x_j^2$. In dealing with continuous variable systems, one can stick to the real notation, employing real vectors $x$, or move to the complex one, which uses the complex vectors $\alpha(x)$. In this paper we will mostly follow the former convention; however, it will be useful, occasionally, to use the latter too. In general, a prompt translation between one set of conventions and the other can be obtained by means of~\eqref{alpha}.

Another important identity involving displacement operators is 
\bb
\int \frac{d^{2m} x}{(2\pi)^m}\, \D(x) Z \D(-x) = I \Tr Z\, ,
\label{identity}
\ee
valid for all trace class operators $Z \in \cT_1(\cH_m)$, with the integral on the left-hand side converging in the weak sense (see~\cite[Proposition 3.5.1]{HOLEVO}). Using displacement operators, we can re-write~\eqref{CCR} in Weyl form as
\bb
\D(x+y) = e^{\frac{i}{2} x^\intercal\Omega y}\, \D(x) \D(y)\, .
\label{Weyl}
\ee
Coherent states are instead defined as
\bb
\ket{x} \coloneqq \D(x) \ket{0}\, ,
\label{coherent}
\ee
where $\ket{0}$ is the vacuum state. We can decompose $\ket{x}$ mode-wise as $\ket{x} = \bigotimes_{j=1}^m \ket{x^{(j)}}$, where $x^{(j)}\coloneqq (x_j, x_{m+j})^\intercal \in \R^2$ is the sub-vector of $x$ obtained by picking the coordinates corresponding to the $j^\text{th}$ coordinate and momentum, and $\ket{x^{(j)}}$ is a single-mode coherent state. In general, the latter can be represented in the single-mode Fock basis according to the identity
\bb
\ket{y} &= e^{-\|y\|^2/4} \sum_{n=0}^\infty \frac{(y_1+i y_2)^n}{\sqrt{2^n\, n!}} \ket{n} \\
&= e^{-\|y\|^2/4} \sum_{n=0}^\infty \frac{\alpha(y)^n}{\sqrt{n!}} \ket{n}\, ,
\label{coherent_Fock}
\ee
where $\alpha(y)\in \C$ is defined in~\eqref{alpha}.

For an arbitrary trace class operator $Z$, we can construct its characteristic function $\chi_Z:\R^{2m}\to \C$ by
\bb
\chi_Z(x) \coloneqq \Tr \big[Z\, \D(x)\big]\, .
\label{chi}
\ee
Characteristic functions are always bounded and furthermore continuous, because of the strong operator continuity of the mapping $x\mapsto \D(x)$. Moreover, in the sense of weak operator convergence it holds that
\bb
Z = \int \frac{d^{2m}x}{(2\pi)^m}\, \chi_Z(x)\, \D(-x)\, .
\label{reconstruction_from_chi}
\ee
By applying~\eqref{Weyl} and~\eqref{coherent_Fock}, one can prove that
%For example, for coherent states one has that
%\bb
%\chi_{\ket{\xi}\!\bra{\xi}}(u) = e^{-\frac14 \|u\|^2 - iu^\intercal\Omega \xi}\, .
%\label{chi_coherent}
%\ee
%From this expression, we find 
\bb \label{eq:char_func}
&\chi_{\ket{x}\!\bra{y}}(u) \\ 
&= \tr[\ket{x}\!\bra{y} \D (u)] \\ &= \braket{y|\D(u)|x} \\ &= \braket{0|\D(-y) \D(u) \D (x)|0} \\
&= e^{-\frac{i}{2} u^\intercal \Omega x} \braket{0| \D(-y) \D(u+x) | 0} \\
&= e^{-\frac{i}{2} u^\intercal \Omega x} e^{\frac{i}{2} y^\intercal \Omega (u+x)} \braket{0|\D(u+x-y)|0} \\ 
&= e^{-\frac{i}{2} u^\intercal \Omega x} e^{\frac{i}{2} y^\intercal \Omega (u+x)} e^{-\frac14 \|u + x-y \|^2}.
%\end{split}
\ee
As special cases, we conclude e.g.\ that
\bb \label{eq:char_func2}
\chi_{\ket{x}\!\bra{x}}(u) &= e^{-\frac14 \|u\|^2 - iu^\intercal\Omega x} , \\
\chi_{\ket{-x}\!\bra{x}}(u) &= e^{-\frac14 \|u - 2x\|^2} , \\
\chi_{\ket{x}\!\bra{-x}}(u) &= e^{-\frac14 \|u + 2x\|^2} . 
\ee
Given two single-mode Fock states $|k\rangle,|j\rangle$~\cite{folland2016harmonic} and $x=x'\oplus x''$: 
\begin{align}\label{eq.charfock}
\chi_{|k\rangle\langle j|}(x)=\left\{\begin{aligned}
&\sqrt{\frac{k!}{j!}}e^{-\frac{\pi}{2}|\omega|^2}(\sqrt{\pi}\omega)^{j-k}L_k^{(j-k)}(\pi|\omega|^2),~~ j\ge k\\
&\sqrt{\frac{k!}{j!}}e^{-\frac{\pi}{2}|\omega|^2}(-\sqrt{\pi}\omega)^{k-j}L_j^{(k-j)}(\pi|\omega|^2),~~ j\le k\,,
\end{aligned}\right.
\end{align}
where $\omega\coloneqq -x'+i\frac{x''}{2\pi}$. Above, the functions $L_k^{(j)}$ are the Laguerre polynomials, defined for any two integers $k,j$ as
\begin{align}\label{eq.Laguerre}
    L_k^{(j)}(x)\coloneqq \sum_{l=0}^k \frac{(k+j)!}{(k-l)!(j+l)!}\,\frac{(-x)^l}{l!}\,.
\end{align}

Interestingly, the correspondence between trace class operators and characteristic functions is injective --- even more strikingly, it can be extended to an isometry between the space of Hilbert--Schmidt operators and that of square integrable functions $\R^{2m}\to \C$~\cite[Theorem~5.3.3]{HOLEVO}. A consequence of the existence of this isometry is the quantum Plancherel theorem, which tells us that for any two trace class operators $Z, Z' \in \cT_1(\cH_m)$,
\bb
\Tr \big[Z^\dag Z'\big] = \int \frac{d^{2m} x}{(2\pi)^m}\, \chi_Z(x)^* \chi_{Z'}(x)\, .
\label{Plancherel}
\ee
The canonical commutation relations are invariant under so-called symplectic unitaries, constructed as follows. A $2m\times 2m$ real matrix $S$ such that $S\Omega S^\intercal = \Omega$ (or equivalently $S\Omega = \Omega S^{-\intercal}$) is called a symplectic matrix. From the defining relation it can be immediately seen that any symplectic matrix must satisfy $\det S=\pm 1$; however, remarkably, it turns out that in fact all symplectic matrices have determinant $1$. To any symplectic matrix we can associate a symplectic unitary $U_S$ acting on $\cH_m$. This is defined by either of the following relations
\bb
U_S^\dag R U_S^{\phantom{\dag}} = SR\, ,\qquad U_S^{\phantom{\dag}} \D(x) U_S^\dag = \D(Sx)\, ,
\label{S}
\ee
where the first identity is to be understood coordinatewise: $(SR)_j = \sum_k S_{jk} R_k$. Note that symplectic matrices form a group, denoted as $\operatorname{Sp}(2m)$, and that the correspondence $S\mapsto U_S$ is a group homomorphism.
In particular,
\bb
U_S^\dag = U_{S}^{-1} = U_{S^{-1}} = U_{\Omega\, S^{\text{\raisebox{.6pt}{$\intercal$}}} \Omega^{\text{\raisebox{.6pt}{$\intercal$}}}} .
\label{U_S_under_inversion}
\ee

Also, from~\eqref{S} we deduce that
\bb
t\left[ U_S^{\phantom{\dag}} \rho\, U_S^\dag \right] = S\, t[\rho]\, ,\qquad V\left[ U_S^{\phantom{\dag}} \rho\, U_S^\dag \right] = S\, V[\rho] S^\intercal\, ,
\label{V_transformation_S}
\ee
where we recall that $t[\omega]$ and $V[\omega]$ denote the mean vector and quantum covariance matrix of the state $\omega$ as defined in~\eqref{firstsecondmoments}. A generic Gaussian unitary is obtained as the product between a symplectic unitary and a displacement operator. States obtained by applying an arbitrary Gaussian unitary to the vacuum state $\ket{0} = \bigotimes_{j=1}^m \ket{0}_j$ are called pure Gaussian states. Often times, displacements can be ignored; we will thus write an arbitrary pure Gaussian state with zero mean as
\bb
\ket{\psi} = U_S \ket{0}\, ,
\label{pure_Gaussian_state}
\ee
where $S\in \operatorname{Sp}(2m)$ is an arbitrary symplectic matrix.

\medskip 
A quantum channel that will be particularly useful to us is the Gaussian white noise channel, defined for $\lambda > 0$ by
\bb
\cN_\lambda(\cdot) \coloneqq \int \frac{d^{2m}x}{(2\pi\lambda)^m}\, e^{-\frac{\|x\|^2}{2\lambda}} \D(x) (\cdot) \D(-x)\, .
\label{noise_random_displacement}
\ee
Using this formula one can show that
\bb
\cN_\lambda:\chi_Z \longmapsto \chi_{\cN_\lambda(Z)} (x) \coloneqq \chi_Z(x)\, e^{-\frac{\lambda}{2}\|x\|^2}\, .
\label{noise_chi}
\ee
Curiously, for $\lambda\in (0,1]$ its action can be expressed alternatively as
%\footnote{To the best of our knowledge this (simple) expression is not reported often, although we believe it to be part of the folklore. \textcolor{red}{CR: would it be worth showing this?}} 
\bb
\cN_\lambda(\cdot) = \int\!\! \frac{d^{2m} x}{(2\pi \lambda)^m}\, \D(x)\, \tau_{\!\frac{1}{2\lambda}-\frac12}^{\otimes m} \D(x)^\dag (\cdot) \D(x)\, \tau_{\!\frac{1}{2\lambda}-\frac12}^{\otimes m} \D(x)^\dag\! ,
\label{noise_thermal_representation}
\ee
where $\tau_\nu$ is the single-mode thermal state with mean photon number $\nu$, given by
\bb
\tau_\nu = \frac{1}{\nu+1} \sum_{n=0}^\infty \left( \frac{\nu}{\nu+1}\right)^n \ketbra{n}\, ,
\label{thermal}
\ee
where $\ket{n}$ stands for the $n^{\text{th}}$ Fock state. Since we could not locate a complete proof of~\eqref{noise_thermal_representation} in the existing literature, we provide a self-contained one in Appendix~\ref{folklore_app}. A special case of~\eqref{noise_thermal_representation} is when $\lambda=1$, in which case $\tau_{\frac{1}{2\lambda}-\frac12} = \ketbra{0}$ (the vacuum state) and
\bb
\cN_1(\cdot) = \int \frac{d^{2m} x}{(2\pi)^m}\, \ketbra{x} (\cdot) \ketbra{x}\, .
\label{noise_1}
\ee
In terms of the real mean vector and quantum covariance matrix, for all $\lambda>0$ we have that
\begin{align}\label{noise_V}
\cN_\lambda:\left\{ \begin{array}{lll} V & \longmapsto & V+2\lambda I\, , \\ t & \longmapsto & t\, . \end{array}\right.
\end{align}
The above channel, $\cN_\lambda$, is just an example within the larger class of Gaussian channels. To construct the most general Gaussian channel, take two arbitrary $2m\times 2m$ real matrices $X$ and $Y$ such that
\bb
Y+i\Omega - iX\Omega X^\intercal \geq 0\, ;
\label{CP_Gaussian}
\ee
the corresponding Gaussian channel, denoted as $\mathcal{G}_{X,Y}$, then acts as
\bb
\mathcal{G}_{X,Y}: \chi_Z\mapsto \chi_{\mathcal{G}_{X,Y}(Z)}(u) \coloneqq \chi_Z\left(\Omega^\intercal X^{\intercal} \Omega\, u\right) e^{-\frac14 u^\intercal \Omega^\intercal Y \Omega u}
\label{G_channel_chi}
\ee
and
\bb
\mathcal{G}_{X,Y}:\left\{ \begin{array}{lll} V & \longmapsto & XVX^\intercal +Y\, , \\ t & \longmapsto & Xt\, . \end{array}\right.
\label{G_channel_V}
\ee
It is worth observing that~\eqref{CP_Gaussian} implies that $Y\geq 0$ is positive semi-definite.

\medskip

An important class of operators on $\cH_m$ that we consider in this paper is the set of Schwartz operators~\cite{Keyl2016}. %The latter
They can be defined as those trace class operators whose characteristic function is a Schwartz function on $\mathbb{R}^{2m}$. We denote the set of Schwartz function as $\mathfrak{S}(\mathbb{R}^{2m})$, and that of Schwartz operators as $\mathfrak{S}(\cH_m)$. In particular, we consider the set $\mathfrak{S}(\cH_m)_0$ of Schwartz operators whose characteristic functions are compactly supported. 

As spaces of Schwartz functions are usually employed as test spaces in the rigorous theory of distributions, we can use the space $\mathfrak{S}(\cH_m)_0$ to formalise the definition of objects --- called symbols --- that would be ill-defined as operators in the traditional sense. For example, given \emph{any} smooth function $\chi$ on $\mathbb{R}^{2m}$, we can construct a symbol $\rho_\chi$ with `characteristic function' $\chi$. This is defined formally as a functional $\rho_\chi:\mathfrak{S}(\cH_m)_0 \to \C$ acting as
%For such operators, and to a smooth function $\chi$ on $\mathbb{R}^{2m}$, we associate a functional $\mathfrak{S}(\cH_m)_0\ni Z\mapsto \rho_{\chi}(Z)$ such that for any compactly supported Schwartz function $\chi_Z$, which is the characteristic function of an operator $Z\in\mathfrak{S}(\cH_m)_0$ defined via Plancherel's theorem~\eqref{Plancherel}, 
\begin{align}\label{symbol}
   \rho_{\chi}(Z)\coloneqq \int \frac{d^{2m}x}{(2\pi)^m}\, \chi(x)^* \chi_Z(x)\,.
\end{align}
This expression %in terms of a trace on the left-hand side above
is justified by the fact that, when $\sigma\in \cT_1(\cH_m)$, we have $\rho_{\chi_\sigma}(Z)=\tr(\sigma Z)$ by Plancherel's theorem~\eqref{Plancherel}. %with $\chi\equiv \chi_\sigma$.
The above functional extends to the whole space $\mathfrak{S}(\cH_m)$ %$Z\in\mathfrak{S}(\cH_m)$
whenever the function $x\mapsto \chi(x)^*\chi_Z(x)$ is integrable.

\subsection{Concentration inequalities}
% In this paper, we make use of powerful  concentration inequalities which we recall here for the reader's ease: we denote by $S^{n-1}$ the $(n-1)$-sphere of radius $1$.

% \begin{lemma}[L\'{e}vy's lemma, see~\cite{aubrun2017alice}]\label{Levy}
% Let $n> 2$. If $f:S^{n-1}\to\mathbb{R}$ is a $1$-Lipschitz function, then for any $t>0$,
% \begin{align*}
%     \mathbb{P}_{\mu_{S^{n-1}}}\Big(f>\mathbb{E}[f]+t\Big)\le \operatorname{exp}\big(-nt^2/2\big)\,,
% \end{align*}
% where $\mu_{S^{n-1}}$ denotes the Haar measure on the $(n-1)$-sphere.
% \end{lemma}

In this paper, we make use of Bernstein's matrix inequality in order to prove that the probability that, on a well-chosen finite-dimensional subspace, the output of our shadow tomography protocol is far from the original unknown state decays exponentially fast in the number $N$ of samples used to gather statistics:

\begin{lemma}[(Bernstein's matrix inequality~\cite{Tropp2011})] Given $N$ i.i.d.\ random matrices $X_1,\dots,X_N\in \mathbb{M}_n(\mathbb{C})$ which obey $\|X_i-\mathbb{E}[X_i]\|_\infty\le R$ almost surely, for some $R >0$, the following tail bound holds:

\bb
\label{matrixBern}
&\mathbb{P}\bigg(\bigg\|\frac1N\! \sum_{i=1}^N\,(X_i\!-\!\mathbb{E}[X_i])\bigg\|_\infty\!\!\ge {\epsilon}\bigg) \le 2n\, e^{-\frac{N\epsilon^2}{2\Sigma^2 + 2R\epsilon/3}} ,
%\exp\Big(-\frac{N\epsilon^2}{2\Sigma^2\!+\!2R\epsilon/3}\Big)
\ee
where the constant $\Sigma$ is defined as $\Sigma^2\coloneqq \|\mathbb{E}[X_1^2]\|_\infty<\infty$.
\end{lemma}

\section{Classical shadow tomography of a CV system} \label{sec:CVtomo}

In this section, we will have a closer look at the shadow tomography scheme sketched in Section~\ref{subsec:shadow_tomography}. The goal of the procedure is to construct a good estimator of an unknown $m$-mode state $\rho$, by measuring as few i.i.d.\ copies of $\rho$ as possible. To this end, we repeatedly sample symplectic matrices from $\operatorname{Sp}(2m)$ according to some probability distribution $\mu$, apply the corresponding symplectic unitary $U_S$ on one copy of $\rho$, implementing the transformation $\rho\mapsto U_S^{\phantom{\dag}} \rho U_S^\dag$, and subsequently perform a fixed Gaussian measurement $\big\{ \psi_x\, \frac{d^{2m}x}{(2\pi)} \big\}_{x\in\mathbb{R}^{2m}}$ on that same state. Here, $\psi_x = \ketbra{\psi_x}$ with $\ket{\psi_x}\coloneqq \D(x)\ket{\psi}$; the Gaussian state $\ket{\psi}$, which uniquely identifies the Gaussian measurement, is a fixed parameter of the shadow tomography protocol. Without loss of generality, we can take $\ket{\psi}$ to have zero mean, in which case, according to~\eqref{pure_Gaussian_state}, we can introduce a symplectic matrix $T\in \operatorname{Sp}(2m)$ satisfying that $U_T^\dag \ket{0} = \ket{\psi}$. The $\dag$ here is immaterial, thanks to~\eqref{U_S_under_inversion}. Now, acting with a displacement operator on the left and on the right yields immediately
\bb
\ket{\psi_x} = \D(x) \ket{\psi} = \gud{T} \left( \gu{T} \D(x) \gud{T} \right) \ket{0} = \gud{T} \ket{Tx}\, ,
\label{psi_x_rewrite}
\ee
where the last step is due to the action of symplectic unitaries on displacement operators, see~\eqref{S}, and to the definition~\eqref{coherent} of coherent states.

The measurement makes the system collapse into a random state $\widetilde{\rho}_x = \psi_x$, where $x$ is distributed with probability distribution
\bb
p_\rho(x|S)\, d^{2m}x = \braket{\psi_x|U_S^{\phantom{\dag}}\rho U_S^\dag|\psi_x} \frac{d^{2m}x}{(2\pi)^{m}}\, .
\ee
When combined with the probability measure $\mu$ on $\operatorname{Sp}(2m)$, this yields a joint probability distribution on $\operatorname{Sp}(2m)\times \mathbb{R}^{2m}$.

We then attempt to undo the effect of the symplectic unitary by applying $U_S^\dag$. This amounts to the mapping $\widetilde{\rho}_x = \psi_x \mapsto \gud{S}\widetilde{\rho}_x \gu{S}$. Once we average over the random variable $S$ and the random variable $X$ whore realisation we denoted with $x$, the whole process yields an effective noisy channel $\mathcal{M}$ modelled as~\eqref{noise_M}. Making this more explicit, we write the action of $\MM$ on an arbitrary trace class operator $Z\in \T(\HH_m)$ as
\bb \label{noise_M_explicit}
\mathcal{M}(Z) = \int \!\! d\mu(S) \int\!\! \frac{d^{2m}x}{(2\pi)^m}\, \braket{\psi_x| \gu{S} Z\, U_S^\dag |\psi_x}\, \gud{S} \psi_x \gu{S}\, .
\ee

%The classical shadow produced by each round of the protocol is therefore $\hat{\rho}_{S,x} = \MM^{-1}\! \left( \gud{S}  \widetilde{\rho}_x \gu{S} \right)$. As we remarked in Section~\ref{subsec:shadow_tomography}, once the Gaussian measurement --- in practice, the symplectic matrix $T$ --- is given, $\hat{\rho}_{S,x}$ does not depend on $\rho$ anymore, but only on the realisations $S,x$ obtained in this particular round. We can therefore, at least in principle, store a classical description of $\hat{\rho}_{S,x}$ in

Our first result allows us to express the action of $\MM$ in a form that is more easily amenable to investigation with phase space methods.

\begin{lemma}
The map $\mathcal{M}:\cT_1(\cH_m)\to\cT_1(\cH_m)$ defined in~\eqref{noise_M} can be re-expressed as 
\bb
\MM(Z) = \int \!\! d\mu(S)\, \mathcal{G}_{I,\, 2 (TS)^{-1} (TS)^{-\intercal}} (Z)\, ,
\label{M_G_expression}
\ee
where $\mathcal{G}_{X,Y}$ is the Gaussian channel defined in~\eqref{G_channel_chi}. 
Moreover, the characteristic function of $\mathcal{M}(Z)$ satisfies
\bb
\chi_{\MM(Z)} (x) &= \int \!\! d\mu(S)\, \chi_Z(x)\, %e^{- \frac12\, x^\intercal (TS)^{^\intercal} (TS)\, x} 
e^{-\frac12 \left\| TS x\right\|^2} = \chi_Z(x)\, f_{\mu, T}(x)\, ,  
\label{MM_chi}
\ee
where
\bb
f_{\mu, T}(x) \coloneqq \int \!\! d\mu(S)\, e^{-\frac12 \left\| TS x\right\|^2} . \label{f_mu_T}
\ee
\end{lemma}

\begin{proof}
We start by noticing that due to~\eqref{psi_x_rewrite} and thanks the fact that $S\mapsto \gu{S}$ is a group homomorphism, one obtains that
\bb
U_S^\dag \ket{\psi_x} = U_S^\dag U_T^\dag \ket{Tx} = \left(U_T U_S\right)^\dag \ket{Tx} = U_{TS}^\dag \ket{Tx}\, .
\label{U_S_psi_x}
\ee
Thus, for any trace class operator $Z \in \cT_1(\cH_m)$
\begin{align}
&\int\!\! \frac{d^{2m}x}{(2\pi)^m}\, \braket{\psi_x| \gu{S} Z U_S^\dag |\psi_x}\, \gud{S} \psi_x \gu{S} \nonumber \\
&\quad \texteq{(i)} \int\!\! \frac{d^{2m}x}{(2\pi)^m}\, \braket{Tx| \gu{TS} Z \gud{TS} |Tx}\, \gud{TS} \ketbra{Tx} \gu{TS} \nonumber\\
&\quad \texteq{(ii)} \int\!\! \frac{d^{2m}y}{(2\pi)^m}\, \braket{y| \gu{TS} Z \gud{TS} |y}\, \gud{TS} \ketbra{y} \gu{TS} \label{fixed_S_channel} \\
&\quad = \gud{TS} \left(\int\!\! \frac{d^{2m}y}{(2\pi)^m}\, \ketbra{y} \gu{TS} Z \gud{TS} \ketbra{y} \right) \gu{TS} \nonumber \\
&\quad \texteq{(iii)} \gud{TS}\, \cN_1\left(\gu{TS} Z \gud{TS} \right) \gu{TS} \nonumber
\end{align}
Here, in~(i) we applied~\eqref{U_S_psi_x} twice; in~(ii) we changed variable, defining $y\coloneqq Tx$, and used the fact that $T$, being symplectic, has determinant $1$; and in~(iii) we employed the representation in~\eqref{noise_1} for the action of $\cN_1$.

Now, let us compute the action of the above transformation at the level of covariance matrices. By applying~\eqref{V_transformation_S} and~\eqref{noise_V}, we see that
\begin{align*}
&V\left[ \gud{TS}\, \cN_1\left(\gu{TS} Z\, \gud{TS} \right) \gu{TS} \right] \\
&\qquad = (TS)^{-1} V\left[ \cN_1\left(\gu{TS} Z\, \gud{TS} \right) \right] (TS)^{-\intercal} \\
&\qquad = (TS)^{-1} \left( V\left[\gu{TS} Z\, \gud{TS} \right] + 2I \right) (TS)^{-\intercal} \\
&\qquad = (TS)^{-1} \left( TS\, V[Z] (TS)^\intercal + 2I \right) (TS)^{-\intercal} \\
&\qquad = V[Z] + 2 (TS)^{-1} (TS)^{-\intercal}\, .
\end{align*}
Comparing the above calculation with~\eqref{G_channel_V}, we see that
\bb
\gud{TS}\, \cN_1\left(\gu{TS} Z\, \gud{TS} \right) \gu{TS} = \mathcal{G}_{I,\, 2 (TS)^{-1} (TS)^{-\intercal}}(Z)\, .
\ee
Using this insight in~\eqref{fixed_S_channel} shows that
\bb
\int\!\! \frac{d^{2m}x}{(2\pi)^m}\, \braket{\psi_x| \gu{S} Z U_S^\dag |\psi_x}\, \gud{S} \psi_x \gu{S} = \mathcal{G}_{I,\, 2 (TS)^{-1} (TS)^{-\intercal}}(Z)\, .
\ee
In turn, the above identity yields~\eqref{M_G_expression} upon integration in $S$ with respect to the measure $\mu$.

We conclude by computing the characteristic function~\eqref{chi} of both sides of~\eqref{M_G_expression}. For an arbitrary $x\in \mathbb{R}^{2m}$, we obtain that
\bb
\chi_{\MM(Z)}(x) &= \int d\mu(S)\, \chi_{\mathcal{G}_{I,\, 2 (TS)^{-1} (TS)^{-\intercal}}(Z)}(x) \\
&\texteq{(iv)} \int d\mu(S)\, \chi_Z(x) e^{-\frac12 x^\intercal \Omega^\intercal (TS)^{-1} (TS)^{-\intercal} \Omega x} \\
&\texteq{(v)} \int d\mu(S)\, \chi_Z(x) e^{-\frac12 x^\intercal (TS)^{\intercal} (TS) x} \\
&= \int d\mu(S)\, \chi_Z(x) e^{-\frac12 \left\| TS x\right\|^2} .
\ee
Here, (iv)~follows from~\eqref{G_channel_chi}, and (v)~is due to the fact that since $TS$ is symplectic, $\Omega^\intercal (TS)^{-1} = (TS)^\intercal \Omega^\intercal$, or upon transposing $(TS)^{-\intercal} \Omega = \Omega (TS)$.
\end{proof}

We therefore see that the action of $\MM$ is actually very simple, amounting to a point-wise multiplication at the level of the characteristic function. Channels of this form are particular examples of so-called \emph{linear bosonic channels}, introduced and studied by Holevo and Werner~\cite{holwer}. Although we will not use this observation in this work, it is worth noting that linear bosonic channels are always approximable --- in the strong operator sense --- by \emph{Gaussian dilatable} channels, i.e.\ channels admitting a Stinespring representation in which the unitary is Gaussian and the ancilla is arbitrary~\cite{KK-VV, G-dilatable}.

As it turns out, $\MM$ is a special type of linear bosonic channel whose action is representable as a random displacement. Namely, using~\eqref{Weyl} one sees that
\begin{align}
\MM(Z) =&\ \int \frac{d^{2m}x}{(2\pi)^m}\, \widetilde{f}_{\mu,T}(x) \, \D(x) Z \D(-x)\, , \label{MM_random_displacement} \\
\widetilde{f}_{\mu,T}(x) \coloneqq&\ \int \frac{d^{2m}u}{(2\pi)^m}\, f_{\mu,T}(u) \, e^{i x^\intercal \Omega u} = f_{\mu,T}(x)\, ,
\label{Fourier_f_mu_T}
\end{align}
i.e.\ $f_{\mu,T}$ coincides with its own Fourier transform. The rigorous proof of the general validity of~\eqref{Fourier_f_mu_T} is deferred to Appendix~\ref{Fourier_f_mu_T_app}. In light of this discussion, it is easy to write down, at least formally, the inverse of $\MM$, which acts as
\bb
\MM^{-1} : \chi_Z \longmapsto \chi_{\MM^{-1}(Z)} (x) \coloneqq \frac{1}{f_{\mu,T}(x)}\, \chi_Z(x)\, .
\label{inverse_MM_chi}
\ee
Note that $f_{\mu,T}(x)>0$ for all $x\in \R^{2m}$, so the above expression is always well defined. In particular, $\MM$ is injective as a linear map, and thus it is invertible on its range.

The problem, naturally, is that such a range is in general significantly smaller than the space of trace class operators (we will shortly see an example of this). This entails that the right-hand side of~\eqref{inverse_MM_chi} is not always the characteristic function of a trace class operator. This, however, does not pose any problem since  the classical shadow obtained at the output of $\MM^{-1}$ is, as the name suggests, just a classical object which we will merely use as a computational tool.
%{\color{green} ND: I suggest replacing the following sentence by the above sentence in red:} We will have to live with this problem, which is not as serious as it may seem, primarily because the classical shadow obtained at the output of $\MM^{-1}$ is, as the name suggests, just a classical object which we will merely use as a computational tool.

\medskip

Incidentally, also $\MM^{-1}$, just like $\MM$, acts as a mere point-wise multiplication at the level of characteristic functions. This implies that one can also try, as done in~\eqref{MM_random_displacement}, to represent it as an affine combination of displacement operators, i.e.\ by writing
\bb
\MM^{-1} (Z) = \int \frac{d^{2m}x}{(2\pi)^m}\, \widetilde{g}_{\mu,T}(x) \, \D(x) Z \D(-x)\, ,
\ee
where $g_{\mu,T}(x) \coloneqq \frac{1}{f_{\mu,T}(x)}$, and
\bb
\widetilde{g}_{\mu,T}(x) \coloneqq \int \frac{d^{2m}u}{(2\pi)^m}\, g_{\mu,T}(u) \, e^{i x^\intercal \Omega u} 
\ee
would be its Fourier transform. The trouble, of course, is that $g_{\mu,T}$ will not be absolutely integrable --- and not even bounded --- in general, so there is little hope to define its Fourier transform unless one appeals to the theory of distributions. 

\medskip

To provide a solution to this apparent issue, let us return to our original problem. We can now formally construct the classical shadow
\bb\label{sigmaformaldef}
\hat{\rho} &= \MM^{-1}\left( \gud{S} \D(x) \psi \D(-x) \gu{S} \right) \\
&= \MM^{-1}\left( \gud{TS} \ketbra{Tx} \gu{TS} \right) \,\nonumber .
\ee
The shadow, $\hat{\rho}$, is more rigorously defined as a functional on the set $\mathfrak{S}(\cH_m)_0$ of Schwartz operators with compactly supported characteristic functions via~\eqref{symbol}. Its corresponding improper characteristic function can hence be computed as follows:
\bb
\hat{\chi}(u)\equiv \chi_{\hat{\rho}}(u) &= \chi_{\MM^{-1}\left( \gud{TS} \ket{Tx}\!\bra{Tx} \gu{TS} \right)}(u) \\
&= \frac{1}{f_{\mu,T}(u)}\, \chi_{\gud{TS} \ket{Tx}\!\bra{Tx} \gu{TS}}(u) \\
&= \frac{1}{f_{\mu,T}(u)}\, \chi_{\ket{Tx}\!\bra{Tx}}(TSu) \\
&= {\frac{1}{f_{\mu,T}(u)}\, e^{-\frac14 \left\|TSu\right\|^2 - i u^\intercal \Omega S^{-1} x}}\, ,
\label{chi_sigma}
\ee
where in the last step we used~\eqref{eq:char_func2}. In other words, for any $Z\in\mathcal{S}(\cH_m)$:
\begin{align}\label{hatrho}
    \hat{\rho}(Z)\coloneqq \int \frac{d^{2m}u}{(2\pi)^m}\,\hat{\chi}(u)^*\,\chi_Z(u)\, ,
\end{align}
whenever the function $u\mapsto \hat{\chi}(u)^* \chi_Z(u)$ is integrable. In what follows, we will also consider the reduced shadow over a subset $A$ of $|A|=r$ modes, formally given by the partial trace $``\hat{\rho}_A\coloneqq \tr_{A^c}(\hat{\rho})"$ of the shadow $\hat{\rho}$. Again, we will use the characteristic function to rigorously define it: given a region $A$ of $r$ modes, it is defined for any $u_A\in\mathbb{R}^{2r}$ as
\begin{align}\label{partialtraceshadow}
    \hat{\chi}_A(u_A)\equiv \chi_{\hat{\rho}_A}(u_{A})\coloneqq \chi_{\hat{\rho}}(u_A,0)\,.
\end{align}
In that case, we write for any $Z_A\in \mathfrak{S}(\cH_r)$:
\begin{align}\label{defrhoA}
    \hat{\rho}_A(Z_A)\coloneqq \int \frac{d^{2r}u_A}{(2\pi)^r}\,\hat{\chi}_{A}(u_A)^*\,\chi_{Z_A}(u_A)\, ,
\end{align}
whenever the function $u_A\mapsto \hat{\chi}_{A}(u_A)^*\,\chi_{Z_A}(u_A)$ is integrable.

The following lemma further justifies the claim made in~\eqref{averageshadowstate} that the shadow $\hat{\rho}$ has 
average $\rho$ by construction.

\begin{lemma}\label{partialparsevalaverage}
For any subset $A$ of $|A|=r$ modes, and all $Z_A\in \mathfrak{S}(\cH_r)_0$ with corresponding characteristic function $\chi_{Z_A}$, the random variable $ \hat{\rho}_A(Z_A)$ defined via~\eqref{defrhoA} is integrable and 
\begin{align*}
\mathbb{E}\big[   \hat{\rho}_A( Z_A)\big]\,= \tr[\rho_A Z_A]\,,
\end{align*}
where the conditional expectation is taken with respect to the probability density function $p_T(S,x)\coloneqq \braket{Tx|U_{TS}^{\phantom{\dag}}\rho U_{TS}^\dag|Tx}$ with respect to %$\mu\otimes \frac{\operatorname{Leb}}{(2\pi)^m}$
$\mu\otimes \frac{d^{2m}x}{(2\pi)^m}$ on $\operatorname{Sp}(2m) \times \mathbb{R}^{2m}$. The result extends to $Z_A\in\mathfrak{S}(\cH_r)$ under the condition of integrability with respect to %$\mu\otimes \frac{\operatorname{Leb}}{(2\pi)^m}\otimes \frac{\operatorname{Leb}}{(2\pi)^r}$
$\mu\otimes \frac{d^{2m}x}{(2\pi)^m} \otimes \frac{d^{2r}u}{(2\pi)^r}$ of the function
\bb
(S,x,u)\mapsto \frac{\chi_{Z_A}(u)\, \chi_{\ket{Tx}\!\bra{Tx}}(TS(u,0)^\intercal)\, p_T(S,x)}{f_{\mu,T}(u,0)}\, .
\ee
\end{lemma}

% \begin{lemma}\label{Fubaverage}
% For any $Z\in\mathfrak{S}(\cH_m)_0$ with corresponding characteristic function $\chi_Z$, the random variable $\tr(\sigma Z)$ defined via~\eqref{symbol} and~\eqref{chi_sigma} is integrable and 
% \begin{align*}
%     \mathbb{E}\big[\tr(\sigma Z)\big]=\tr({Z}\rho)\,,
% \end{align*}
% where the conditional expectation is taken with respect to the probability density function $p_T(S,x)\coloneqq \langle Tx|U_{TS}\rho U_{TS}^\dagger|Tx\rangle$ with respect to $\mu\otimes \frac{\operatorname{Leb}}{(2\pi)^m}$ on $\operatorname{Sp}(2m) \times \mathbb{R}^{2m}$. The result extends to $Z\in\mathfrak{S}(\cH_m)$ under the condition of integrability with respect to $\mu\otimes \frac{\operatorname{Leb}}{(2\pi)^m}\otimes \frac{\operatorname{Leb}}{(2\pi)^m}$ of the function $(S,x,u)\mapsto \chi_Z(u)\chi_{|Tx\rangle\langle Tx|}(TSu)p_T(S,x)\,f_{\mu,T}(u)^{-1}$.

% \end{lemma}

\begin{proof}
We present the proof for $|A|=m$ since the case $|A|=r<m$ follows the exact same strategy using that $\hat{\chi}(u_A,0)=\hat{\chi}_{A}(u_A)$ by definition. Now, on the one hand, if $Z\in \mathfrak{S}(\cH_m)_0$ then the characteristic function $\chi_Z$ of $Z$ is compactly supported; therefore, owing to the boundedness of $\chi_Z$ we deduce that
%Since
the function
\bb
(S,x,u)\mapsto \frac{\chi_Z(u)\chi_{\ket{Tx}\!\bra{Tx}}(TSu)\,p_T(S,x)}{f_{\mu,T}(u)}
\ee
is integrable with respect to $\mu\otimes \frac{d^{2m}x}{(2\pi)^m} \otimes \frac{d^{2m}u}{(2\pi)^m}$. On the other hand, if only $Z\in \mathfrak{S}(\cH_m)$ such integrability is assumed by hypothesis. Therefore, in both cases thanks to Fubini's theorem and~\eqref{chi_sigma} we have that
%\begin{widetext}    
\begin{align*}
&\mathbb{E}\big[\hat{\rho}( Z)\big] \\
&\ =\! \int \!\!\!\frac{d^{2m}u}{(2\pi)^m}\! \int\!\!\! d\mu(S) \!\int\!\!\! \frac{d^{2m}x}{(2\pi)^m} \,\chi_Z(u)\,\frac{\chi_{\ket{Tx}\!\bra{Tx}}(TSu)\, p_T(S,x)}{f_{\mu,T}(u)} \\
&\ \texteq{(i)}\! \int\!\!\! \frac{d^{2m}u}{(2\pi)^m} \! \int\!\!\! d\mu(S)\, \chi_Z(u)\,\frac{\chi_{\cN_1\left(U_{TS}^{\phantom{\dag}}\rho U_{TS}^\dag\right)}(TSu)}{f_{\mu,T}(u)} \\
&\ \texteq{(ii)}\! \int\!\!\! \frac{d^{2m}u}{(2\pi)^m} \! \int\!\!\! d\mu(S)\, \frac{\chi_Z(u)\,\chi_{U_{TS}^{\phantom{\dag}}\rho U_{TS}^\dag}(TSu)\,e^{-\frac{1}{2}\|TSu\|^2}}{f_{\mu,T}(u)} \\
&\ \texteq{(iii)}\! \int\!\!\! \frac{d^{2m}u}{(2\pi)^m} \! \int\!\!\! d\mu(S)\, \frac{\chi_Z(u)\,\chi_{\rho}(u)\,e^{-\frac{1}{2}\|TSu\|^2}}{f_{\mu,T}(u)} \\
&\ =\! \int\!\!\! \frac{d^{2m}u}{(2\pi)^m}\, \chi_Z(u)\,\chi_{\rho}(u)\\
&\ =\tr[Z\rho]\,.
\end{align*}
%\end{widetext}
In~(i) we changed variable, defining $y\coloneqq Tx$, used the fact that $T$, being symplectic, has determinant $1$, and employed the representation in~\eqref{noise_1}; in~(ii) we used~\eqref{noise_chi}; finally, in~(iii) we leveraged~\eqref{S}.
\end{proof}

In this section, we have made rigorous our first intuitive notion of a CV shadow by means of its characteristic function. We have also seen that the latter reduces to the characteristic function %recovers that
of the original unknown state $\rho$ on average when integrated against sufficiently smooth functions. The goal of the next section is to prove that under some physically relevant conditions such as energy boundedness of the state $\rho$, these integrals are enough to estimate expected values of observables with respect to the state $\rho$ to high accuracy.

\subsection{Finite moments assumption}

As we saw in the previous section, the shadow $\hat{\rho}$, which we formally defined through its characteristic function, is in general unbounded. In this section, we show that this issue can be fixed if we further assume that the photon number distribution of the unknown state $\rho$ satisfies some %energy constraint.
moment constraints. This allows us to show that the projection of $\rho$ onto a certain finite subspace of Hilbert--Schmidt operators is enough to get a sufficiently good approximation of it while running the shadow tomography protocol.

More precisely, we argue that all we need is to ensure finite rank convergence of density operators and an energy constraint on the initial state $\rho$.

% We recall that the quantum characteristic function of a function $F$ is defined by

% \[ \chi(x,\xi) = \int_{\mathbb R^n} e^{2\pi i q y} F(y+\tfrac{\xi}{2},y-\tfrac{\xi}{2}) \ dy.\]
% In addition, the map $L^2(\mathbb R^{2n}) \ni F\mapsto \chi \in L^2(\mathbb R^{2n})$ is unitary.

% By density, there is $\widetilde \chi \in C_c^{\infty}(\mathbb R^{2n})$ such that $\Vert \chi-\widetilde \chi \Vert_{L^2(\mathbb R^{2n})} <\varepsilon$ for $\varepsilon>0$ fixed. 

Let $Z\in\cT_2(\cH_m)$ be a Hilbert Schmidt operator and denote by $\chi_Z$ its characteristic function. By density, one can find a function $\widetilde \chi_Z \in C_c^{\infty}(\mathbb R^{2m})$ such that $\Vert \chi_Z-\widetilde \chi_Z \Vert_{L^2(\mathbb R^{2m})}<\varepsilon$ and $\widetilde \chi_Z$ is the characteristic function of some operator $\widetilde{Z}$, i.e. $\chi_{\widetilde{Z}}=\widetilde{\chi}_Z$. The operator $\widetilde Z$ is a Schwartz operator, and is in particular trace class, and clearly $\Vert Z -\widetilde Z \Vert_2 <\varepsilon$~\cite{Keyl2016}. Thus, $\hat{\rho}( \widetilde{Z})$ is now necessarily a well-defined quantity by the quantum Plancherel formula (cf.~\eqref{Plancherel}). In particular, we have that, for any $Z\in\cT_2(\cH_m)$,
\begin{align*}
 \vert \Tr(Z\rho)-\hat{\rho}(\widetilde Z)\vert &\le \vert \Tr((Z-\widetilde Z) \rho) + \vert \Tr(\rho \widetilde Z )-\hat{\rho}(\widetilde{Z})\vert \\
 &\le \varepsilon +\vert \Tr(\rho \widetilde Z )-\hat{\rho}(\widetilde{Z})\vert\,.
 \end{align*}
This can be used in the above context by choosing $Z=Z_{\mathbf{n}_1\mathbf{n}_2} = \ketbraa{\mathbf{n}_1}{\mathbf{n}_2}$, given two multi-mode Fock states $\ketbraa{\mathbf{n}_1}{\mathbf{n}_2}\in \cH_m$, and defining
\begin{equation}
\label{eq:projection1}
 \mathcal P_{M}(T) \coloneqq  \sum_{\mathbf{n}_1,\mathbf{n}_2 \in \{0,..,M\}^m}  \langle \mathbf{n}_1|T|\mathbf{n}_2\rangle\, \ketbraa{\mathbf{n}_1}{\mathbf{n}_2}\equiv P_M T P_M\,.
 \end{equation}
Using the above approximation scheme, we can approximate the Schwartz operators $Z_{\mathbf{n}_1\mathbf{n}_2}$ by Schwartz operators $\widetilde Z_{\mathbf{n}_1\mathbf{n}_2}$ with smooth compactly supported characteristic functions and define an auxiliary map
\begin{equation}
\label{eq:projection2}
\widetilde{\mathcal P}_{M}(T)\coloneqq \sum_{\mathbf{n}_1,\mathbf{n}_2 \in \{0,..,M\}^m}  \Tr\left[\widetilde{Z}_{\mathbf{n}_2\mathbf{n}_1}T\right] \ketbraa{\mathbf{n}_1}{\mathbf{n}_2}\, .
\end{equation}
For sake of simplicity, we will use the same notation for projected reduced states on subsets of $r<m$ modes. We further assume that the characteristic function of $\widetilde Z_{\mathbf{n}_1\mathbf{n}_2}$ is obtained from that of $\ketbraa{\mathbf{n}_1}{\mathbf{n}_2}$ by point-wise multiplication by a simple compactly supported function. Namely, we set
\begin{align}\label{tildechin1n2}
\widetilde{\chi}_{\mathbf{n}_1\mathbf{n}_2}(u) \coloneqq \widetilde{\chi}_{Z_{\mathbf{n_1}\mathbf{n}_2}}(u) = \chi_{\ket{\mathbf{n}_1}\!\bra{\mathbf{n}_2}}(u) \prod_{j=1}^m \xi_{\eta,R}\big(u^{(j)}\big)\, ,
\end{align}
%is the product $\chi_{\ket{\mathbf{n}_1}\!\bra{\mathbf{n}_2}}\cdot \xi_{\eta,R}^{\otimes m}$ of $\chi_{\ket{\mathbf{n}_1}\!\bra{\mathbf{n}_2}}$ with the i.i.d.~product of a $1$-mode smooth function $0\le\xi_{\eta,R}\le 1$ with $0<\eta<R$ satisfying
where as usual $u^{(j)}\coloneqq (u_j, u_{j+m})^\intercal \in \R^2$, and, given some parameters $0<\eta<R$, the smooth function $\xi_{\eta,R}:\R^2\to [0,1]$ satisfies that
\bb \label{xi}
\xi_{\eta,R}(z) = \left\{\begin{array}{ll} 1, & \qquad \text{if $\|z\|\le \eta$,} \\[1ex] 0, & \qquad \text{if $\|z\|\ge R$.} \end{array} \right.
\ee
In particular, this allows us to evaluate $\rho_\chi \big(\widetilde Z_{\mathbf{n}_2\mathbf{n}_1}\big)$ for functionals $\rho_\chi$ defined as in~\eqref{symbol}. 

% \begin{lemma} \label{boundPPtilde}
% With the above notation, we have that for all $\alpha>0$ and $\eta^2> 2M^2$,
% \begin{align*}
%     \|L_{(I+N_m)^\alpha}R_{(I+N_m)^{\alpha}}(\mathcal{P}_M-\widetilde{\mathcal{P}}_{{M}})\|_{2\to 1}\le  \delta_0(\eta,M,\alpha,m)\,,
% \end{align*}
% for some explicit rapidly vanishing function $\eta\mapsto \delta_0(\eta,M,\alpha,m)$. A crude bound on $\delta_0$ for $\eta\ge 4M^2$ is
% \begin{align*}
%     \delta_0(\eta,M,\alpha,m)\le (3mM)^{6mM+2\alpha}\,e^{-\frac{\eta^2m}{2}}\,\eta^{2Mm}\,.
% \end{align*}
% \end{lemma}

% {
% [I propose a new version of this lemma, presented below (proofs in Appendix~\ref{double_truncation_app}). It has some considerable advantages over the old: first, it is explicit; second, the decay in $\eta$ is qualitatively the same, i.e.\ $O\big(e^{-\frac{m}{4} \eta^2} \eta^{2mM} \big)$ but comes with a substantially better coefficient. In the above statement, the coefficient is $(3mM)^{6mM+2\alpha}$, which grows extremely fast with $M$. In the new version below, we have instead $\sim \frac{M^{2\alpha+m} m^{2\alpha} 3^{mM}}{2^{mM} \big((2M)!\big)^{m/2}}$, which actually \emph{tends to zero, and does so rapidly,} when $M\to\infty$. This discussion deals with the behaviour of the bound when $\eta$ is very large compared to $M$, but I suspect that its qualitative features remain the same throughout the range of parameters.]}

\begin{lemma} \label{double_truncation_lemma}
For all $\alpha>0$, all non-negative integers $m$ (number of modes) and $M$ (Fock truncation number), and all $\eta\geq 0$, it holds that
\bb
&\left\|\mathcal{L}_{(I+N_m)^\alpha} \mathcal{R}_{(I+N_m)^{\alpha}}\left(\mathcal{P}_M-\widetilde{\mathcal{P}}_{{M}}\right)\right\|_{2\to 1} \leq \delta_0(\eta, M, \alpha, m)\, ,
\ee
where the rapidly vanishing function $\delta_0$ is defined by
\bb
&\delta_0(\eta, M, \alpha, m) \\[-1ex]
&\quad \coloneqq (m M+1)^{2\alpha} (M+1)^m\, 3^{mM} e^{-\frac{m}{4} \eta^2} \left(\sum_{p=0}^{2M} \frac{\eta^{2p}}{2^p p!} \right)^{m/2} .
\ee
\end{lemma}

Next, we consider an approximation of $\rho$ in the following norm: given $0\le \alpha<n$ and two trace class operators $\rho,\rho'\in\cT_1(\cH_m)$ with $\tr(\rho (I+N_m)^n),\tr(\rho' (I+N_m)^n)<\infty$, we denote $H_m\coloneqq I+N_m$ and
 \begin{align*}
     \|\rho-\rho'\|^{(\alpha)}_1 \coloneqq \left\|H_m^{\alpha/2}(\rho-\rho')H_m^{\alpha/2}\right\|_1\,.
 \end{align*}
Similar norms were previously defined in~\cite{QCLT} under the name of m-mode bosonic Sobolev norms.

\begin{prop} \label{prop:EC}
Let $N_m$ be the number operator on $\cH_m$, $H_m = I+N_m$, and $\rho$ a state such that for some $n>0$ we have $E_m^{(n)}\coloneqq\Tr[\rho\, H_m^n]<\infty$. Define $\mathcal P_M$ and $\widetilde{\mathcal P}_{ M}$ as in~\eqref{eq:projection1} and~\eqref{eq:projection2}, respectively. Then for any $0\le  \alpha<n$
\begin{align}
& \left\|\rho-\mathcal{P}_M(\rho)\right\|_1^{(\alpha)}\le 2%(1+M)
(M+2)^{-\frac{n-\alpha}{2}}\,E_m^{(n)}, \label{eq:EC_1}\\
& \left\|\rho-\widetilde{\mathcal{P}}_{{M}}(\rho)\right\|^{(\alpha)}_1\le 2%(1+M)
(M+2)^{-\frac{n-\alpha}{2}}\,E_m^{(n)} + \delta_0\big(\eta,M,\tfrac{\alpha}{2},m\big) \,, \label{eq:EC_2}
\end{align}
where $\delta_0$ was introduced in %Lemma~\ref{boundPPtilde} with the assumption that $\eta^2>2M^2$.
Lemma~\ref{double_truncation_lemma}.
\end{prop}

\begin{proof}
We use a duality argument for this proof. For any $O$ such that $\big\|H_m^{-\alpha/2} O H_m^{-\alpha/2}\big\|_\infty \le 1$, we find
\begin{align*}
& \left| \tr\big[O(\rho-\mathcal{P}_M(\rho))\big] \right| \\
&\quad = \left| \Tr \left[H_m^{-\frac{\alpha}{2}}OH_m^{-\frac{\alpha}{2}}H_m^{\frac{\alpha}{2}}(\rho-\mathcal P_M (\rho))H_m^{\frac{\alpha}{2}}\right] \right| \\
&\quad \le \left\|H_m^{\frac{\alpha}{2}}(\rho-\mathcal{P}_M(\rho))H_m^{\frac{\alpha}{2}} \right\|_1 \\
&\quad \le \left\|H_m^{-\frac{n-\alpha}{2}} H_m^{\frac{n}{2}}(I-P_M)\rho \,H_m^{\frac{n}{2}} H_m^{-\frac{n-\alpha}{2}} \right\|_1 \\
&\quad \quad + \left\|H_m^{-\frac{n-\alpha}{2}}H_m^{\frac{n}{2}}P_M\rho(I-P_M) \,H_m^{\frac{n}{2}}H_m^{-\frac{n-\alpha}{2}} \right\|_1 \\
&\quad \textleq{(i)} \left\|H_m^{-\frac{n-\alpha}{2}} (I-P_M)\right\|_\infty \left\| H_m^{\frac{n}{2}} \rho \,H_m^{\frac{n}{2}} \right\|_1 \left\|H_m^{-\frac{n-\alpha}{2}} \right\|_\infty \\
&\quad \quad + \left\|H_m^{-\frac{n-\alpha}{2}}\right\|_\infty \|P_M\|_\infty \left\| H_m^{\frac{n}{2}}\rho \,H_m^{\frac{n}{2}} \right\|_1 \left\| (I-P_M) H_m^{-\frac{n-\alpha}{2}} \right\|_\infty \\
&\quad \textleq{(ii)} 2 %(1+M)
(M+2)^{-\frac{n-\alpha}{2}}E_m^{(n)}\,.
\end{align*}
Here, (i)~comes from a repeated application of H\"older's inequality, and we also observed that $P_M$ and $H_m$ commute; (ii)~is because
\bb
\left\|H_m^{-\frac{n-\alpha}{2}} (I-P_M)\right\|_\infty &= \left\| \sumno_{\mathbf{n}\in \N^m \setminus [M]^m} (1+|\mathbf{n}|)^{-\frac{n-\alpha}{2}} \ketbra{\mathbf{n}}\right\|_\infty \\
&= \max_{\mathbf{n}\in \N^m \setminus [M]^m} (1+|\mathbf{n}|)^{-\frac{n-\alpha}{2}} \\
&= (2+M)^{-\frac{n-\alpha}{2}} .
\ee
This proves~\eqref{eq:EC_1}. Next, %by assumption
we have that
\begin{align*}
&\left| \tr\left[O\,(\mathcal{P}_M-\widetilde{\mathcal{P}}_{{M}})(\rho)\right] \right| \\
&\quad = \left| \Tr\left[H_m^{-\frac{\alpha}{2}}OH_m^{-\frac{\alpha}{2}}H_m^{\frac{\alpha}{2}}(\mathcal P_M (\rho)-\widetilde{\mathcal{P}}_{M}(\rho))H_m^{\frac{\alpha}{2}}\right] \right| \\
&\quad \textleq{(iii)} \left\|H_m^{-\frac{\alpha}{2}}OH_m^{-\frac{\alpha}{2}} \right\|_\infty \left\|H_m^{\frac{\alpha}{2}}\,(\mathcal{P}_M-\widetilde{\mathcal{P}}_{M})(\rho)H_m^{\frac{\alpha}{2}}\right\|_1 \\
&\quad \textleq{(iv)} \delta_0\left(\eta,M,\tfrac{\alpha}{2},m\right) ,
\end{align*}
where (iii)~is again H\"older's inequality, and (iv)~comes from our assumptions together with Lemma~\ref{double_truncation_lemma}. Combining this with~\eqref{eq:EC_1} yields~\eqref{eq:EC_2} and concludes the proof.
%This concludes the proof.
\end{proof}
In order to use the approximation bounds of~\ref{prop:EC}, it remains to estimate how well the empirical average of the shadows
\begin{align}\label{empaverage}
{\hat{\sigma}}^{(N)} \coloneqq \frac{1}{N} \sum_{i=1}^N \hat{\rho}^{(i)}\, ,
\end{align}
where $\hat{\rho}^{(i)}$ are $N$ i.i.d.~copies of the shadow $\hat{\rho}$ constructed in~\eqref{hatrho}, approximates $\mathcal{P}_{M}(\rho)$ and $\widetilde{\mathcal{P}}_{{M}}(\rho)$ depending on the choice of the Gaussian shadow tomography scheme. For this, we introduce, for any set $A$ of $|A|=r$ modes, the operators
\begin{align}
   & {\sigma}^{(N)}_{A}(M)\coloneqq \sum_{\mathbf{n}_1,\mathbf{n}_2 \in \{0,..,M\}^{r}}\,\hat{\sigma}_A^{(N)}(\ketbraa{\mathbf{n}_2}{\mathbf{n}_1})\,\ketbraa{\mathbf{n}_1}{\mathbf{n}_2},\nonumber \\
    & 
 \widetilde{\sigma}_{A}^{(N)}(M)\coloneqq \sum_{\mathbf{n}_1,\mathbf{n}_2 \in \{0,..,M\}^r}  \hat{\sigma}_A^{(N)}(\widetilde  Z_{\mathbf{n}_2\mathbf{n}_1}) \, \ketbraa{\mathbf{n}_1}{\mathbf{n}_2}\,.\label{empiricalshadows}
\end{align}
Since both operators are supported on a finite dimensional subspace, we can resort to the matrix Bernstein inequality~\eqref{matrixBern} in order to prove that, with high probability, $\|{\sigma}_{A}^{(N)}(M)-{\mathcal{P}}_{{M}}(\rho_A)\|_1,\|\widetilde{\sigma}_{A}^{(N)}(M)-\widetilde{\mathcal{P}}_{{M}}(\rho_A)\|_1\le \epsilon$ for $N$ large enough. In the next section, we explain in more detail how we use the matrix Bernstein inequality in the case of both the homodyne and the heterodyne detection strategies.

\section{Homodyne and heterodyne shadow tomography}

\subsection{Local homodyne detection}\label{homodyne}

Let us first consider the scenario in which one performs a homodyne detection along a random direction in phase space. Here, $m$ independent random matrices $S_1,\dots , S_m$ are distributed uniformly (according to the Haar measure) on the intersection $\operatorname{Sp}(2) \cap \operatorname{SO}(2)$ between the symplectic and the orthogonal group. It is useful to note that $\operatorname{Sp}(2) \cap \operatorname{SO}(2) \simeq \operatorname{U}(1)$, where on the right-hand side we have the unitary group of $1\times 1$ matrices, so that the Haar measure on $\operatorname{Sp}(2) \cap \operatorname{SO}(2)$ is essentially that on $\operatorname{U}(1)$. In this simple homodyne case the algorithm for shadow tomography is summarised as follows:
\begin{enumerate}[(1)]
\item A copy of $\rho$ is loaded, and $m$ matrices $S_1,\dots, S_m\in \operatorname{U}_1$ are drawn at random according to the Haar measure. In other words $U_j=R_{\theta_j}$ for some angle $\theta_j\in[-\pi,\pi]$, where
\begin{align*}
R_\theta \coloneqq \begin{pmatrix} \cos\theta & -\sin \theta \\ \sin\theta & \cos\theta  \end{pmatrix} \,.
\end{align*}
\item The rotation $U_{S}\equiv U_{S_1}\otimes \dots\otimes U_{S_m} $ is applied to $\rho$, where $S=S_1\oplus\dots\oplus S_m$, obtaining the state $U_{S}\rho\, U_{S}^\dag$.
\item The output state of (2) is subjected to a homodyne measurement along the position axis, yielding the classical outcome $x=(x_1,\dots,x_m)^\intercal\in\mathbb{R}^{2m}$. %\footnote{Admittedly, not the happiest notation. We should exchange $r\leftrightarrow -r$, so we homodyne along $x$ instead...}.
\item We construct the classical shadow $\hat{\rho}$, which is the final output of this round of the protocol.
\end{enumerate}
After several rounds have been conducted, we can process the classical shadows as we prefer. A typical method would be that of computing the empirical average $\hat{\sigma}^{(N)}$ as defined in~\eqref{empaverage} for the above protocol. We can then use that operator e.g.\ for computing expected values of observables, or else reduced density operators, etc. To model  homodyning in a rigorous way, let us introduce a parameter $s>0$, which we will later take to infinity, and let us set
\bb
T^{\oplus m} = T^{\oplus m}_s \coloneqq \begin{pmatrix} e^{-s}  & 0 \\ 0 & e^{s}  \end{pmatrix}^{\bigoplus m} ,
\ee
These choices of $\mu\sim \operatorname{Haar}(\operatorname{U}_1)$ and $T$ completely determine our shadow tomographic setting. In the limit $s\to\infty$, the measurement will reproduce a homodyne measurement along the position axis of each mode. To make things more concrete, given a threshold $M\in \mathbb{N}$ and $\alpha\ge 0$ we define
\begin{equation}
\begin{aligned}\label{homodyneshadow}
\hat{\rho}_A(M) &\coloneqq \bigotimes_{j\in A}\,\Big(\sumno_{n_1,n_2=0}^M\,\hat{\rho}_j(M)_{n_1,n_2} \ketbraa{n_1}{n_2} \Big) ,\\
\hat{\rho}_j(M)_{n_1,n_2} &\coloneqq \int\!\! dy\ |y|\,  e^{- i (y,0) S_j^{-\intercal}\Omega S_j x_j}\,\chi_{\ket{n_2}\!\bra{n_1}}(S_j(y,0)^\intercal )\,,
\end{aligned}
\end{equation}
By (\ref{eq.charfock}), the absolute value of the characteristic function $\chi_{|n_2(j)\rangle\langle n_1(j)|}(S_j(y,0)^T)$ is upper bounded by
\begin{align*}
  \sqrt{\frac{n_2(j)!}{n_1(j)!}}\, e^{-\frac{\pi}{2}|\omega|^2}(\sqrt{\pi}|\omega|)^{M_j-m_j}\,L_{m_j}^{(M_j-m_j)}(\pi|\omega|^2)\,,
\end{align*}
where $\omega\coloneqq -x'+i\frac{x''}{2\pi}$ with $x=x'\oplus x''=S_j(y,0)^T$, $m_j\coloneqq \min\{n_1(j),n_2(j)\}$, $M_j\coloneqq \max\{n_1(j),n_2(j)\}$, and where the Laguerre polynomials $L_j^{(k)}$ were defined in~\eqref{eq.Laguerre}. Since $\frac{|y|}{2\pi}\le |\omega|\le |y|$, we therefore have that $\left\|H_r^{\frac{\alpha}{2}}\hat{\rho}_A(M)H_r^{\frac{\alpha}{2}}\right\|_\infty$ is almost surely bounded by 
%\begin{widetext}
\bb
&\Sigma^{(\alpha)}_r(M) \\
&\ \coloneqq \Bigg\|\, \Bigg(\!\left((1+|\mathbf{n}_1|)(1+|\mathbf{n}_2|)\right)^{\frac{\alpha}{2}}\prod_{j=1}^r\,\sqrt{\frac{n_2(j)!}{n_1(j)!}} \\
&\qquad \times\int\!\! dy\, |\sqrt{\pi}y|^{1+M_j-m_j}\,\,e^{-\frac{|y|^2}{8\pi}}\left|L_{M_j}^{M_j-m_j}\big(\pi|y|^2\big)\right| \Bigg)_{\mathbf{n}_1, \mathbf{n}_2} \Bigg\|_\infty ,
\ee
%\end{widetext}
{where $\mathbf{n}_1, \mathbf{n}_2 \in\{0,\dots ,M\}^r$,} by using the simple fact that given two matrices $A=\{(a_{ij})\}_{ij}$ and $B=\{(b_{ij})\}_{ij}$ with $|a_{ij}|\le b_{ij}$ for all $i,j$, then $\|A\|_\infty\le \|B\|_{\infty}$. This can be proved as follows. For a vector $v$ with entries $v_i$ and norm $\|v\|_2^2 = \sum_i |v_i|^2$, define the new vector $w$ with entries $w_i\coloneqq |v_i|$, so that $\|w\|^2=1$. Using the triangle inequality, it is straightforward to verify that $\|A v\|^2 \leq \|B w\|^2$; this entails that $\|A\|_\infty = \sup_{\|v\|=1} \|Av\| \leq \sup_{\|w\|=1} \|B w\| = \|B\|_\infty$. {Continuing, as before} we then denote 
\begin{align*}
    \sigma_A^{(N)}(M)\coloneqq\frac{1}{N}\sum_{i=1}^N\,\hat{\rho}^{(i)}_A(M)\,,
\end{align*}
where $\hat{\rho}^{(i)}_A(M)$ are i.i.d.~random matrices of law the one of $\hat{\rho}_A(M)$. With this, we are ready to state our first main result:

\begin{thm}\label{thmhomod}
With the above notation,
 %  and assuming that $\tr(\rho_A N_r^{\textcolor{red}{n}})\le rE$ for all $A$ of size $|A|=r$, we have that for integer $M= \Big\lceil \Big(\frac{32rE}{\epsilon^2}\Big)^{1/n}\Big\rceil$ and
%  \begin{align*}
%  	N\ge \frac{(M+1)^{2r}}{3\epsilon^2}\log\left(\frac{2\big(m(M+1)\big)^r}{\delta}\right)\Big\{24\,\Sigma_r(M)^2+4\Big(\Sigma_r(M)+1\Big)\epsilon\Big\}\,,
%  \end{align*}
% and for any region $A$ of size $|A|\le r$, 
% \begin{align*}
%     \left\|N^{-1}\,\sum_{i=1}^N \widetilde{\sigma}^{(i)}_A(M)-\rho_A\right\|_1\le \epsilon
% \end{align*}
% with probability at least $1-\delta$, where $\widetilde{\sigma}^{(i)}_A(M)$ are i.i.d. random matrices distributed as $\widetilde{\sigma}_A(M)$. Similarly, for \begin{align*}
%     N\ge \frac{(M+1)^{2r}}{3\epsilon^2}\log\left(\frac{2L(M+1)^r}{\delta}\right)\Big\{24\,\Sigma_r(M)^2+4\Big(\Sigma_r(M)+1\Big)\epsilon\Big\}\,,
% \end{align*}
% we have that for any set of $L$ observables $O_j$ on regions $A_j$ of size at most $r$ and with $\|O_j\|_\infty\le 1$, 
% \begin{align*}
%     \max_j \left|\tr\left[ O_j\,N^{-1}\sum_{i=1}^N \left(\widetilde{\sigma}_{A_j}^{(i)}(M)-\rho_{A_j}\right)\right] \right|\le \epsilon
% \end{align*}
% with probability at least $1-\delta$. 
given $0\le \alpha <n$ such that $E^{(n)}_r\coloneqq \max_{|A|\le r}\tr(\rho_A H_r^n)<\infty$ and  
\begin{equation}
\begin{aligned}
N&\ge \frac{(M+1)^{2r}}{3\epsilon^2} \Big\{24\,\Sigma^{(\alpha)}_r(M)^2+4\Big(\Sigma_r^{(\alpha)}(M)+E_r^{(\alpha)}\Big)\epsilon\Big\} \\
&\quad\times \log\left(\frac{2\big[m(M+1)\big]^r}{\delta}\right) ,
\end{aligned}
\end{equation}
where
\begin{equation}
M=\Big\lceil \big(4E_r^{(n)}/\epsilon\big)^{\frac{2}{n-\alpha}}\Big\rceil ,
\end{equation}
we have that for any region $A$ of size $|A|\le r$, it holds that
\begin{equation}
\left\|\sigma_A^{(N)}(M)-\rho_A\right\|^{(\alpha)}_1\le \epsilon
\end{equation}
with probability at least $1-\delta$. Similarly, for 
\begin{equation}
\begin{aligned}
N &\ge \frac{(M+1)^{2r}}{3\epsilon^2} \Big\{24\,\Sigma^{(\alpha)}_r(M)^2+4\Big(\Sigma_r^{(\alpha)}(M)+E_r^{(\alpha)}\Big)\epsilon\Big\} \\
&\quad \times \log\left(\frac{2L(M+1)^r}{\delta}\right)\,,
\end{aligned}
\end{equation}
we have that for any set of $L$ observables $O_j$ on regions $A_j$ of size at most $r$ and with $\left\|H_{r}^{-\frac{\alpha}{2}}O_jH_{r}^{-\frac{\alpha}{2}}\right\|_\infty\le 1$, 
\begin{equation}
\max_j \big|\tr\big[O_j\,(\sigma_{A_j}^{(N)}(M)-\rho_{A_j})\big]\big|\le \epsilon
\end{equation}
with probability at least $1-\delta$.

\end{thm}

\begin{rem}
We essentially recover the same dependence of the number of samples in terms of the logarithm of the number of observables/total number of modes and on the exponential of the size of the regions $A$ as in the qubit setting of~\cite{Huang2020}.
\end{rem}

\begin{proof}
The first part of the proof consists in computing the function $f_{\mu^{\otimes m},T^{\oplus m}}$ defined as in~\eqref{f_mu_T} in the limit $s \to\infty$. It is given by

\begin{align}
 f_{\mu^{\otimes m}, T^{\oplus m}}(u)  &\coloneqq \int \!\! d\mu^{\otimes m}(S_1,\dots , S_m)\, e^{-\frac12\, u^\intercal (T^{\oplus m}S)^{^\intercal} (T^{\oplus m}S)\, u}\nonumber\\ &=\prod_{j=1}^m f_{\mu,T}(u_j)\,,\nonumber 
\end{align}
with $u=(u_1,\dots, u_m)^T\in\mathbb{R}^{2m}$, where for each mode $j$ we have 
\bb
f_{\mu,T}(u_j) &\texteq{(i)} \int_{-\pi}^{+\pi} \frac{d\theta}{2\pi}\, e^{-\frac12 u_j^\intercal R_\theta^\intercal D_s R_\theta u_j} \\
&\texteq{(ii)} \int_{-\pi/2}^{+\pi/2} \frac{d\theta}{\pi}\, e^{-\frac12 u_j^\intercal R_\theta^\intercal D_s R_\theta u_j} \\
&\texteq{(iii)} \int_{-\pi/2}^{+\pi/2} \frac{d\theta}{\pi}\, e^{-\frac{\|u_j\|^2}{2}\, e_1^\intercal R_\theta^\intercal D_s R_\theta e_1} \\
&= \int_{-\pi/2}^{+\pi/2} \frac{d\theta}{\pi}\, e^{-\frac{\|u_j\|^2}{2}\, \left( e^{-2s} \cos^2\theta + e^{2s} \sin^2\theta \right)} \\
&\texteq{(iv)} e^{-\frac{\Vert u_j \Vert^2 \cosh(2s) }{2}}  I_0\Big(-\tfrac{\Vert u_j \Vert^2 \sinh(2s) }{2}\Big)\\
&= e^{-\frac{\Vert u_j \Vert^2 \cosh(2s) }{2}}  I_0\Big(\tfrac{\Vert u_j \Vert^2 \sinh(2s) }{2}\Big).
\ee
% By the asymptotic expansion of the modified Bessel function
% \[ I_0\Big(\tfrac{\Vert u_j \Vert^2 \sinh(2s) }{2}\Big)= \frac{e^{\tfrac{\Vert u_j \Vert^2 \sinh(2s) }{2}}}{\sqrt{\sinh(2s)\pi  }\| u_j \|}\Big(1+\mathcal O((\| u_j \|^{2}\sinh(2s))^{-1}) \Big),\]
% we see that 
% \[ f_{\mu,T}(u_j)= \frac{e^{-\tfrac{\Vert u_j \Vert^2 e^{-2s} }{2}}}{\sqrt{\sinh(2s)\pi  }\Vert u_j \Vert}\Big(1+\mathcal O((\Vert u_j \Vert^{2}\sinh(2s))^{-1}) \Big).\]
% In particular, we have
%  \begin{align*}
%      \frac{(e^{s}\|u_j\|)}{(2\pi)^{\frac{1}{2}}}\,f_{\mu,T}(u_j)\tends{}{s\to \infty} 1\,,
%  \end{align*}
Here, in~(i) we introduced the matrix $D_s=T_s^2$; in~(ii) we noted that the integrand is invariant under rotations of $\pi$, in~(iii) we observed that the integral is invariant under rotations of $u_j$, and thus chose to compute it for $u_j = \|u_j\| e_1 \coloneqq \|u_j\| (1,0)^\intercal$; and in~(iv) we recognised the integral representation for the modified Bessel function $$I_0(x)\coloneqq\frac{1}{\pi}\,\int_{0}^\pi e^{\pm x\cos\theta}d\theta\,.$$ Plugging this into~\eqref{chi_sigma}, we have found that
\begin{align*}
\chi_{\hat{\rho}}(u) &= \frac{e^{-\frac14 \left\|T^{\oplus m}Su\right\|^2 - i u^\intercal \Omega S x}}{f_{\mu^{\otimes m},T^{\oplus m}}(u)}\, \\
&=\prod_{j=1}^m\,\frac{e^{\frac{\|u_j\|^2\,\operatorname{cosh}(2s)}{2}}}{I_0\Big(\frac{\|u_j\|^2\operatorname{sinh}(2s)}{2}\Big)}e^{-\frac14 \left\|TS_ju_j\right\|^2 - i u_j^\intercal \Omega S_j^{-1} x_j}\\
&\equiv \prod_{j=1}^m\,\chi_{\hat{\rho}_j}(u_j)\,.
\end{align*}
In other words, $\chi_{\hat{\rho}}$ is formally the characteristic function of a tensor product of forms $\hat{\rho}=\hat{\rho}_1\otimes \dots \otimes \hat{\rho}_m$. 
Moreover, by the asymptotic expansion of the modified Bessel function
\begin{align*}
 &I_0\Big(\tfrac{\Vert u_j \Vert^2 \sinh(2s) }{2}\Big) \nonumber\\
 &= \frac{e^{\tfrac{\Vert u_j \Vert^2 \sinh(2s) }{2}}}{\sqrt{\sinh(2s)\pi  }\Vert u_j \Vert}\Big(1+\mathcal O((\Vert u_j \Vert^{2}\sinh(2s))^{-1}) \Big),
 \end{align*}
we see that 
\begin{align*}
&\chi_{\hat{\rho}_j}(u_j) = \frac{\sqrt{2\sinh(2s)\pi}\Vert u_j \Vert\,e^{\frac{\Vert u_j \Vert^2 e^{-2s}}{2}}e^{-\frac14 \|T_s S_j u_j\|^2 - i u_j^\intercal \Omega S_j^{-1} x_j} }{1+\mathcal O((\Vert u_j \Vert^{2}\sinh(2s))^{-1})} 
\end{align*}
which shows that $\chi_{\hat{\rho}}(u)$ can be integrated against any element of the Fock basis as soon as $s>0$. Taking the distributional limit, we find
\begin{align*}
\chi_{\hat{\rho}_j}(u_j)\tends{}{s\to\infty}  2\pi \|u_j\|\, \delta \left( (S_j u_j)_2 \right) e^{- i u_j^\intercal \Omega S_j^{-1} x_j}\,.
\end{align*}
Next, we fix a region $A$ of size $|A|=r$ and consider for a fixed squeezing $s$
\begin{align*}
    \hat{\rho}_{A}(s,M)\coloneqq \bigotimes_{j\in A}\,\Big(\sum_{n_1,n_2=1}^M\, \hat{\rho}_j(\ketbraa{n_2}{n_1}|)\,\ketbraa{n_1}{n_2}\Big)\,.
\end{align*}
We have from Lemma~\ref{partialparsevalaverage} that
\begin{align*}
    \mathbb{E}\big[\hat{\rho}_{A}(s,M)\big]=\mathcal{P}_M(\rho_A)
\end{align*}
which holds for any squeezing parameter $s$. Therefore, we can take the limit $s\to\infty$ above. We find
%\begin{widetext}
\bb
&\hat{\rho}_{A}(s,M) = \bigotimes_{j\in A}\!\bigg(\sum_{n_1,n_2=1}^M\!\!\ketbraa{n_1}{n_2}\int\!\!\chi_{\hat{\rho}_j}(u_j)\,\chi_{\ket{n_2}\!\bra{n_1}}(u_j)\,\frac{d^2 u_j}{(2\pi)^{r}}\bigg)
\ee
{so that}
\bb
\hat{\rho}_{A}(s,M) &\tends{}{s\to\infty} \hat{\rho}_A(M) \coloneqq \bigotimes_{j\in A} \hat{\rho}_j(M) ,\\
\hat{\rho}_j(M) &\coloneqq \sum_{n_1,n_2=1}^M\!\!\ketbraa{n_1}{n_2}\\
&\qquad \times \int\!\! |y|\, e^{- i (y,0) S_j^{-\intercal}\Omega S_j x_j}\,\chi_{\ket{n_2}\!\bra{n_1}}(S_j(y,0)^\intercal )\,d y ,
\label{sigmaANM}
\ee
%\end{widetext}
with $\mathbb{E}[\hat{\rho}_A(M)]=\mathcal{P}_M(\rho_A)$. Then, we have seen below Equation~\eqref{homodyneshadow} that, almost surely, $\|H_r^{\frac{\alpha}{2}}\hat{\rho}_A(M)H_r^{\frac{\alpha}{2}}\|_\infty\le 
     \Sigma^{(\alpha)}_r(M)$. Moreover, we have that $\|H_r^{\frac{\alpha}{2}}\mathcal{P}_M(\rho_A)H_r^{\frac{\alpha}{2}}\|_\infty\le E_r^{(\alpha)}$, therefore almost surely $\|H_r^{\frac{\alpha}{2}}(\hat{\rho}_A(M)-\mathcal{P}_M(\rho_A))H_r^{\frac{\alpha}{2}}\|_\infty\le E_r^{(\alpha)}+\Sigma^{(\alpha)}_r(M)$. Finally, $\|\mathbb{E}[(H_r^{\frac{\alpha}{2}}\hat{\rho}_A(M)H_r^{\frac{\alpha}{2}})^2]\|_\infty\le\Sigma_r^{(\alpha)}(M)^2$. We can hence use the equivalence of matrix norms together with the matrix Bernstein inequality~\eqref{matrixBern} to get
\begin{align}\label{matrixBernhomodyne}
    & \mathbb{P}\Big(\big\|\sigma_A^{(N)}(M)-\mathcal{P}_M(\rho_A)\big\|^{(\alpha)}_1\ge {\epsilon}\Big)\nonumber\\
    &\le\mathbb{P}\Big(\big\|H_r^{\frac{\alpha}{2}}\big(\sigma_A^{(N)}(M)-\mathcal{P}_M(\rho_A)\big)H_r^{\frac{\alpha}{2}}\big\|_\infty\ge {\epsilon}(M+1)^{-r}\Big) \nonumber\\
    &\le 2(M+1)^r\,\operatorname{exp}\left(-\frac{3N\epsilon^2(M+1)^{-2r}}{6\Sigma_r^{(\alpha)}(M)^2+2(\Sigma_r^{(\alpha)}(M)+E_r^{(\alpha)})\epsilon}\right).
\end{align}
 Therefore, by a union bound, we get that 
\begin{align}\label{unionbound}
    & \mathbb{P}\Big(\exists A, \,|A|\le r :\, \big\|\sigma_A^{(N)}(M)-\mathcal{P}_M(\rho_A)\big\|^{(\alpha)}_1\ge {\epsilon}\Big)\nonumber\\
    & \le 2\big[m(M+1)\big]^r\nonumber\\
    & \quad \times \operatorname{exp}\left(-\frac{3N\epsilon^2(M+1)^{-2r}}{6\Sigma_r^{(\alpha)}(M)^2+2(\Sigma_r^{(\alpha)}(M)+E_r^{(\alpha)})\epsilon}\right)\,.
\end{align}
Next, we use Proposition~\ref{prop:EC} and choose $M=\Big\lceil \big(4E_r^{(n)}/\epsilon\big)^{\frac{2}{n-\alpha}}\Big\rceil$, so that $\|\mathcal{P}_M (\rho_A)-\rho_A\|_1^{(\alpha)}\le \epsilon/2$ and
\begin{align*}
 & \mathbb{P}\Big(\exists A, \,|A|\le r :\,\big\|\sigma_A^{(N)}(M)-\rho_A\big\|^{(\alpha)}_1\ge {\epsilon}\Big)\nonumber\\
 &\le \mathbb{P}\Big(\exists A, \,|A|\le r :\,\big\|\sigma_A^{(N)}(M)-\mathcal{P}_M(\rho_A)\big\|^{(\alpha)}_1\ge \frac{\epsilon}{2}\Big) \nonumber\\
  &\le 2\big[m\big(M+1\big)\big]^r\nonumber\\
  & \quad \times \operatorname{exp}\left(-\frac{3N\epsilon^2(M+1)^{-2r}}{24\Sigma^{(\alpha)}_r(M)^2+4(\Sigma^{(\alpha)}_r(M)+E_r^{(\alpha)})\epsilon}\right)\,.
\end{align*}
Therefore, choosing $N$ as in the statement of the theorem, we obtain that the probability that on any subset $A$ of at most $r$ modes $\|\sigma_A^{(N)}(M)-\rho_A\big\|^{(\alpha)}_1\ge {\epsilon}$ is at least $1-\delta$.

The results for a fixed number of $L$ observables $O_1,\dots, O_L$ supported on regions $A_1,\dots, A_L$ of size at most $r$ follow after replacing the above union bounds over regions $A$ by a union bound over the observables $O_j$:
\begin{align*}
  &  \mathbb{P}\Big(\exists j: \, |\tr(O_j\sigma_{A_j}^{(N)}(M))-\tr(O_j\,\mathcal{P}_M(\rho_{A_j}))|\ge {\epsilon}\Big)\nonumber\\
  & \le 2L\,\big[(M+1)\big]^r\, \operatorname{exp}\left(-\frac{3N\epsilon^2(M+1)^{-2r}}{6\Sigma_r^{(\alpha)}(M)^2+2(\Sigma_r^{(\alpha)}(M)+E_r^{(\alpha)})\epsilon}\right)
\end{align*}
and the result follows after replacing $m^r$ by $L$ in the estimate for $N$. 

\end{proof}

\subsection{Local heterodyne detection}

Even simpler than homodyne detection, the simplest Gaussian shadow tomographic setting is that where the measurement employed is a heterodyne detection
\bb
\left\{ \frac{1}{(2\pi)^{m/2}} \ketbra{x}\right\}_{x\in \R^{2m}}\, ,
\label{heterodyne}
\ee
i.e.\ $T=I$, and all unitaries employed are passive, i.e.\ such that $\big[ \gu{S}, \frac12 R^\intercal R\big] =0$. Naturally, this is the same as requiring that $S$ be (not only symplectic but also) orthogonal. Since passive unitaries send coherent states to coherent states, amounting to a rotation in that space, the effective measurement being carried out on $\rho$ is the same irrespectively of $S$ --- the only thing changing is that the outcome is rotated. Mathematically, this means that the probability distribution of the random variable $\gud{S} \D(x) \psi \D(-x) \gu{S}$, where $x$ is the outcome of the heterodyne detection~\eqref{heterodyne} on $\gu{S}\psi \gud{S}$, is the same as that of the random variable $\D(y) \psi \D(-y)$, where $y$ is the outcome of the heterodyne on $\psi$. This means that in the shadow tomography protocol we can skip the unitary operation altogether without losing any data. In what follows we will therefore set $S=I$ (deterministically) without loss of generality. 

With the above simplifications, one can see that $f_{\mu,T}(u) = e^{-\frac12 \|u\|^2}$, so that $\MM=\mathcal{N}_1$ (cf.~\eqref{noise_chi}). Therefore, the classical shadow will have improper characteristic function
\begin{align*}
\chi_{\hat{\rho}}(u) = e^{  \frac14 \|u\|^2-iu^\intercal \Omega x}\,.
\end{align*}
Once again, the characteristic function tensorises: given $u=(u_1,\dots, u_m)^\intercal$ and $x=(x_1,\dots, x_m)^\intercal$:
\begin{align}
    &\chi_{\hat{\rho}}(u)=\prod_{j=1}^m\chi_{\hat{\rho}_j}(u_j)\end{align}
    where
    \begin{align} \chi_{\hat{\rho}_j}(u_j)\coloneqq  e^{  \frac14 \|u_j\|^2-iu_j^\intercal \Omega x_j}\,.
\end{align}
Now, if we want to use the classical shadow to compute expectation values, we can formally use Plancherel's relation 
\bb
\hat{\rho}( O) &= \int \frac{d^{2m}u}{(2\pi)^m}\, \chi_{\hat{\rho}}(u)\, \chi_O(u)\\
&=  \int \frac{d^{2m}u}{(2\pi)^m}\, e^{+\frac14 \|u\|^2 -iu^\intercal \Omega x }\, \chi_O(u) \, .
\label{hetero_expectation_value_shadow}
\ee
In order for this to make sense, we should make sure that not only $O$ is trace class (instead of bounded), but also that $\chi_O(u)$ decays sufficiently rapidly, for instance like $\sim e^{-\frac{\lambda}{4}\|u\|^2}$, with $\lambda>1$. This decay is too fast --- but \emph{barely} too fast --- to be useful in practice. For instance, if $O$ has a finite expansion in the Fock basis then $\chi_O(u) \sim p(u) e^{-\frac{1}{4}\|u\|^2}$ as $\|u\|\to \infty$, where $p(u)$ is some polynomial of the entries of $u$. We get rid of the diverging Gaussian in~\eqref{hetero_expectation_value_shadow}, but not of the diverging {polynomial.} In order to take care of this issue, we make use of the approximations $\widetilde{Z}_{\mathbf{n}_1\mathbf{n}_2}$ of the Schwartz operators $\ketbraa{\mathbf{n}_1}{\mathbf{n}_2}$ as well as of the auxiliary map $\widetilde{\mathcal{P}}_M$ and the corresponding matrices $\widetilde{\sigma}_A^{(N)}(M)$ introduced in~\eqref{eq:projection2}, resp. in~\eqref{empiricalshadows}. We consider the matrices 
\begin{align*}
    &\hat{\rho}_A(M)\coloneqq  \sum_{\mathbf{n}_1,\mathbf{n}_2\in\{0,...,M\}^{|A|}}\, \hat{\rho}(\widetilde{Z}_{\mathbf{n}_2\mathbf{n}_1})\,\ketbraa{\mathbf{n}_1}{\mathbf{n}_2},\quad \text{where}\\
    &\hat{\rho}(\widetilde{Z}_{\mathbf{n}_2\mathbf{n}_1})\coloneqq \int_{\mathbb{R}^{2r}} \widetilde{\chi}_{\mathbf{n}_2\mathbf{n}_1}(u) \,e^{  \frac14 \|u\|^2-iu^\intercal \Omega x}\,\frac{d^{2r}u}{(2\pi)^r}\,
\end{align*}
where $\widetilde{\chi}_{\mathbf{n}_2\mathbf{n}_1}$ is defined in~\eqref{tildechin1n2}. From Lemma~\ref{partialparsevalaverage}, $\mathbb{E}\big[\hat{\rho}_{A}(M)\big]=\widetilde{\mathcal{P}}_M(\rho_A)$. For $\alpha\ge 0$, $R>0$ and $r,M\in\mathbb{N}$, we introduce the matrix norm 
\begin{align*}
    &\widetilde{\Sigma}^{(\alpha)}_r(M,R)\nonumber \coloneqq \left\| K_r(M,R) \right\|_\infty\, ,
\end{align*}
where the entries of the $(M+1)^r\times (M+1)^r$ matrix $K_r(M,R)$ are defined by
\begin{align*}
    &K_r(M,R)_{\mathbf{n_1}, \mathbf{n_2}} \\
    &\ \, \coloneqq \left((1\!+\!|\mathbf{n}_1|)(1\!+\!|\mathbf{n}_2|)\right)^{\frac{\alpha}{2}} \int_{|u|\le R}\!\! |\widetilde{\chi}_{\mathbf{n}_2\mathbf{n}_1}(u)|\,e^{\frac{1}{4}\|u\|^2} \frac{d^{2r}u}{(2\pi)^r}\, ,
\end{align*}
and $\mathbf{n_1}, \mathbf{n_2} \in \{0,\ldots, M\}^r$.
%     \left\lVert \left( \right)_{\mathbf n_1, \mathbf n_2 \in\{0,\dots ,M\}^r} \right\rVert_\infty
With these definitions, we can write the inequality
\bb
\left\|H_r^{\frac{\alpha}{2}} \hat{\rho}_A(M)\, H_r^{\frac{\alpha}{2}}\right\|_\infty\le \widetilde{\Sigma}^{(\alpha)}_r(M,R)\, .
\ee
As before, we then denote 
\begin{align*}
    \sigma_A^{(N)}(M)\coloneqq\frac{1}{N}\sum_{i=1}^N\,\hat{\rho}^{(i)}_A(M)\,,
\end{align*}
where $\hat{\rho}^{(i)}_A(M)$ are i.i.d.~random matrices of law the one of $\hat{\rho}_A(M)$. Our second main result is stated below: 
     
\begin{thm}
With the notation introduced above, given $0\le \alpha <n$, $r\in\mathbb{N}$ such that $E^{(n)}_r\coloneqq \max_{|A|\le r}\tr(\rho_A H_r^n)<\infty$, and $0<\eta<R$, set
\bb
M\!\coloneqq \min\!\Big\{\!&M'\!\!\in\!\mathbb{N}:\eta^2 > 2M'^2\\
&2(1\!+\!M')^{\frac{\alpha-n}{2}} \!E_r^{(n)} + \delta_0\big(\eta,M'\!\!,\tfrac{\alpha}{2},r\big)\le \tfrac{\epsilon}{2}\Big\}\,,
\ee
with $\delta_0$ defined as in Proposition~\ref{prop:EC}. Then for any

\bb
N\ge\inf_{0<\eta<R} &\frac{(M+1)^{2r}}{3\epsilon^2}\log\left(\frac{2\big[m(M+1)\big]^r}{\delta}\right)\\
&\Big[24\,\widetilde{\Sigma}^{(\alpha)}_r(M\!,R)^2\\
&\ +4\Big(\widetilde{\Sigma}_r^{(\alpha)}(M\!,R)+E_r^{(\alpha)}+\delta_0(\eta,M,\tfrac{\alpha}{2},r)\Big)\epsilon\Big]\,,
\ee

we have that for any region $A$ of size $|A|\le r$, it holds that
\bb
\left\|\sigma_A^{(N)}(M)-\rho_A\right\|^{(\alpha)}_1\le \epsilon
\ee
with probability at least $1-\delta$. Similarly, for
\bb
N\ge \inf_{0<\eta<R} &\frac{(M+1)^{2r}}{3\epsilon^2}\log\left(\frac{2L(M+1)^r}{\delta}\right) \\
&\Big[24\,\widetilde{\Sigma}^{(\alpha)}_r(M\!,R)^2\\
&\ +4\Big(\widetilde{\Sigma}_r^{(\alpha)}(M\!,R)+E_r^{(\alpha)}+\delta_0(\eta,M,\tfrac{\alpha}{2},r)\Big)\epsilon\Big]\,,
\ee

we have that for any set of $L$ observables $O_j$ on regions $A_j$ of size at most $r$ and with $\left\|H_{r}^{-\frac{\alpha}{2}}O_jH_{r}^{-\frac{\alpha}{2}}\right\|_\infty\le 1$, it holds that
\bb
\max_j \big|\tr\big[O_j\,(\sigma_{A_j}^{(N)}(M)-\rho_{A_j})\big]\big|\le \epsilon
\ee
with probability at least $1-\delta$.
\end{thm}

% {[Ideally, a thm statement should be understandable by somebody who's just consulting the paper and hasn't read anything else in it. Maybe we want to make the above statement more self-contained and simpler? Can we express $N$ as a simpler function of $\epsilon$?]}

\begin{proof}

By construction, we have almost surely that $\|H_r^{\frac{\alpha}{2}}\hat{\rho}_A(M)H_r^{\frac{\alpha}{2}}\|_\infty\le 
     \widetilde{\Sigma}^{(\alpha)}_r(M,R)$. Moreover, we have that
     \begin{align*}
     & \|H_r^{\frac{\alpha}{2}}\widetilde{\mathcal{P}}_M(\rho_A)H_r^{\frac{\alpha}{2}}\|_\infty\nonumber\\
     & \le \|\widetilde{\mathcal{P}}_M(\rho_A)\|_1^{(\alpha)}\\
     & \le \|(\widetilde{\mathcal{P}}_M-\mathcal{P}_M)(\rho_A)\|_1^{(\alpha)}+ \|\mathcal{P}_M(\rho_A)\|_1^{(\alpha)}\\
     & \le  \delta_0(\eta,M,\frac{\alpha}{2},r)+ E_r^{(\alpha)}
     \end{align*}
    for $\eta^2> 2M^2$, by~Lemma~\ref{double_truncation_lemma}. Therefore almost surely $\|H_r^{\frac{\alpha}{2}}(\hat{\rho}_A(M)-\widetilde{\mathcal{P}}_M(\rho_A))H_r^{\frac{\alpha}{2}}\|_\infty\le E_r^{(\alpha)}+\delta_0(\eta,M,\frac{\alpha}{2},r)+\widetilde{\Sigma}^{(\alpha)}_r(M,R)$. Finally, $\|\mathbb{E}[(H_r^{\frac{\alpha}{2}}\hat{\rho}_A(M)H_r^{\frac{\alpha}{2}})^2]\|_\infty\le\widetilde{\Sigma}_r^{(\alpha)}(M,R)^2$. We hence use the matrix Bernstein inequality~\eqref{matrixBern} to get in terms of a constant $C=C(N,\epsilon,R,\alpha,r)$ with \[C\coloneqq \frac{3N\epsilon^2(M+1)^{-2r}}{6\widetilde{\Sigma}_r^{(\alpha)}(M,R)^2+2(\widetilde{\Sigma}_r^{(\alpha)}(M,R)+\delta_0(\eta,M,\frac{\alpha}{2},r)+E_r^{(\alpha)})\epsilon}\]
    the estimate
\begin{align}\label{matrixBernhetero}
    &\mathbb{P}\Big(\big\|\widetilde{\sigma}_A^{(N)}(M)-\widetilde{\mathcal{P}}_M(\rho_A)\big\|^{(\alpha)}_1\ge {\epsilon}\Big)\\
    &\le 2(M+1)^r\,\operatorname{exp}\left(-C(N,\epsilon,R,\alpha,r)\right)\,.\nonumber
\end{align}
Therefore, by a union bound, we get that 
\begin{align*}
   & \mathbb{P}\Big(\exists A, \,|A|\le r :\, \big\|\widetilde{\sigma}_A^{(N)}(M)-\widetilde{\mathcal{P}}_M(\rho_A)\big\|^{(\alpha)}_1\ge {\epsilon}\Big)\\
   &\le 2\big[m(M+1)\big]^r\,\operatorname{exp}\left(-C(N,\epsilon,R,\alpha,r)\right)\,.
\end{align*}
Next, we use Proposition~\ref{prop:EC} and choose $M=\min\{M'\in\mathbb{N}:\,2(1+M')^{\frac{\alpha-n}{2}}\,E_r^{(n)}+\delta_0\big(\eta,M',\frac{\alpha}{2},r\big)\le \frac{\epsilon}{2}\}$, so that $\|\widetilde{\mathcal{P}}_M (\rho_A)-\rho_A\|_1^{(\alpha)}\le \epsilon/2$, and
\begin{align*}
& \mathbb{P}\Big(\exists A, \,|A|\le r :\,\big\|\widetilde{\sigma}_A^{(N)}(M)-\rho_A\big\|^{(\alpha)}_1\ge {\epsilon}\Big) \nonumber\\
&\le \mathbb{P}\Big(\exists A, \,|A|\le r :\,\big\|\widetilde{\sigma}_A^{(N)}(M)-\widetilde{\mathcal{P}}_M(\rho_A)\big\|^{(\alpha)}_1\ge \frac{\epsilon}{2}\Big) \,.
\end{align*}
Therefore, choosing $N$ as in the statement of the theorem, we obtain that the probability that on any subset $A$ of at most $r$ modes $\|\widetilde{\sigma}_A^{(N)}(M)-\rho_A\big\|^{(\alpha)}_1\ge {\epsilon}$ is at least $1-\delta$.

The results for a fixed number of $L$ observables $O_1,\dots, O_L$ supported on regions $A_1,\dots, A_L$ of size at most $r$ follow after replacing the above union bounds over regions $A$ by a union bound over the observables $O_j$:
\begin{align*}
&    \mathbb{P}\Big(\exists j: \, |\tr(O_j\widetilde{\sigma}_{A_j}^{(N)}(M))-\tr(O_j\,\widetilde{\mathcal{P}}_M(\rho_{A_j}))|\ge {\epsilon}\Big)\\
    &\le 2L\,\big[(M+1)\big]^r\,\operatorname{exp}\left(-C(N,\epsilon,R,\alpha,r)\right)
\end{align*}
and the result follows after replacing $m^r$ by $L$ in the estimate for $N$.

\end{proof}

\subsection{Comparison to related work}
\label{compare}
In a concurrent and independent work~\cite{michael2016new}, the authors also developed a shadow tomography protocol for bosonic systems. Their framework is equivalent to ours in the homodyne and heterodyne settings. Let us consider the former in the one-mode setting: the shadows in~\cite{michael2016new} are constructed in the Fock basis by taking $\check{\rho}_{mn}\coloneqq e^{i(m-n)\theta}f_{mn}(x)$, after making a random rotation $\theta$ and obtaining $x$ from a measurement in the position axis. The functions $f_{mn}=f_{nm}$ are the so-called pattern functions~\cite{lvovsky2009continuous,PhysRevA.50.4298} and are defined using Fock state wavefunctions $\psi_m$ ($m^{th}$ energy eigenstate of the harmonic
oscillator) and $\phi_n$ ($n^{th}$ non-normalizable solution of
the Schr\"{o}dinger equation of a harmonic oscillator) as
\begin{align*}
f_{mn}(x)\coloneqq \frac{\partial}{\partial x} (\psi_m(x)\phi_n(x)),\qquad n\ge m\,.
\end{align*}
Then, they construct the following estimator from $N$ copies of the unknown state:
\begin{align*}
\check{\sigma}^{(N)}_{mn}\coloneqq \frac{1}{N}\sum_{i=1}^N\,\check{\rho}^{(i)}_{mn}=\frac{1}{N}\sum_{i=1}^N e^{i(m-n)\theta_i}f_{mn}(x_{\theta_i})\,,
\end{align*}
where the angles $\theta_i$ are picked uniformly at random, and $x_{\theta}=x'_\theta\oplus x''_\theta=\sqrt{2\pi}x'\oplus \frac{1}{\sqrt{2\pi}}x''$ in the notations of Section~\ref{homodyne}. It turns out that these pattern functions can be equivalently defined in terms of their Fourier transforms (eq.~(31) of~\cite{PhysRevA.61.063819}, see also eqs.~(20)-(21) of~\cite{aubry2008state}): for $n\ge m$:
\begin{align*}
\widetilde{f}_{mn}(t)&\coloneqq \int e^{-itx}\,f_{mn}(x)\,dx\\
&=\pi (-i)^{n-m}\,\sqrt{\frac{2^{m-n}m!}{n!}}\,|t|\,t^{n-m}e^{-\frac{t^2}{4}}\,L_{m}^{(n-m)}\Big(\frac{t^2}{2}\Big)\,.
\end{align*}
Therefore $|\widetilde{f}_{mn}(x_\theta)|$ corresponds to the coordinate $\Sigma_1^{(0)}(M)_{mn}$ for $m,n\le M$. As observed in~\cite{Gandhari2022}, the pattern functions have been already studied for tomography purposes (see e.g.~Lemma 7.1 in~\cite{Artiles2005}). In particular, 
\begin{align*}
    \sum_{0\le j\le k\le M}\|f_{kj}\|^2_\infty =\mathcal{O}(M^{7/3})\,.
\end{align*}
These bounds can be directly used to get control over $\|\Sigma_1^{(0)}(M)\|_\infty$. That way, we recover Theorem III.1 of~\cite{michael2016new}. Slightly better bounds can be achieved using estimates in Lemma 4 of~\cite{aubry2008state}. 
% For instance:
%  \begin{align*}
%   \|\Sigma^{(0)}_1(M)\|_\infty =\mathcal{O}( M^{\frac{11}{4}})\,.
%  \end{align*}
% Moreover, instead of bounding $\|\mathbb{E}[\hat{\rho}_A(M)^2]\|_\infty$ by $\Sigma_r^{(0)}(M)^2$ as we did in the proof of Theorem~\ref{thmhomod}, we can get the following refined bound:
% \begin{align*}
%     \|\mathbb{E}[\hat{\rho}_A(M)^2]\|_\infty
% \end{align*}
%  With these bounds at hand, we arrive at the following
%  \begin{cor}
%  blabla
%  \end{cor}
It is also worth observing that, from the proof of Theorem~\ref{thmhomod}, one can extend these results to more realistic models of homodyning with finitely squeezed resources.

\section{Learning non-linear functionals of the states}

So far, we considered properties of the quantum system which could be related to local linear functionals of the unknown state. In finite dimensions, a simple trick permits the estimation of non-linear functionals, e.g.~the entropy of entanglement~\cite{Huang2020,Huang2021}. Here, we show that the technique developed in these works combined with recent energy-constrained continuity bounds provide us with a similar extension. For sake of conciseness, we will only consider the entropy of a reduced CV state over $r\le m$ modes of the $m$-mode state:
\begin{align*}
    S(\rho_A)\coloneqq -\tr(\rho_A\ln\rho_A)\,,
\end{align*}
where we denote by $A$ the set of those modes, so that $|A|=r$. We also assume that the unknown state $\rho$ has locally finite energy: $\tr((I+N^{(j)})\rho_j)\coloneqq E<\infty$, where $N^{(j)}$ corresponds to the number operator on mode $j$, and where $\rho_j$ is the reduced state on mode $j\in A$. Here, we also restrict ourselves to the shadows constructed in our local homodyne detection scheme of~\eqref{homodyne}, although similar conclusions can be drawn from shadows arising from a local heterodyne detection and states with higher moment constraints. Given the shadow 
$\sigma_A^{(N)}(M)$ arising from the homodyne scheme of~\eqref{homodyne}, and $d_p\in\mathbb{N}$, we denote the matrix polynomial
\begin{align*}
H^{(d_p)}(\sigma_A^{(N)}(M))&\coloneqq -\tr(\sigma_A^{(N)}(M)-P_{M})\\
&\qquad \qquad +\sum_{k=2}^{d_p}\frac{\tr\big[(P_M-\sigma_A^{(N)}(M))^k\big]}{k(k-1)}\,.
\end{align*}
% In what follows, we denote $f(M)\coloneqq h\Big(\frac{2{rE}}{(1+M)^{\frac{1}{2}}-2rE}\Big)+rEh\Big(\frac{2}{(1+M)^{\frac{1}{2}}-2rE}\Big)+r\ln(M+1)\frac{4{rE}}{(1+M)^{\frac{1}{2}}-2rE}$, where $h(.)$ denotes the binary entropy:
% \begin{align*}
%     M_0\coloneqq \max\left\{\left\lceil  \left(\frac{4rE}{\frac{\epsilon^2}{12(M+1)^re}2^{-4(M+1)^r/\epsilon}}\right)^2\right\rceil,\,\lceil{f^{-1}(\epsilon/3)\rceil} \right\}\,.
% \end{align*}
\begin{thm}
With the notation of the previous paragraph, we have that for any $\epsilon>0$ and $M\in\mathbb{N}$,
\begin{align*}
&\mathbb{P}\Big(\exists A,\,|A|\le r:\, |S(\rho_A)-H^{\big(\big\lceil \frac{3(M+1)^r}{\epsilon}\big\rceil\big)}(\sigma_A^{(N)}(M))|\ge \epsilon\Big)\\
&\le 2\big[m(M+1)\big]^r\,\operatorname{exp}\left(-\frac{3N\epsilon'^2(M+1)^{-2r}}{6\Sigma_r^{(0)}(M)^2+2(\Sigma_r^{(0)}(M)+1)\epsilon'}\right)\,,
\end{align*}
where $\epsilon'=\frac{\epsilon^2}{12(M+1)^re}2^{-4(M+1)^r/\epsilon}$.

% In the above framework, and assuming that $\tr(\rho_A N_r)\le rE$ for all $A$,
% Now, we have from the first part of~\eqref{thmhomodynetomo}, choosing $M= \big\lceil\frac{16rE}{\epsilon^2}\big\rceil$ and $N\ge \frac{(M+1)^{2r}}{3\epsilon^2}\log\Big(\frac{2\big(m(M+1)\big)^r}{\delta}\Big)\Big\{24\,\Sigma_r(M)^2+4\Big(\Sigma_r(M)+1\Big)\epsilon\Big\}$,

% for $\epsilon<2rE/(1+2rE)$
\end{thm}

\begin{proof}
By~\cite[Theorem 3]{becker2021classical}, we have that for all $\gamma\in [0,rE/(1+rE)]$, and any state $\rho_A'$ on subsystem $A$ with $\|\rho_{A}-\rho_A'\|_1\le 2\gamma$,
\begin{align}\label{ECcontinuity}
    |S(\rho_A)-S(\rho_A')|\le h(\gamma)+rE\,h\big(\gamma/(rE)\big)\,,
\end{align}
where $h(.)$ denotes the binary entropy. Next, we pick $\rho_A'\coloneqq \mathcal{P}_M(\rho_A)/\tr(\mathcal{P}_M(\rho_A))$. By Proposition~\ref{prop:EC}, we have that, assuming $1+M>4r^2E^2$,
\begin{align*}
   &  \|\rho_A-\rho_A'\|_1=\frac{\|\tr(\mathcal{P}_M(\rho_A))\rho_A-\mathcal{P}_M(\rho_A)\|_1}{\tr(\mathcal{P}_M(\rho_A))}\,\\
    & \le \frac{|\tr(\mathcal{P}_M(\rho_A))-1|+\|\rho_A-\mathcal{P}_M(\rho_A)\|_1}{1-|\tr(\mathcal{P}_M(\rho_A))-1|}\nonumber\\
    & \le \frac{2\|\rho_A-\mathcal{P}_M(\rho_A)\|_1}{1-\|\mathcal{P}_M(\rho_A)-\rho_A\|_1}\,\\
    & \le \frac{4{rE}}{(1+M)^{\frac{1}{2}}-2rE}\equiv 2\gamma\,.
\end{align*}
Next, we approximate $S(\rho_A')$ in terms of $S(\mathcal{P}_M(\rho_A))$, where the entropy of a sub-normalised positive trace class operator $A$ is defined as $S(A)\coloneqq -\tr A\ln A$: denoting $\tr(\mathcal{P}_M(\rho_A))\coloneqq \lambda\le 1$,
\begin{align}
    & S(\rho_A')-S(\mathcal{P}_M(\rho_A))\nonumber\\
    &=-\tr\frac{\mathcal{P}_M(\rho_A)}{\lambda}\ln\frac{\mathcal{P}_M(\rho_A)}{\lambda}+\tr\mathcal{P}_M(\rho_A)\ln\mathcal{P}_M(\rho_A)\nonumber\\
    &=\frac{1-\lambda}{1-(1-\lambda)}\,S(\mathcal{P}_M(\rho_A))+\ln\lambda\nonumber\\
    &\le \frac{\|\rho_A-\mathcal{P}_M(\rho_A)\|_1}{1-\|\rho_A-\mathcal{P}_M(\rho_A)\|_1}\,S(\mathcal{P}_M(\rho_A))\nonumber\\
    &\le \gamma\,S(\mathcal{P}_M(\rho_A))\nonumber\\
    &\le r\ln(M+1)\,\gamma\label{ECcontinuitybound}\,.
\end{align}
% Combining with~\eqref{ECcontinuity}, we have showed that for $\epsilon<\frac{2rE}{1+2rE}$, $M= \big\lceil\frac{16rE}{\epsilon^2}\big\rceil$ satisfies  $\sqrt{M}>(2+4rE)/\sqrt{rE}$ and
% \begin{align}\label{eqSEfinal}
%     |S(\rho_A)-S(\rho_A')|\le h\Big(\frac{\epsilon}{2+\epsilon}\Big)+rE\,h\Big(\frac{2\epsilon }{{4rE}-2rE\epsilon}\Big)\,,
% \end{align}
% and in particular $S(\rho_A')$ approximates $S(\rho_A)$ well for $M$ large enough. 
Next, we use a polynomial approximation of $S(\mathcal{P}_M(\rho_A))$ as was already done in~\cite{Huang2021}: 
\begin{align*}
    &H^{(d_p)}(\mathcal{P}_M(\rho_A))
    \nonumber\\ &\coloneqq -\tr(\mathcal{P}_M(\rho_A)-P_{M})+\sum_{k=2}^{d_p}\frac{\tr\big[(P_M-\mathcal{P}_M(\rho_A))^k\big]}{k(k-1)}
\end{align*}
where $d_p$ is the truncation degree which we will choose later. Since $P_M$ projects onto a subspace of dimension $(M+1)^r$, a simple extension of the proof leading to~\cite[Equation (K46)]{Huang2021} gives us
\begin{align}\label{Sfinal}
    |S(\mathcal{P}_M(\rho_A))-H^{(d_p)}(\mathcal{P}_M(\rho_A))|\le \frac{(M+1)^r}{d_p}\,.
\end{align}
Now, we re-express the function $H^{(d_p)}(\mathcal{P}_M(\rho_A))$ as a linear function of $(\mathcal{P}_M(\rho_A))^{\otimes d_p}$:
%\begin{widetext}
\bb
&H^{(d_p)}(\mathcal{P}_M(\rho_A)) \\
&=-1+(M\!+\!1)^r+\sum_{k=2}^{d_p}\frac{1}{k(k-1)}\,\sum_{j=0}^k\binom{k}{j}\,(-1)^j\tr((\mathcal{P}_M(\rho_A))^j)\\
&=-1+(M\!+\!1)^r \\
&\qquad + \sum_{j=0}^{d_p}\frac{1}{j!}\,\tr(S_j(\mathcal{P}_M(\rho_A))^{\otimes j})\,(-1)^j \!\!\sum_{k=\max\{2,j\}}^{d_p}\frac{(k\!-\!2)!}{(k\!-\!j)!}\\
&\equiv -1+(M\!+\!1)^r+\sum_{j=0}^{d_p}\frac{1}{j!}\,C_j\,\tr(S_j(\mathcal{P}_M(\rho_A))^{\otimes j})\,,
\ee
%\end{widetext}
where $S_j$ is a generalised swap operator over $j$ subsystems. Let us now chose $\sigma_A^{(N)}(M)$ to be the classical shadow arising from the homodyne scheme of~\eqref{homodyne}. We have, similarly to~\cite[Lemma 11]{Huang2021}
\begin{align}\label{eqhdp}
   &| H^{(d_p)}(\mathcal{P}_M(\rho_A))-H^{(d_p)}(\sigma_A^{(N)}(M))|\\
   %&\le \sum_{j=0}^{d_p}\,\frac{1}{j!}|C_j|\,|\tr\big(S_j\,((\mathcal{P}_M(\rho_A))^{\otimes j}-(\sigma_A^{(N)}(M))^{\otimes j}) \big)| \\
   &\le \sum_{j=0}^{d_p}\,\frac{1}{j!}|C_j|\,\|(\mathcal{P}_M(\rho_A))^{\otimes j}-(\sigma_A^{(N)}(M))^{\otimes j}\|_1\,.\nonumber 
\end{align}
Next, we estimate the trace distances in the above summand. By an arbitrary labelling $A_1,\dots,A_j$ of the subsystems so that $(\mathcal{P}_M(\rho_A))^{\otimes j}\coloneqq \rho_{A_1}'\otimes \dots\otimes\rho_{A_j}'$ and $(\sigma_A^{(N)}(M))^{\otimes j}\coloneqq \sigma_{A_1}\otimes\dots\otimes\sigma_{A_j}$, we have
\bb
&\|(\mathcal{P}_M(\rho_A))^{\otimes j}-(\sigma_A^{(N)}(M))^{\otimes j}\|_1 \\
&=\|\rho_{A_1}'\otimes\dots\otimes \rho_{A_j}'-\sigma_{A_1}\otimes\dots\otimes \sigma_{A_j}\|_1 \\
&=\big\|(\rho_{A_1}'-\sigma_{A_1})\otimes \rho_{A_2}'\otimes\dots\otimes\rho_{A_j}'\\
&\qquad +\sigma_{A_1}\otimes (\rho_{A_2}'\otimes\dots\otimes \rho_{A_j}'-\sigma_{A_2}\otimes\dots\otimes\sigma_{A_j})\big\|_1 \\
&\le \|\rho_{A_1}'-\sigma_{A_1}\|_1\|\rho_{A_2}'\|_1\dots\|\rho_{A_j}'\|_1 \\
&\qquad +\|\sigma_{A_1}\|_1\,\big\|\rho_{A_2}'\otimes\dots\otimes \rho_{A_j}'-\sigma_{A_2}\otimes\dots\otimes\sigma_{A_j}\big\|_1 \\
&\le \|\rho_{A_1}'-\sigma_{A_1}\|_1 \\
&\qquad +(1+\|\rho_{A_1}'\!-\!\sigma_{A_1}\|_1)\,\big\|\rho_{A_2}'\!\otimes\!\dots\!\otimes\! \rho_{A_j}'-\sigma_{A_2}\!\otimes\!\dots\!\otimes\!\sigma_{A_j}\big\|_1 .
\label{multitriangle}
\ee
Now, we recall that the homodyne tomography protocol provides us with the following concentration bound borrowed from~\eqref{unionbound}
\begin{align}
    & \mathbb{P}\Big(\exists A, \,|A|\le r :\, \big\|\sigma_A^{(N)}(M)-\mathcal{P}_M(\rho_A)\big\|_1\ge {\epsilon'}\Big)\nonumber\\
    &\le 2\big[m(M+1)\big]^r\,\nonumber\\
    & \quad \times \operatorname{exp}\left(-\frac{3N\epsilon'^2(M+1)^{-2r}}{6\Sigma_r^{(0)}(M)^2+2(\Sigma_r^{(0)}(M)+1)\epsilon'}\right)\,.
\end{align}
Plugging this bound into~\eqref{multitriangle}, we have that with high probability
\begin{align*}
&\|(\rho_A')^{\otimes j}-(\sigma_A^{(N)}(M))^{\otimes j}\|_1 \nonumber\\&\le \epsilon'+\big(1+\epsilon'\big)\|(\mathcal{P}_M(\rho_A))^{\otimes j-1}-(\sigma_A^{(N)}(M))^{\otimes j-1}\|_1\,.
\end{align*}
After iterating the procedure $j$ times, we find 
\begin{align*}
  &  \|(\mathcal{P}_M(\rho_A))^{\otimes j}-(\sigma_A^{(N)}(M))^{\otimes j}\|_1 \nonumber\\&\le \epsilon'\sum_{i=0}^{j-1}\big(1+\epsilon'\big)^i\le \big(1+\epsilon'\big)^j-1\,.
\end{align*}
Into~\eqref{eqhdp}, we have that
\begin{align}
    & | H^{(d_p)}(\mathcal{P}_M(\rho_A))-H^{(d_p)}(\sigma_A^{(N)}(M))|\nonumber\\ & \le \sum_{j=0}^{d_p}\frac{1}{j!}|C_j|\,\Big(\Big(1+\epsilon'\Big)^j-1\Big)\nonumber\\
    &\le \big(\big(1+\epsilon'\big)^{d_p}-1\big)\,\sum_{j=0}^{d_p}\,\frac{|C_j|}{j!}\nonumber\\
    & \le  \big(\big(1+\epsilon'\big)^{d_p}-1\big)\,2^{d_p}\,.\label{hpfinal}
\end{align}
Then, combining~\eqref{ECcontinuity},~\eqref{ECcontinuitybound},~\eqref{Sfinal} and~\eqref{hpfinal}, we end up with 
\begin{align*}
    & |S(\rho_A)-H^{(d_p)}(\sigma_A^{(N)}(M))|\nonumber\\
    &\le h(\gamma)+rEh\Big(\frac{\gamma}{rE}\Big)+2r\ln(M+1)\gamma\nonumber\\
    & \quad +\frac{(M+1)^r}{d_p}+ \big(\big(1+\epsilon'\big)^{d_p}-1\big)\,2^{d_p}\,.
\end{align*}
Now, the first three terms on the right-hand side above are small for $M$ large enough, the fourth term is small for $d_p$ large enough, and the last term is small for $N$ large enough. The result follows after choosing $M$ so that the first three terms are smaller than $\epsilon/3$, $d_p=\Big\lceil \frac{3(M+1)^r}{\epsilon}\Big\rceil$ so that the fourth term is smaller than $\epsilon/3$, and finally $\epsilon'=\frac{\epsilon^2}{12(M+1)^re}2^{-4(M+1)^r/\epsilon}$ so that the last term is smaller than $\epsilon/3$.
\end{proof}

\section{Discretisation scheme via quasi Monte-Carlo integration}

Since the objects we are manipulating in this paper are defined on continuous metric spaces, we need to explain how to devise an efficient description $\sigma^{\operatorname{dis}}_N$ of the estimator state $\sigma_N$ in terms of a number of discrete parameters which scales at most polynomially with the number $m$ of modes and such that $\tr(O(\sigma^{\operatorname{dis}}_N-\sigma_N))$ can be controlled for the observables $O$ for which we want to learn the average $\tr(\rho O)$. One natural strategy consists e.g.~in approximating the coefficients $\tr(\widetilde{Z}_{\mathbf{n}_1\mathbf{n}_2}\sigma_N)$ of the matrix $\widetilde{\mathcal{P}}_{{M}}(\sigma_N)$. Since the characteristic function of $\widetilde{Z}_{\mathbf{n}_1\mathbf{n}_2}$ is compactly supported in a ball $B(0,R)$ around the origin, it suffices to store its values on a net of small enough mesh.

% \subsection{Quasi Monte-Carlo integration}

In order to evaluate traces  $\tr(\tilde Z_{\mathbf n_1 \mathbf n_2} \sigma_N),$ in an efficient way, we shall employ quasi Monte-Carlo techniques using the isometry~\eqref{symbol}.
We start by recalling that an $L^1$ function $u: \Omega \to \mathbb R$ is of bounded variation if
\begin{align*}
 &\operatorname{TV}(u)\nonumber\\& \coloneqq \sup \left\{ \int_{\Omega} u(x) \operatorname{div}(\phi(x)); \phi \in C_c^1(\Omega, \mathbb R^n); \Vert \phi \Vert_{\infty}\le 1 \right\} \nonumber\\
 & \quad <\infty\,
 \end{align*}
 and denote by $\operatorname{TV}(K,\mathbb{R})$ the space of functions of bounded variation. In particular, for $u \in C^1(\Omega, \mathbb R^n)$ one just has 
\begin{align*}
 \operatorname{TV}(u) = \int_{\Omega} \vert \nabla u(x) \vert \ dx.
\end{align*}

For a complex valued function to be of bounded variation, we accordingly require that both its real and imaginary part are of bounded variation. 
Let $K=[-l_j,l_j]^m$ be a compact set and $\ell: [0,1]^m \to K$ a linear map. We set $S =\{t_k\}_{k \in \mathbb N},$ where $t_k \in [0,1]^m$ is a \emph{Halton sequence}, see~\cite{Niederreiter} for a definition, in a pairwise prime basis $b_1,..,b_m.$ 
We then introduce for $f \in \operatorname{TV}(K,\mathbb R)$, the numerical integral 
\[I_K(f,k) = \frac{1}{k} \sum_{i=1}^k f(\ell(t_i)).\]
It then follows that there exists a constant $C(b_1,...,b_d)$ independent of both $f$ and $N$ such that, see~\cite[Theorem $2.11$]{Niederreiter},
\begin{align*}
 &\Big \vert I_K(f,k)-\int_{K} f(x) \ dx \Big \vert \nonumber\\ &\le \operatorname{TV}(f) \,C(\ell(t_1),...,\ell(t_m))\, \frac{\log(k)^m}{k}.
 \end{align*}
Thus, this technique allows us to approximate high-dimensional integrals with errors that depend only very mildly on the number of modes $m.$ We can then apply this construction to $$f(x)=\frac{\chi_{\vert \mathbf n_1\rangle \langle \mathbf n_2 \vert}(x) \chi_{\sigma_N}(x) \prod_{j=1}^m\xi_{\eta,R}(x_j)}{(2\pi)^m},$$ with $K = [-R,R]^m.$

\section{Examples}

In this section, we test our homodyne classical shadow tomography method by means of numerical simulations.

\begin{figure}[h!]
\includegraphics[width=4.2cm]{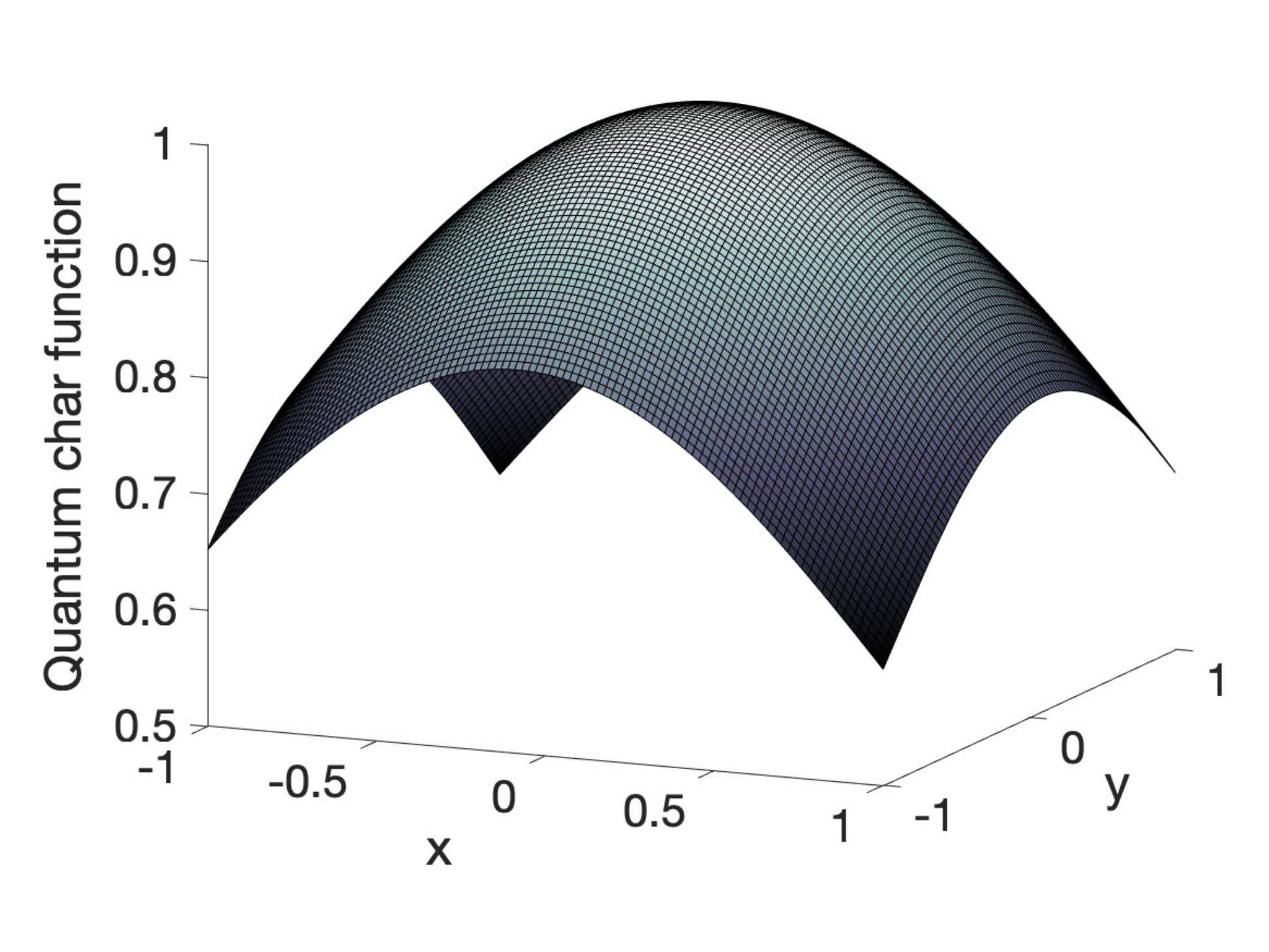}
\includegraphics[width=4.2cm]{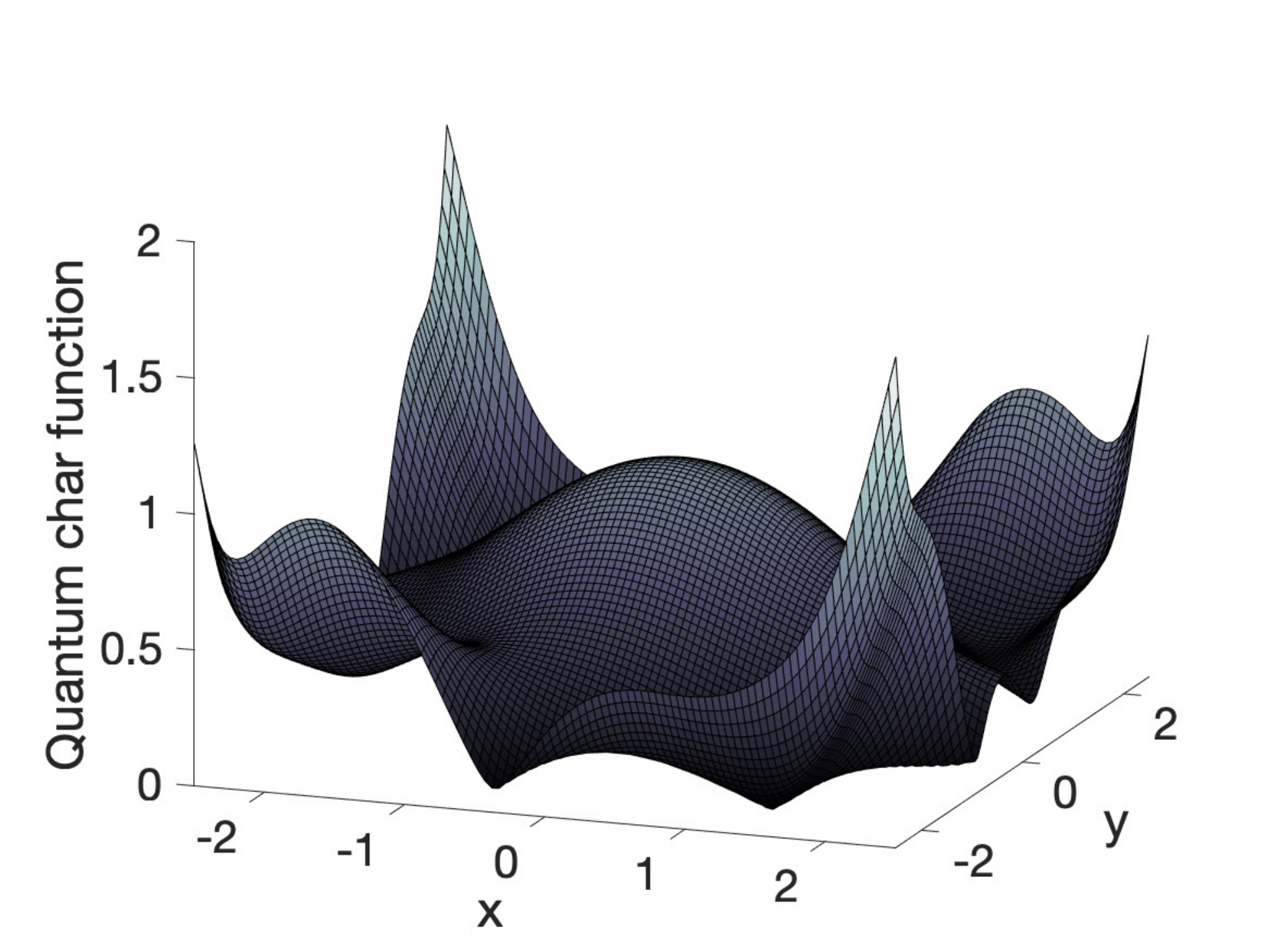}\\
\includegraphics[width=4.2cm]{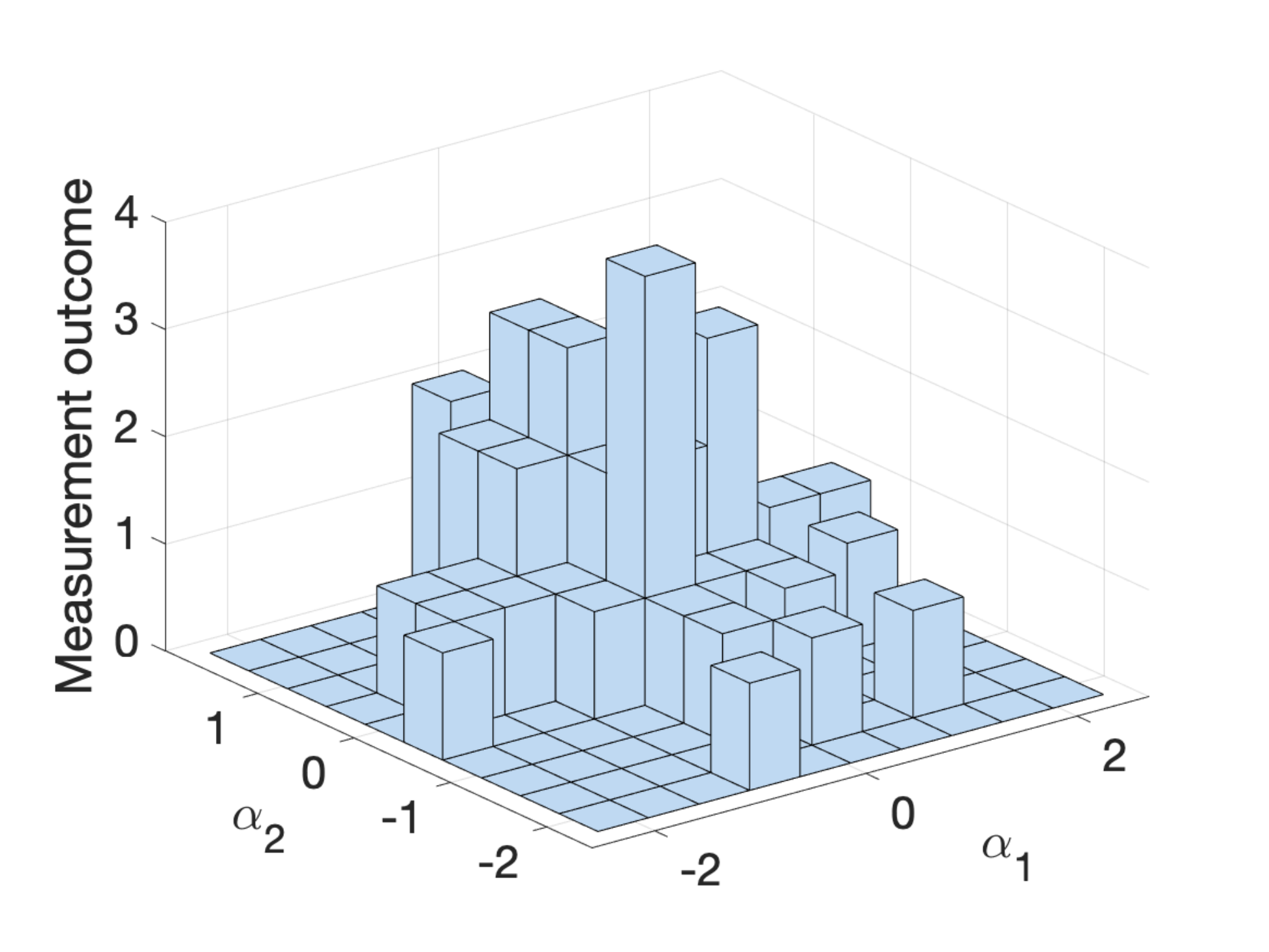}
\includegraphics[width=4.2cm]{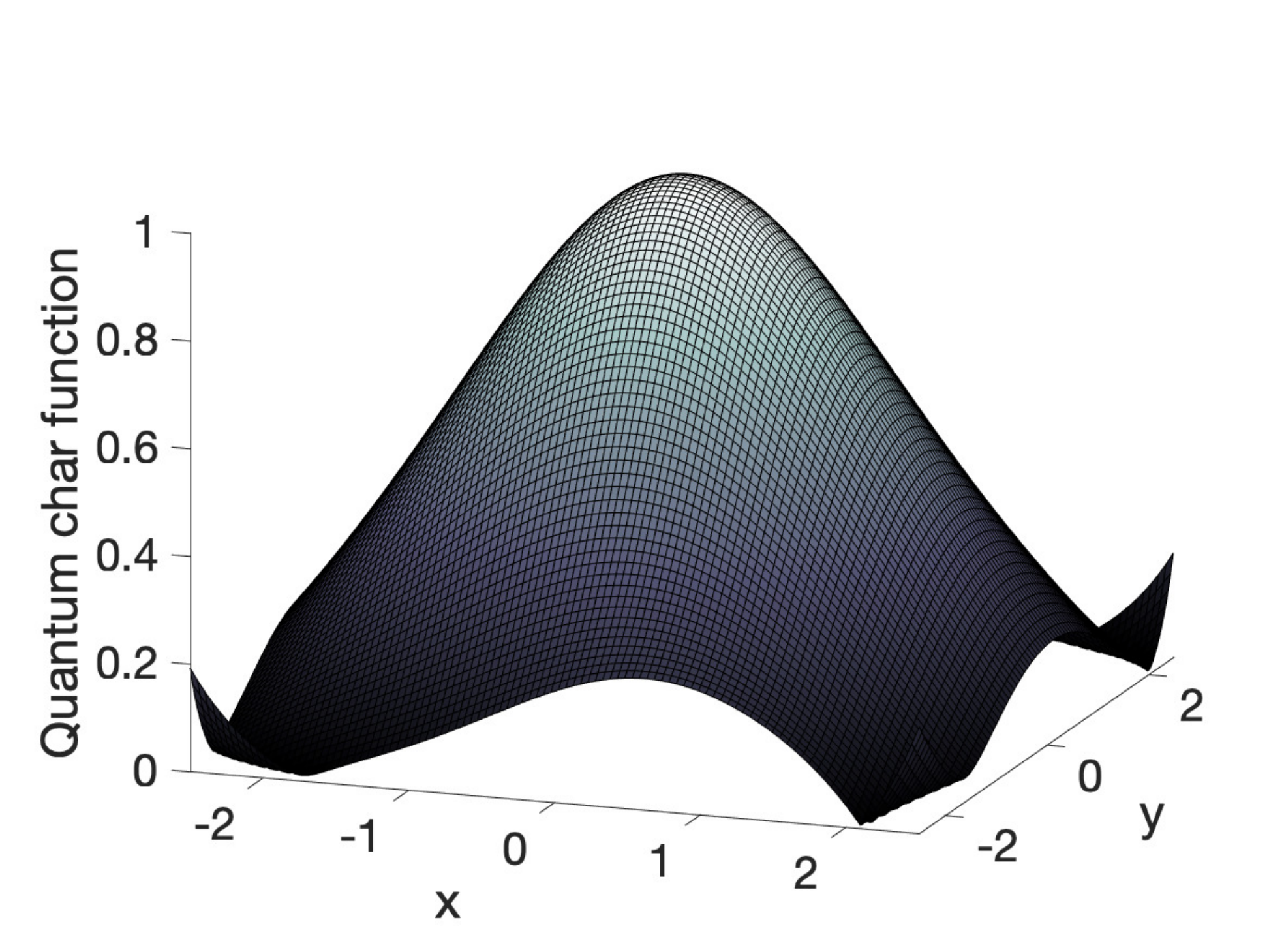}\\
\includegraphics[width=4.2cm]{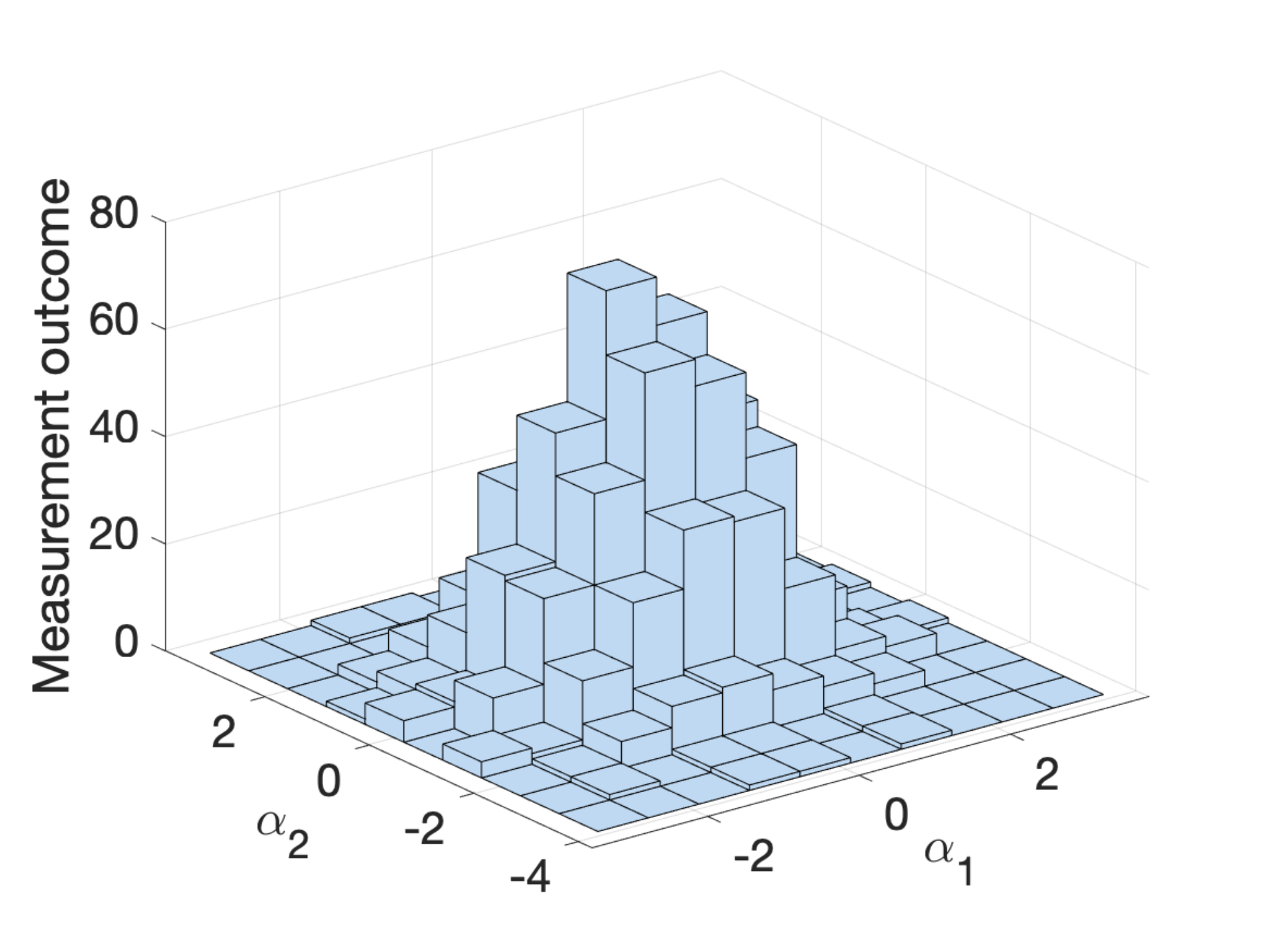} 
\includegraphics[width=4.2cm]{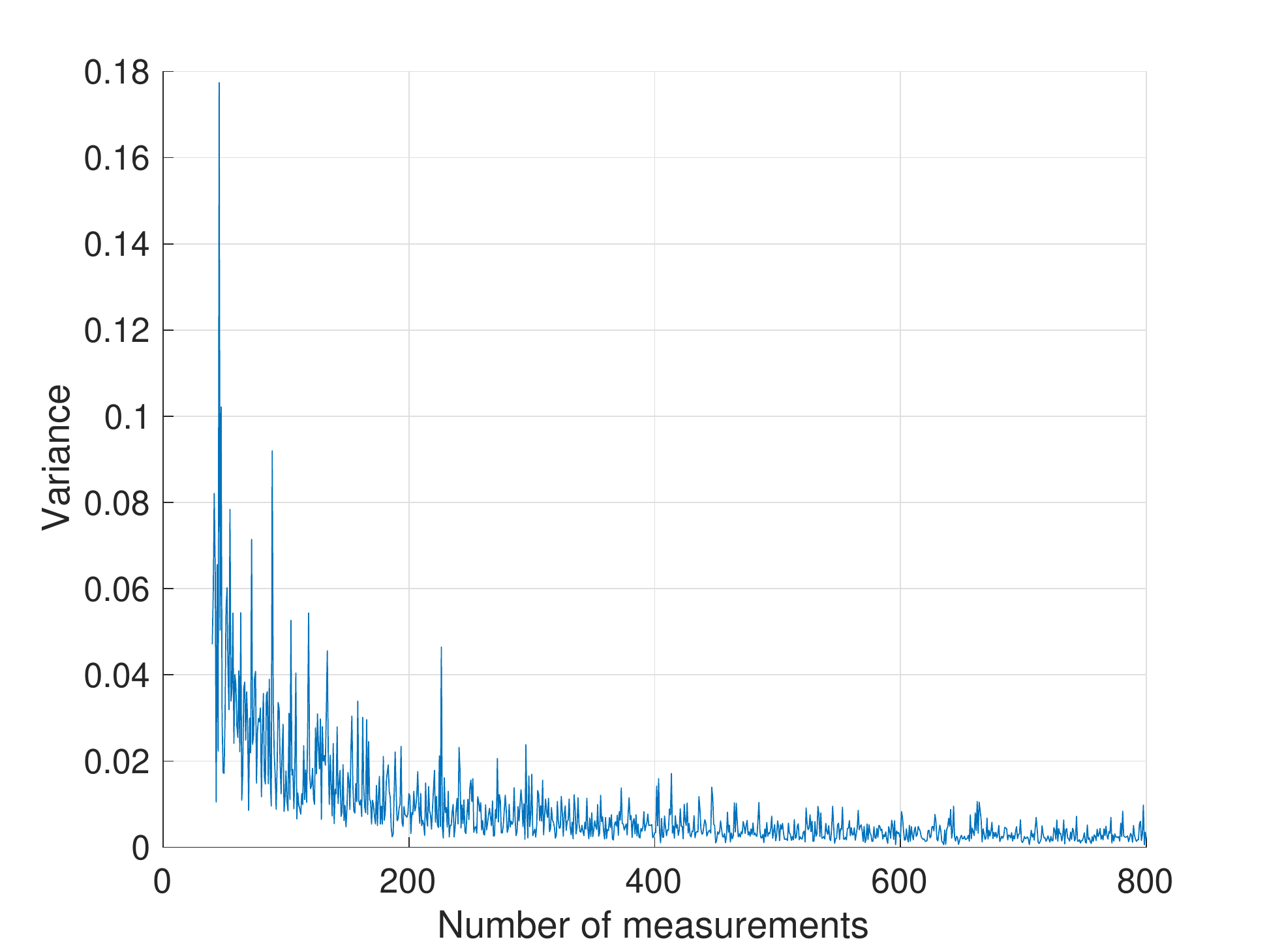}
\caption{{The quantum characteristic function of the actual state is depicted in the first plot on the top left, while that of the reconstructed state is given in the plot beside it (for $N=50$) and the one below it (for $N=1000$).} For $N=50$ (above) the quantum characteristic function gets almost perfectly approximated within the unit square but diverges outside of it. On the right, we show the histogram of numerical measurement outcomes generating our approximation. For $N=1000$ (below) the quantum characteristic function is well approximated inside the square $[-2,2]^2$. In the center, we show the histogram of numerical measurement outcomes generating our approximation. On the bottom right, we show the variance of the approximated characteristic function to the real characteristic function inside $[-2,2]^2.$ \label{fig:approx}}
\end{figure}
\begin{figure}[h!!]
\includegraphics[width=4cm]{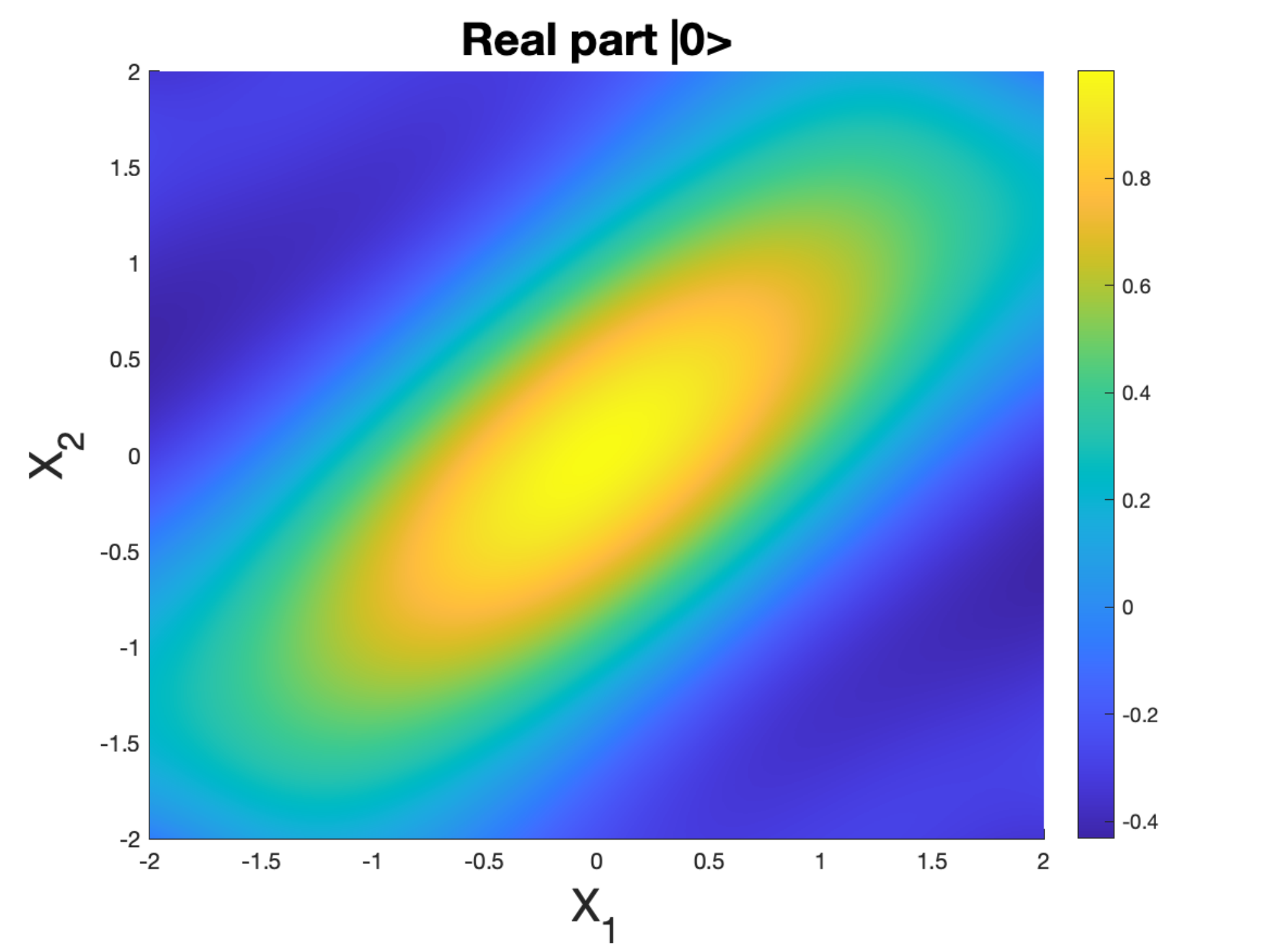}
\includegraphics[width=4cm]{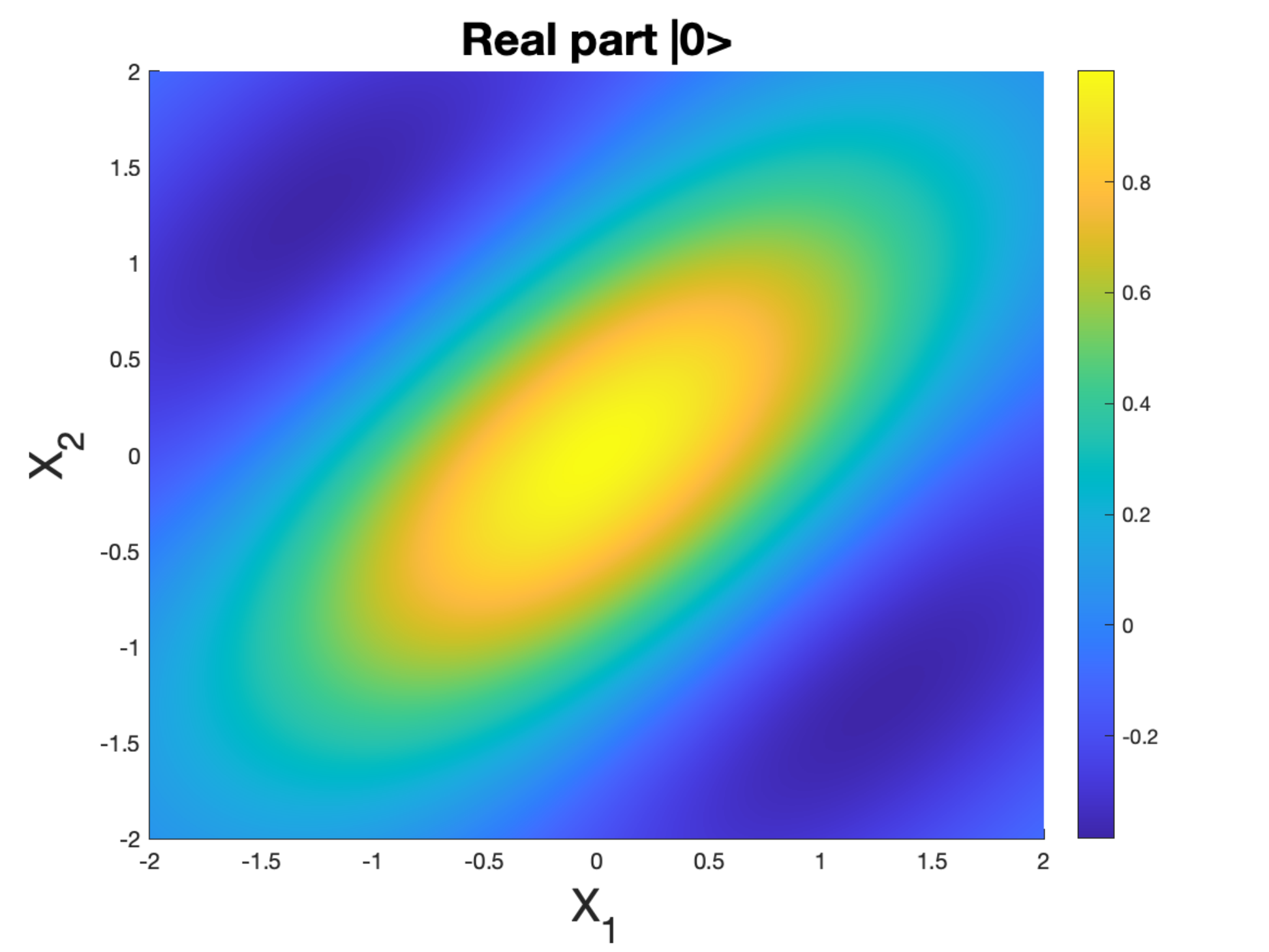} 
\includegraphics[width=4cm]{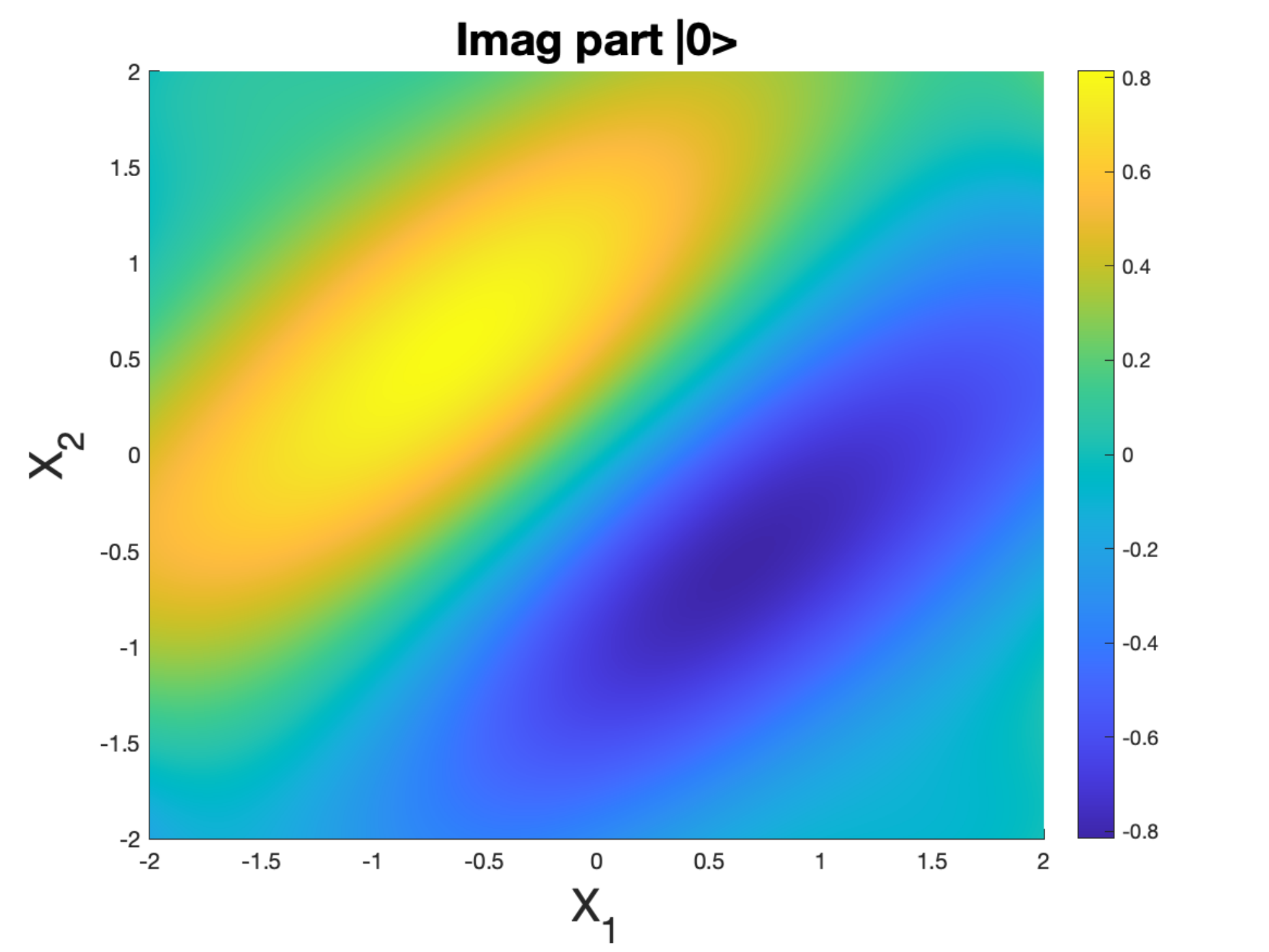}
\includegraphics[width=4cm]{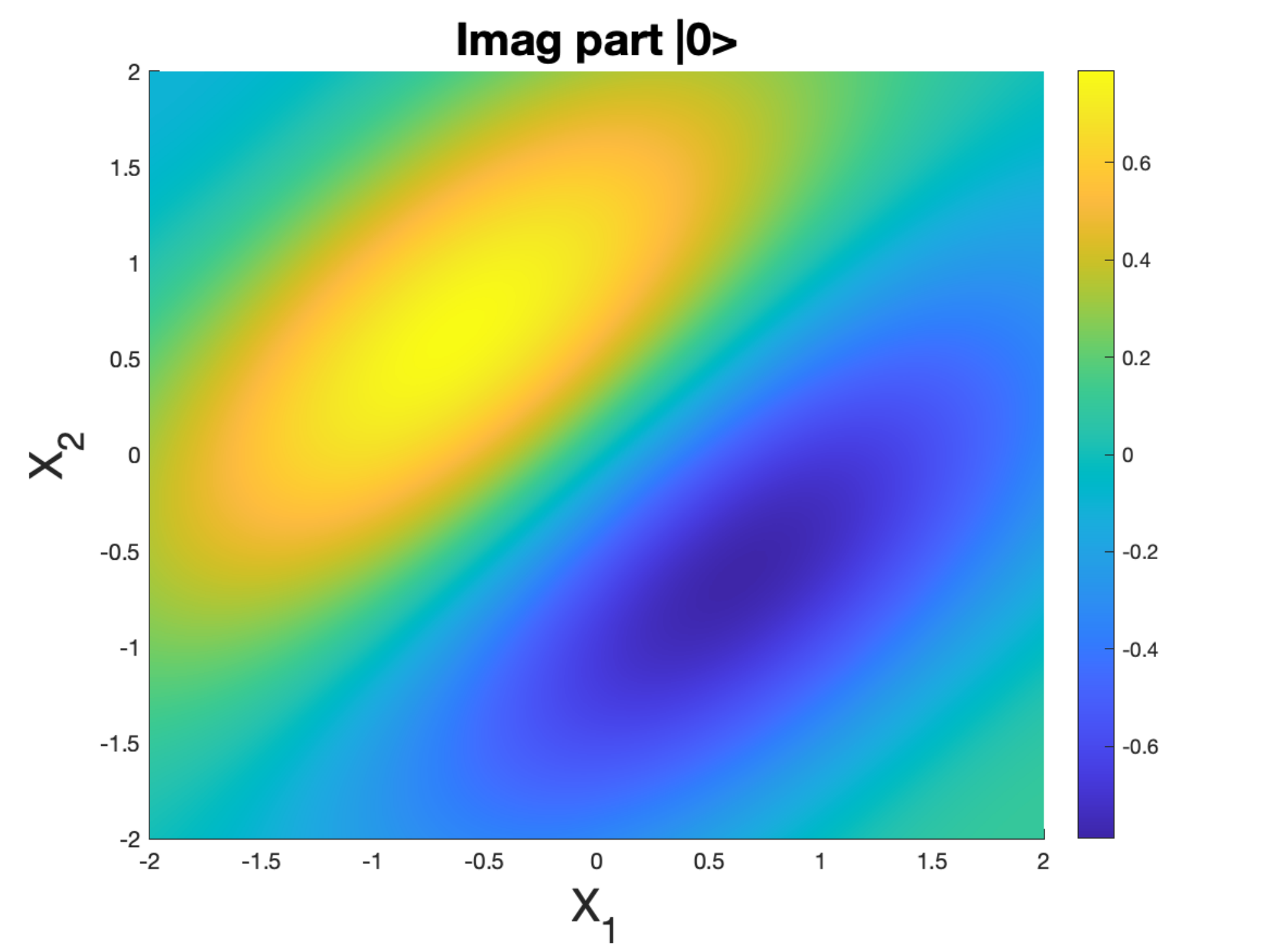}\\
\includegraphics[width=4cm]{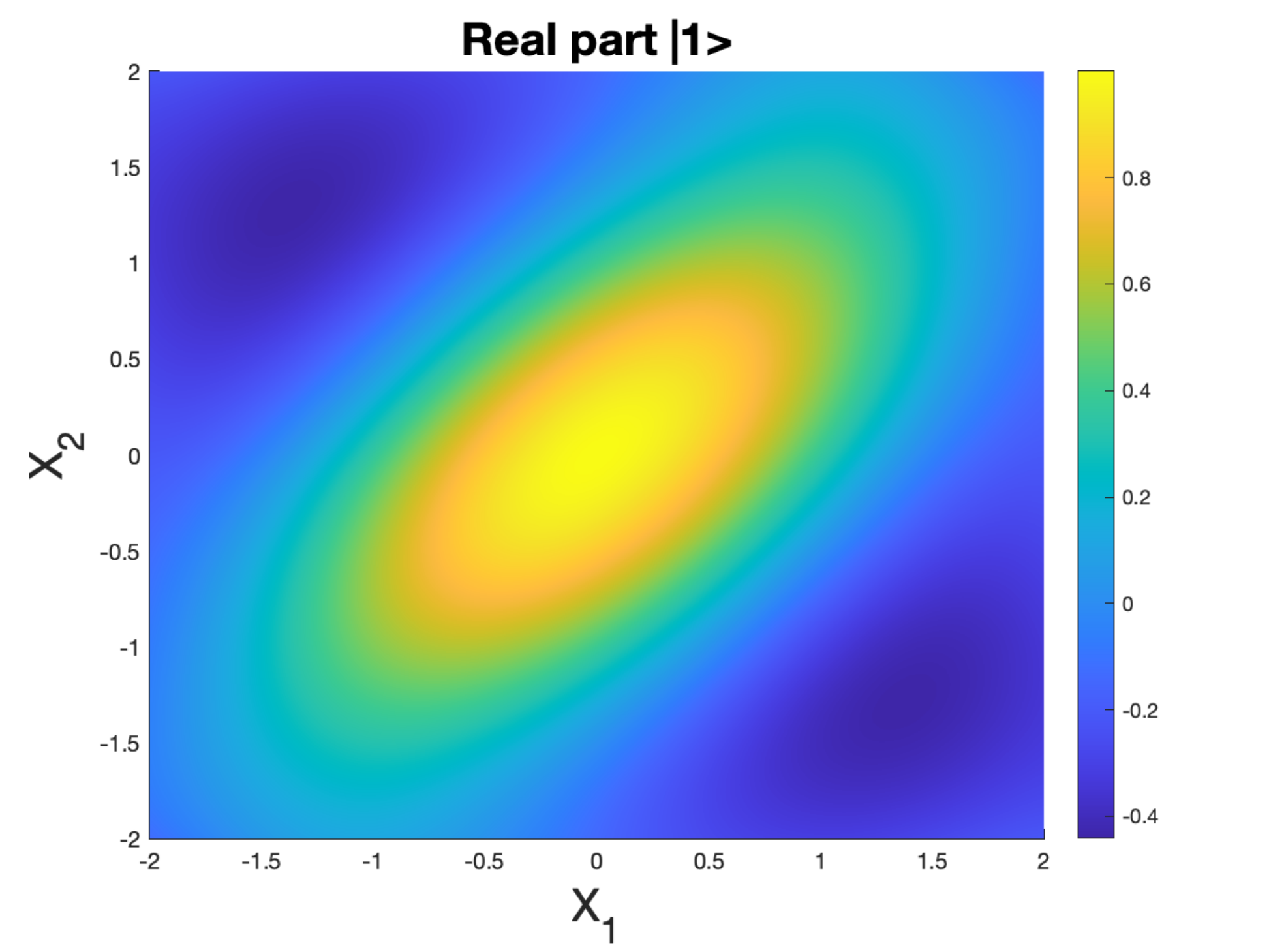}
\includegraphics[width=4cm]{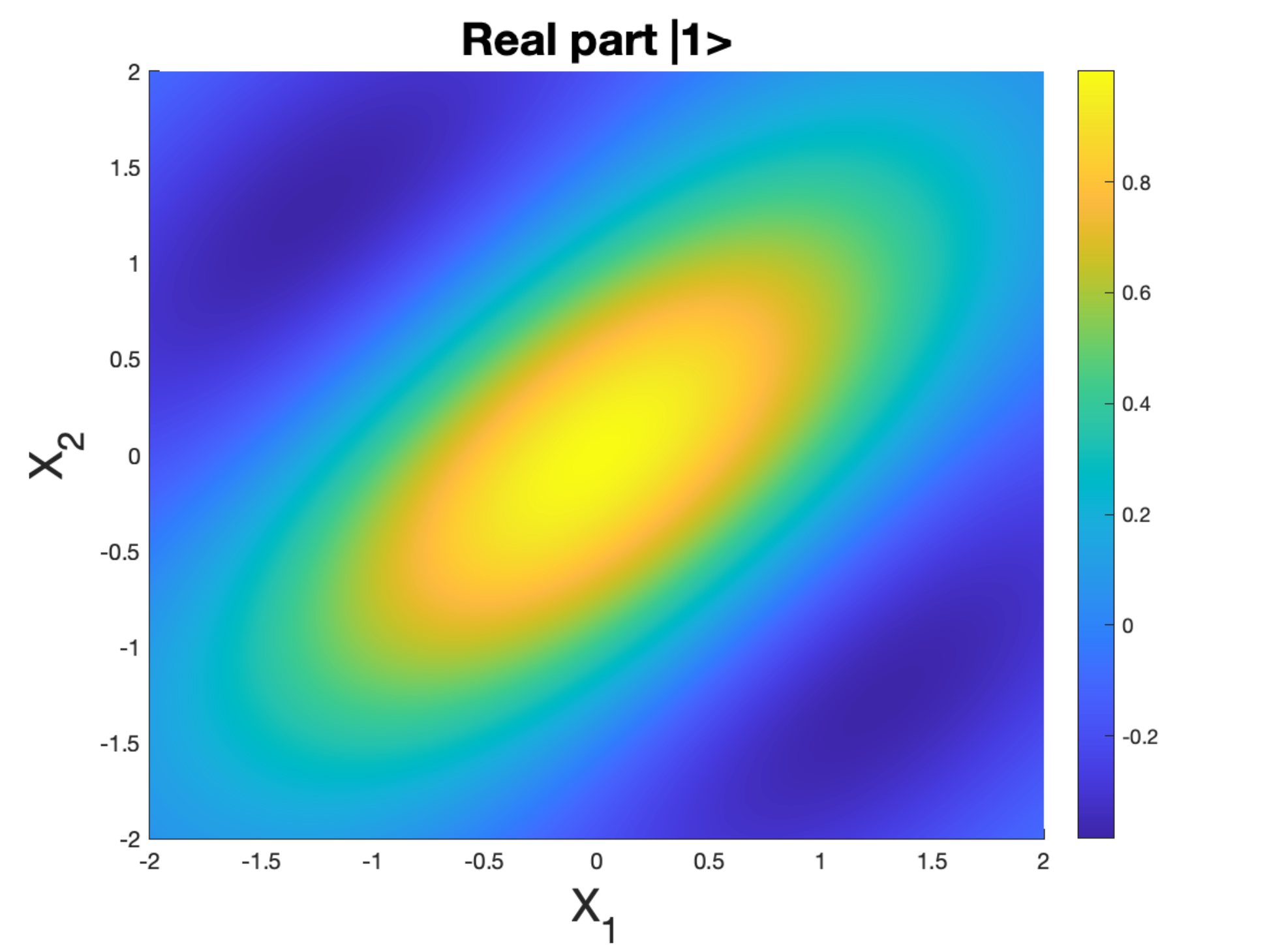} 
\includegraphics[width=4cm]{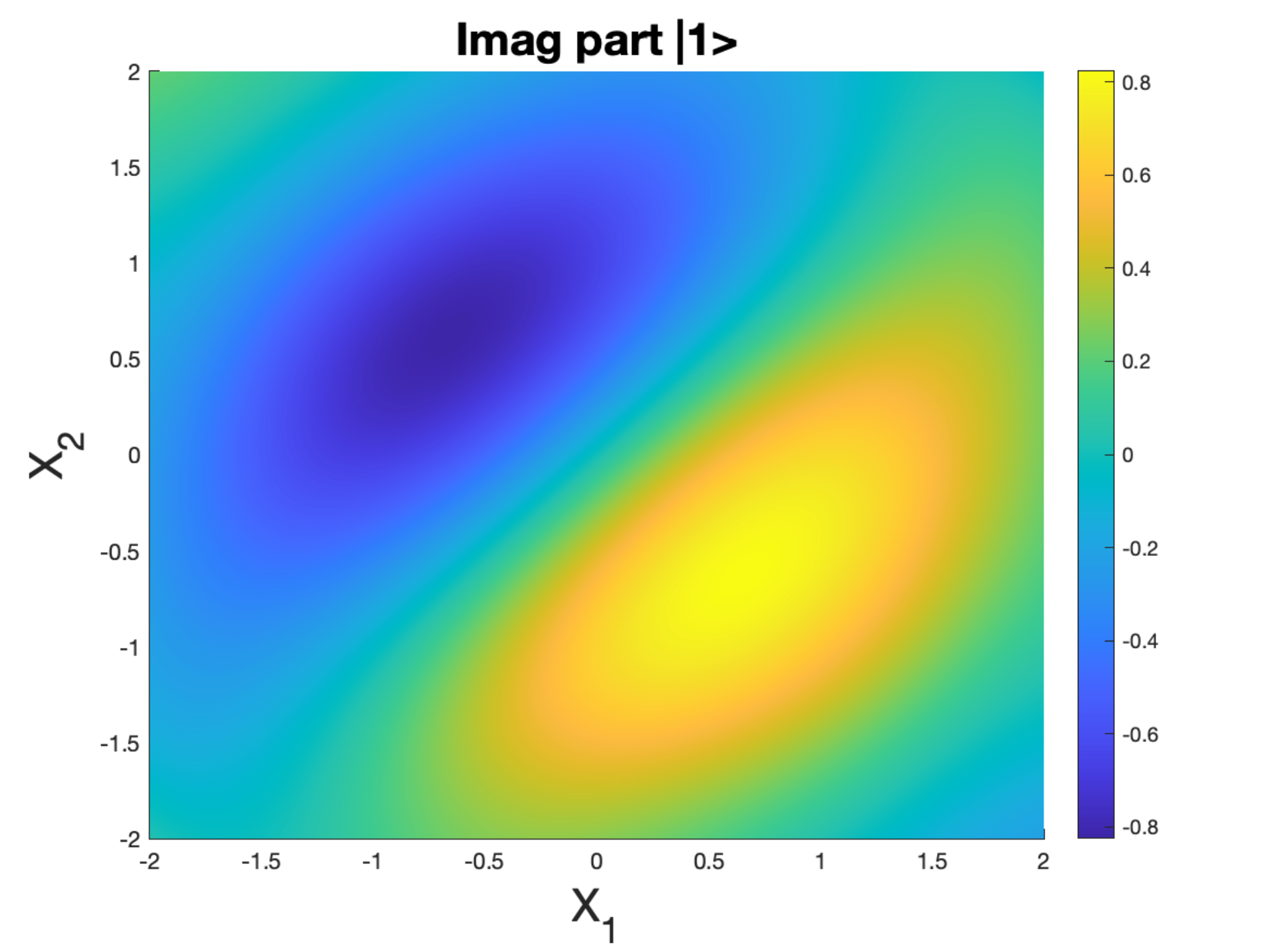}
\includegraphics[width=4cm]{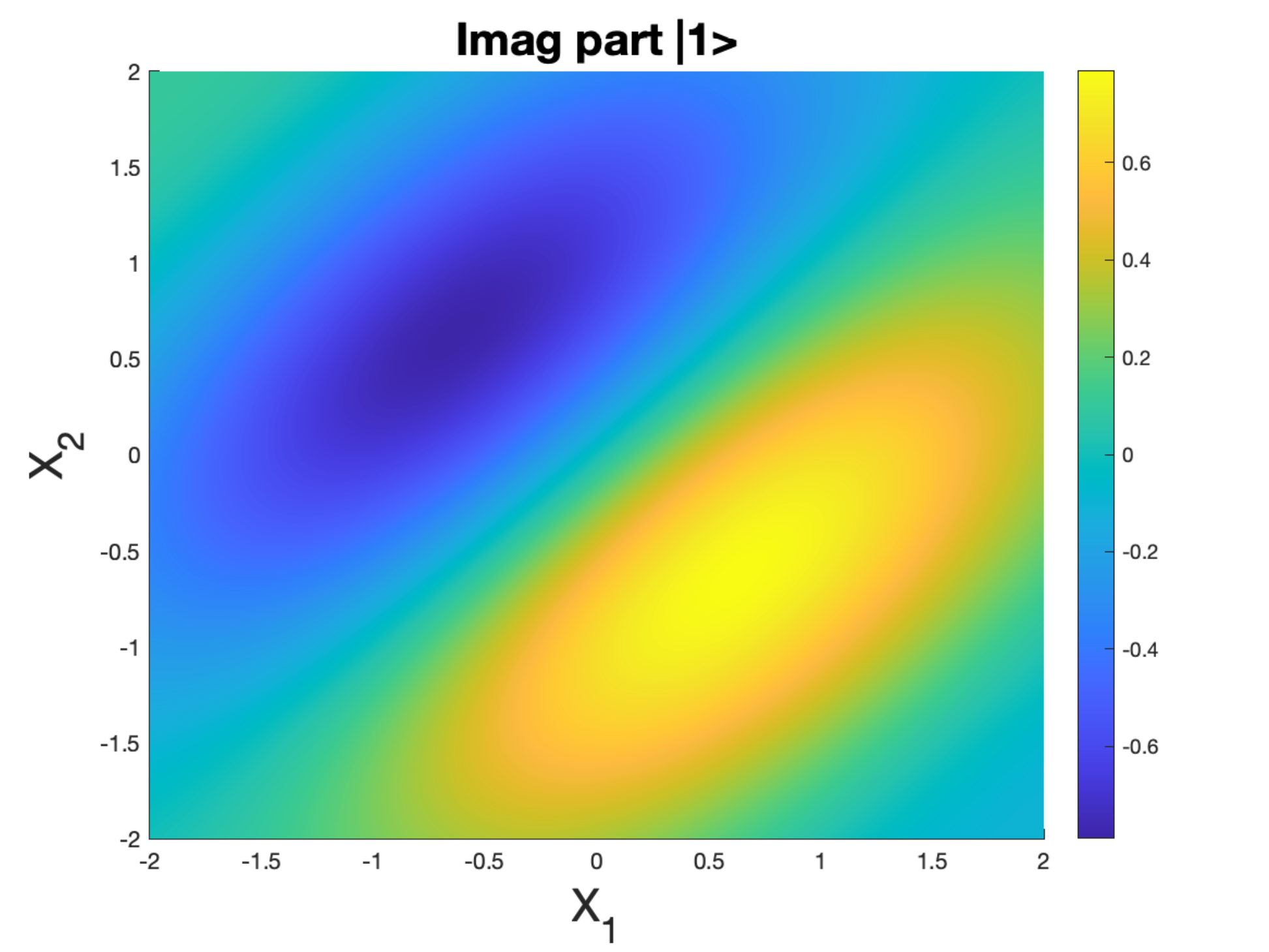}\\
\includegraphics[width=4cm]{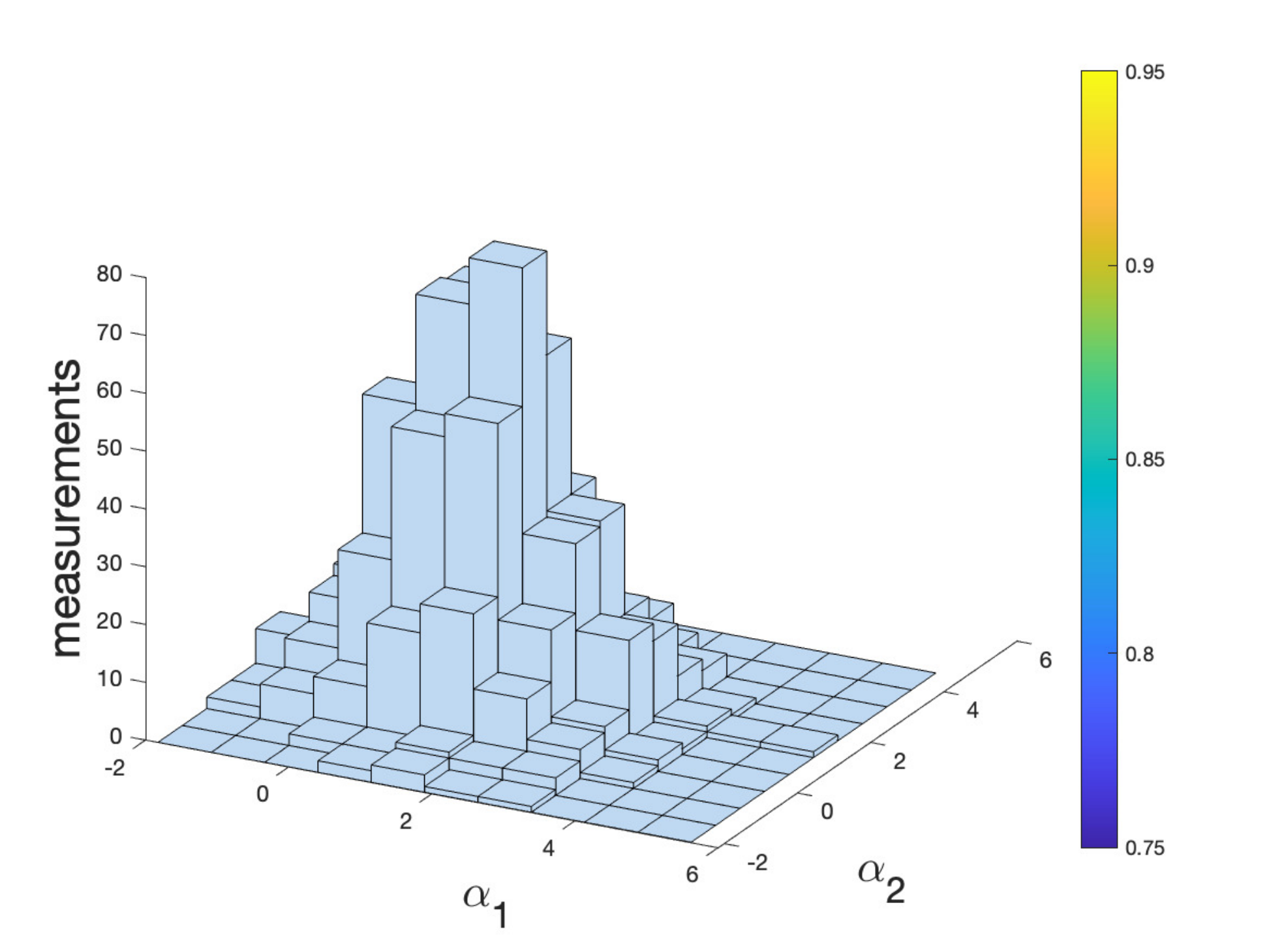}
\includegraphics[width=4cm]{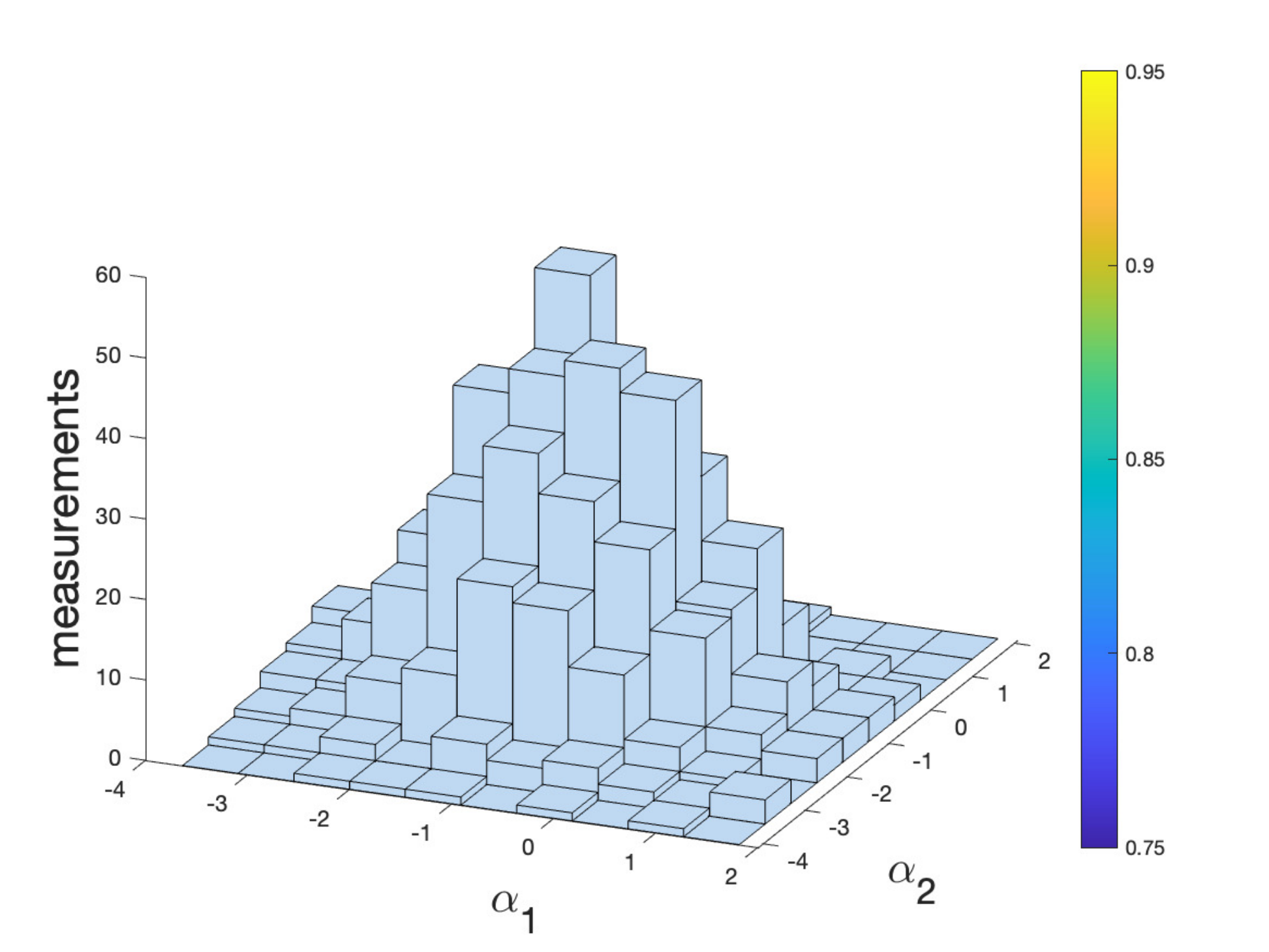} 
\includegraphics[width=4cm]{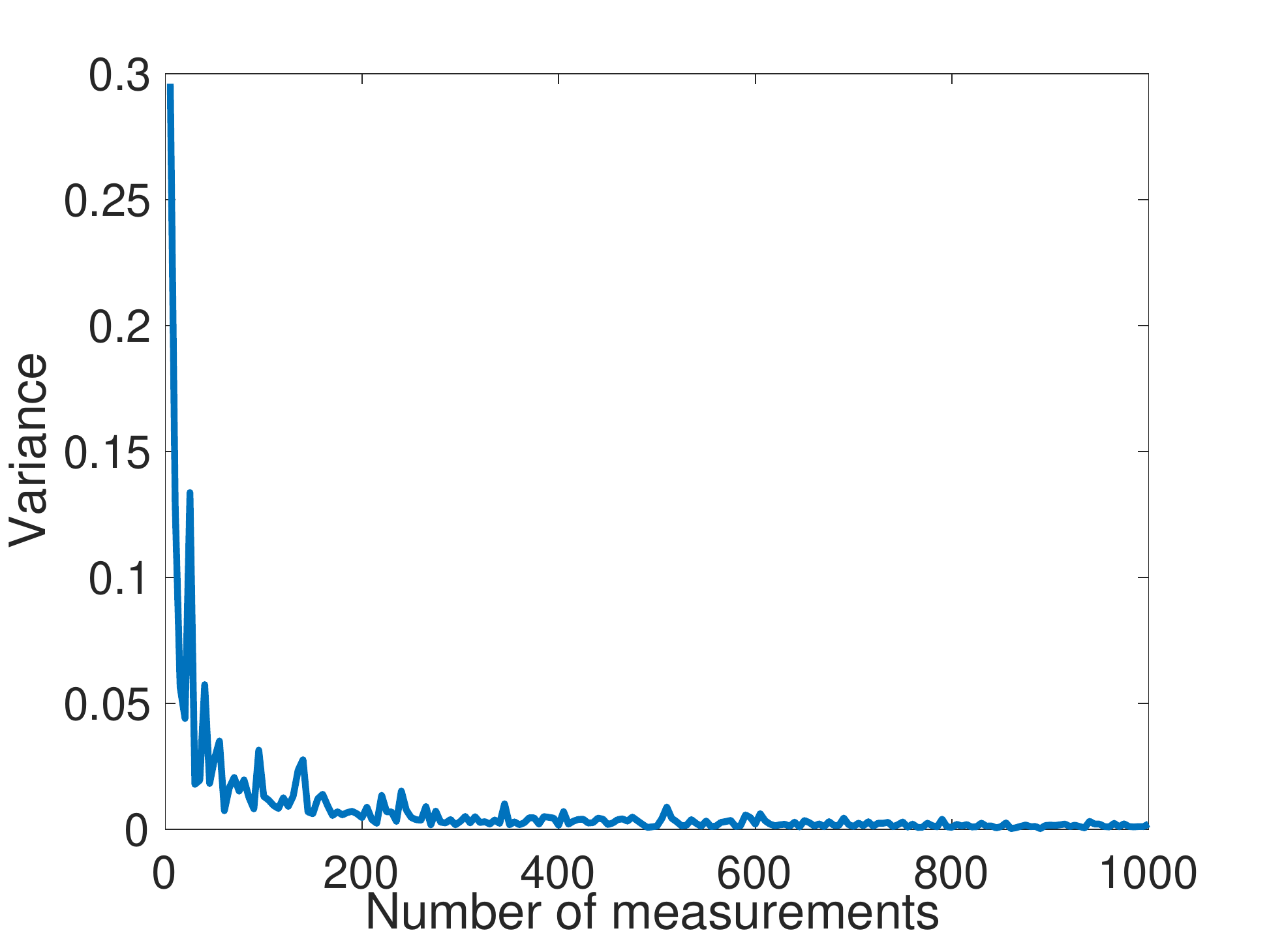}
\includegraphics[width=4cm]{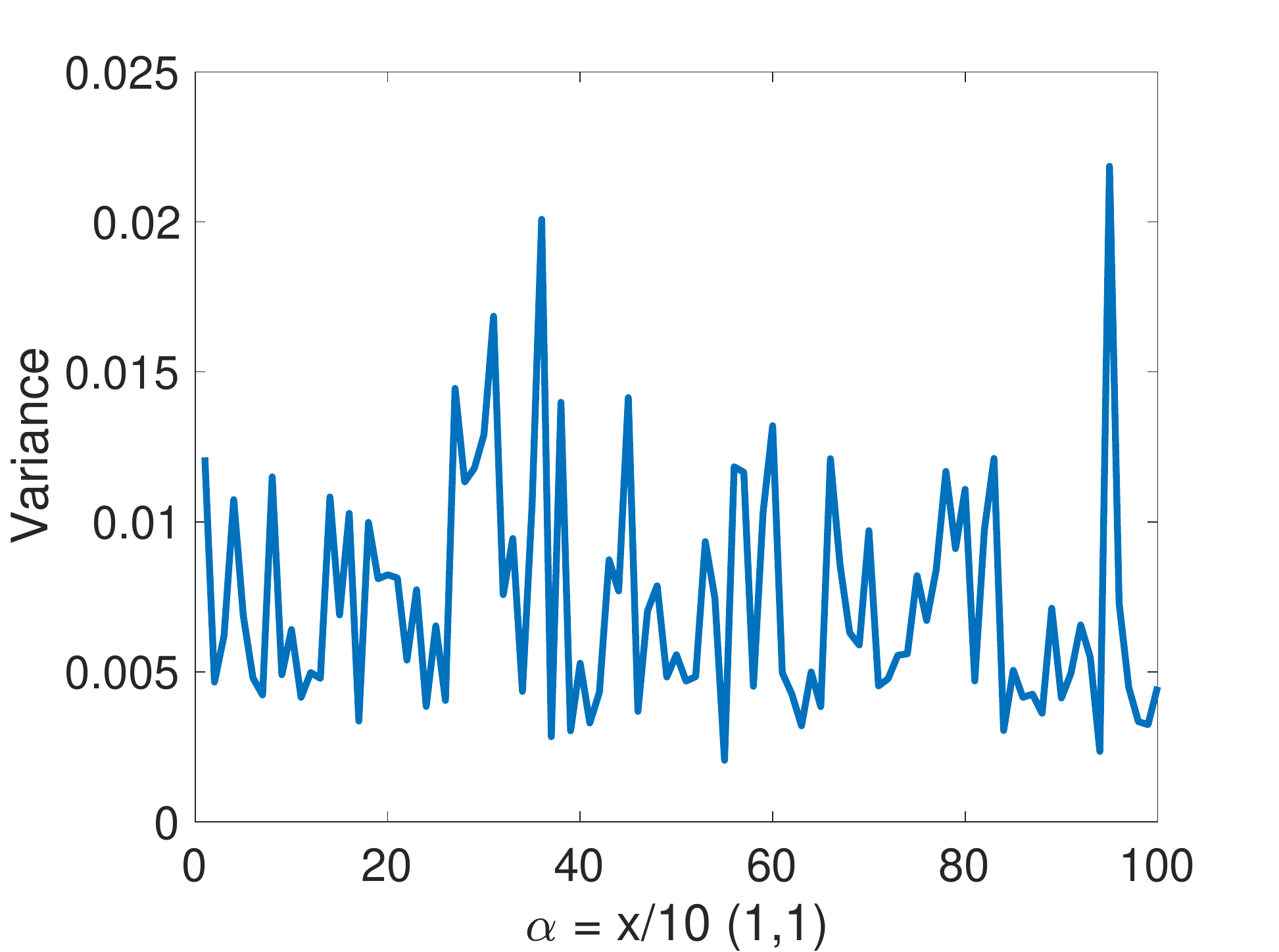}
\caption{Real and imaginary parts of the quantum characteristic functions of $\ket 0_{\operatorname{Cat}}$(top line) and $\ket 1_{\operatorname{Cat}}$ (center line), respectively, reconstructed with $N=200$, with reconstructed function (left) and true one (right) on $[-2,2]$ for $\alpha=(1,1)$\label{fig:qcf}. 
At the bottom, we illustrate the numerical measurement outcomes for the above reconstructions ($\ket 0_{\operatorname{Cat}}$ left, $\ket 1_{\operatorname{Cat}}$ right) and the variance (for $\ket{0}_{\operatorname{Cat}}$) for fixed $\alpha = (1,1)$ and varying number of measurements (left) and fixed number of measurements, $N=200$, and varying $\alpha$ (right). }
\end{figure}

\begin{figure}[h!!]
\includegraphics[width=4cm]{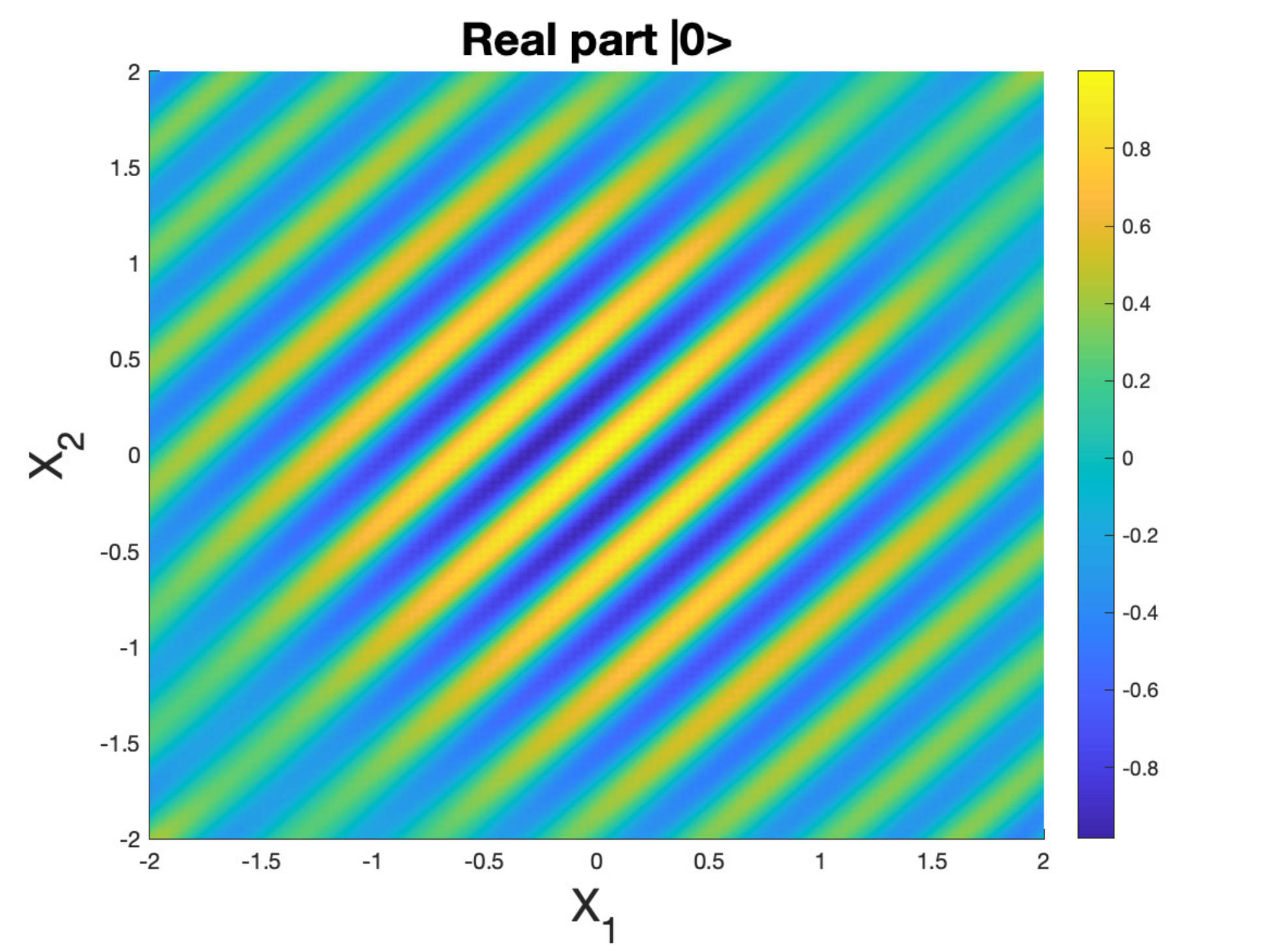}
\includegraphics[width=4cm]{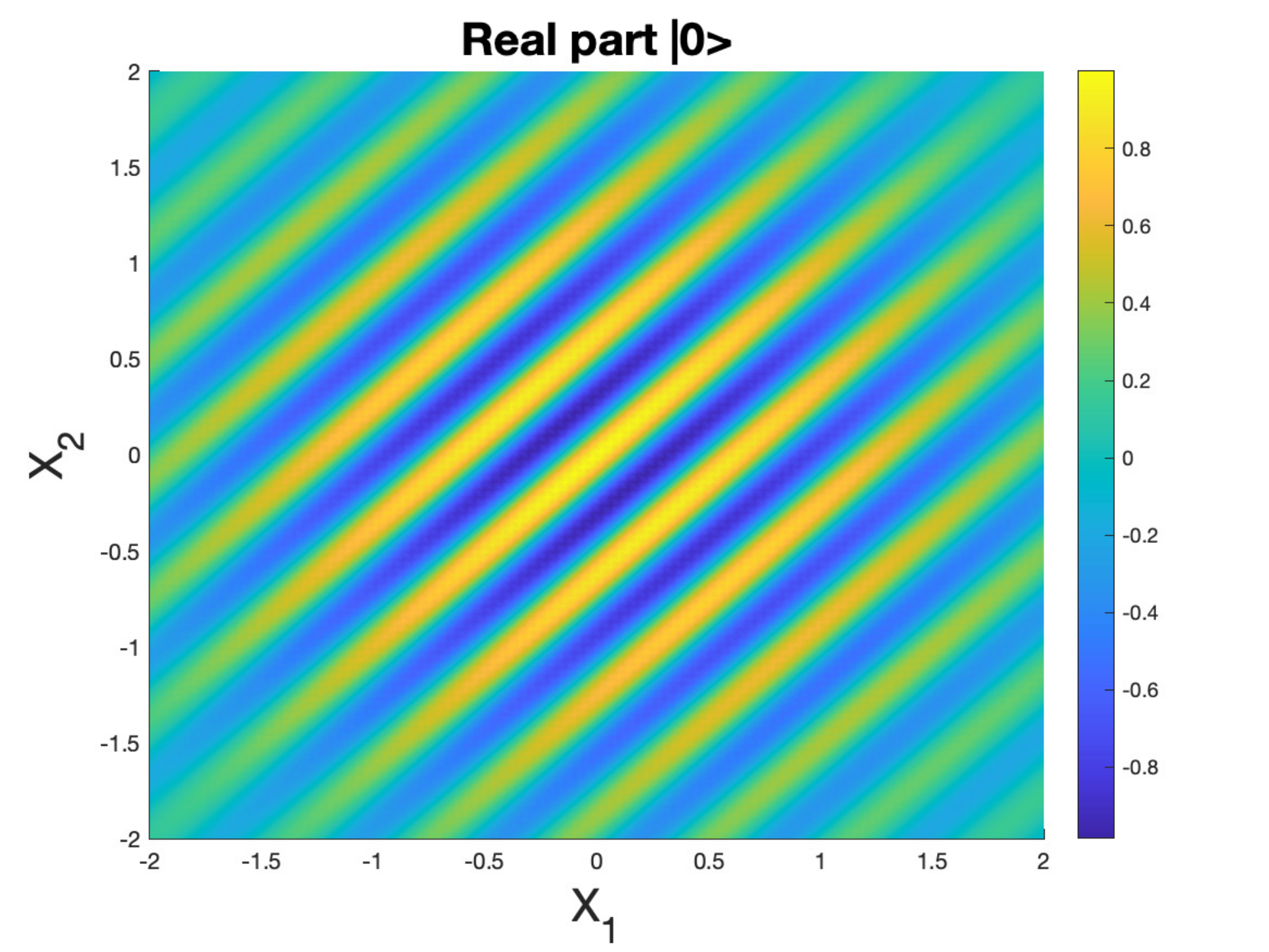} 
\includegraphics[width=4cm]{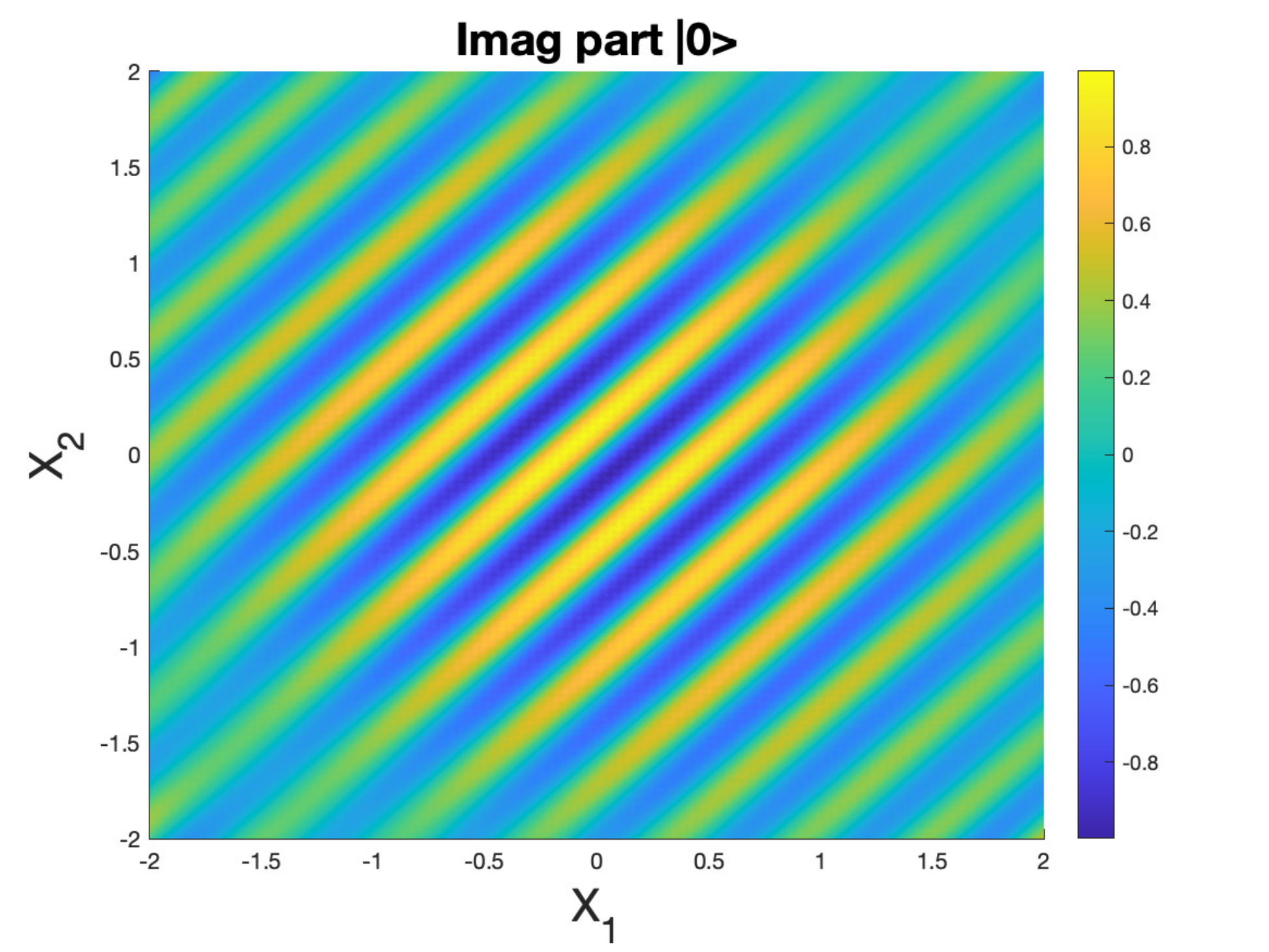}
\includegraphics[width=4cm]{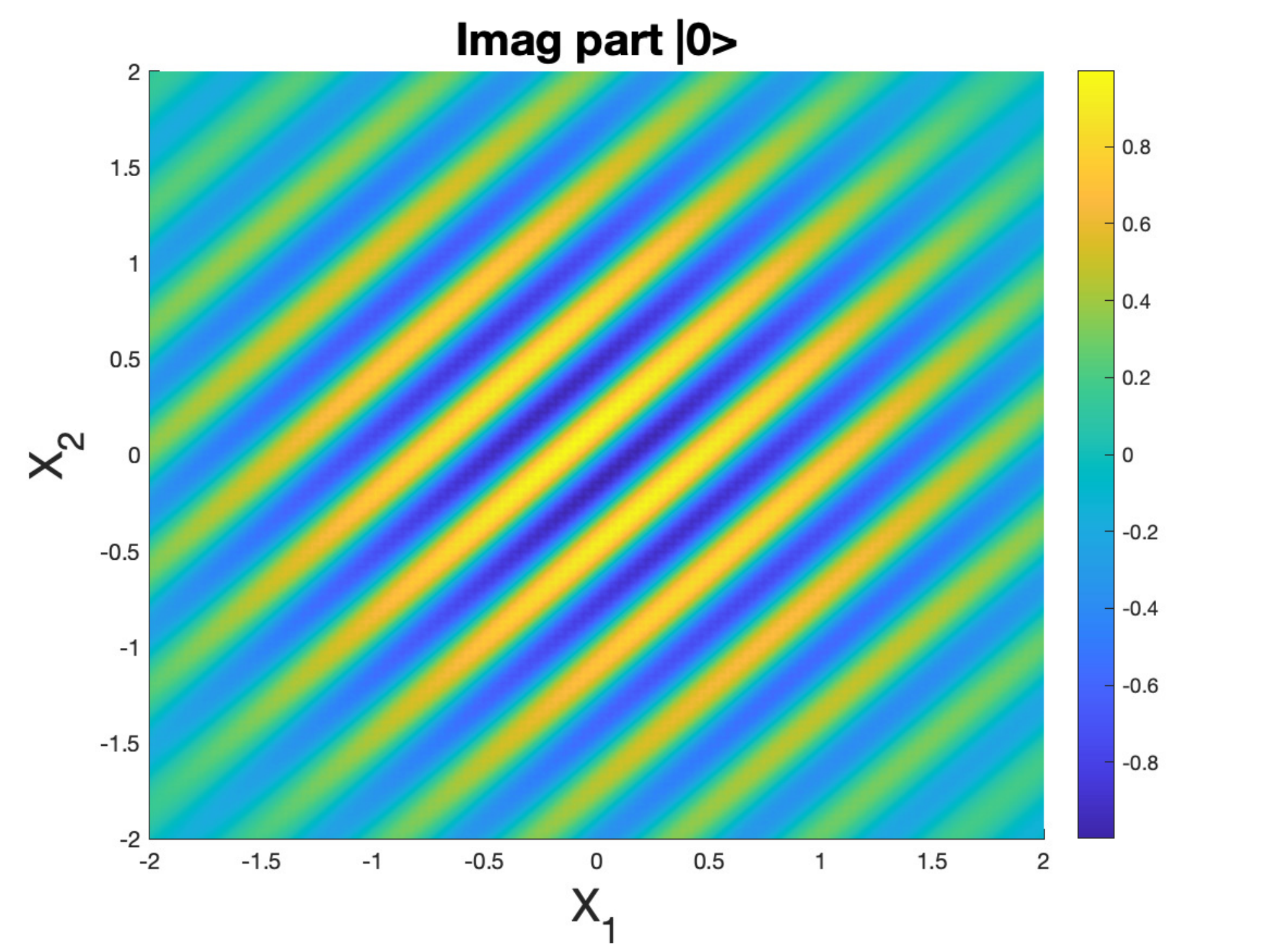}\\
\includegraphics[width=4cm]{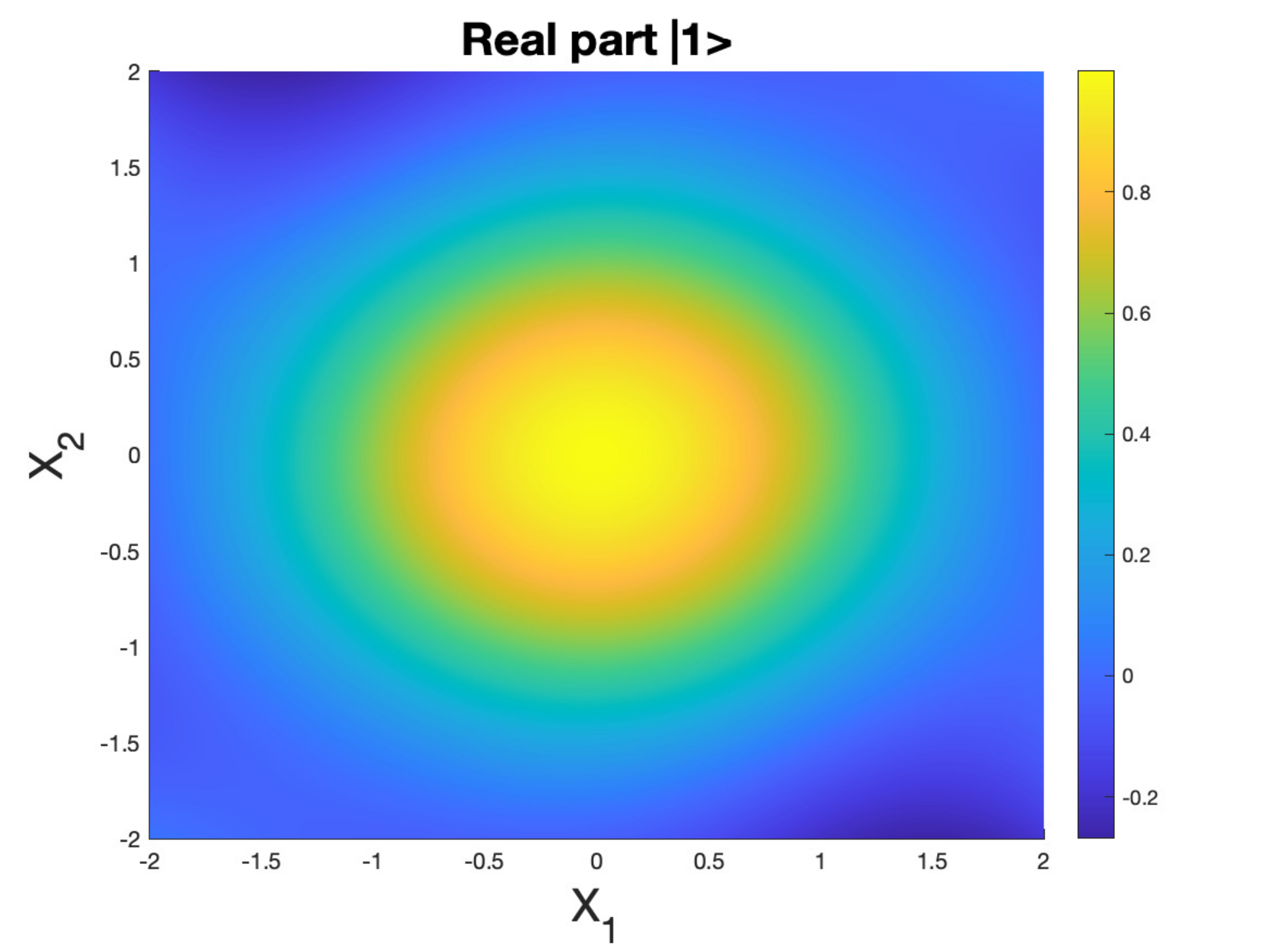}
\includegraphics[width=4cm]{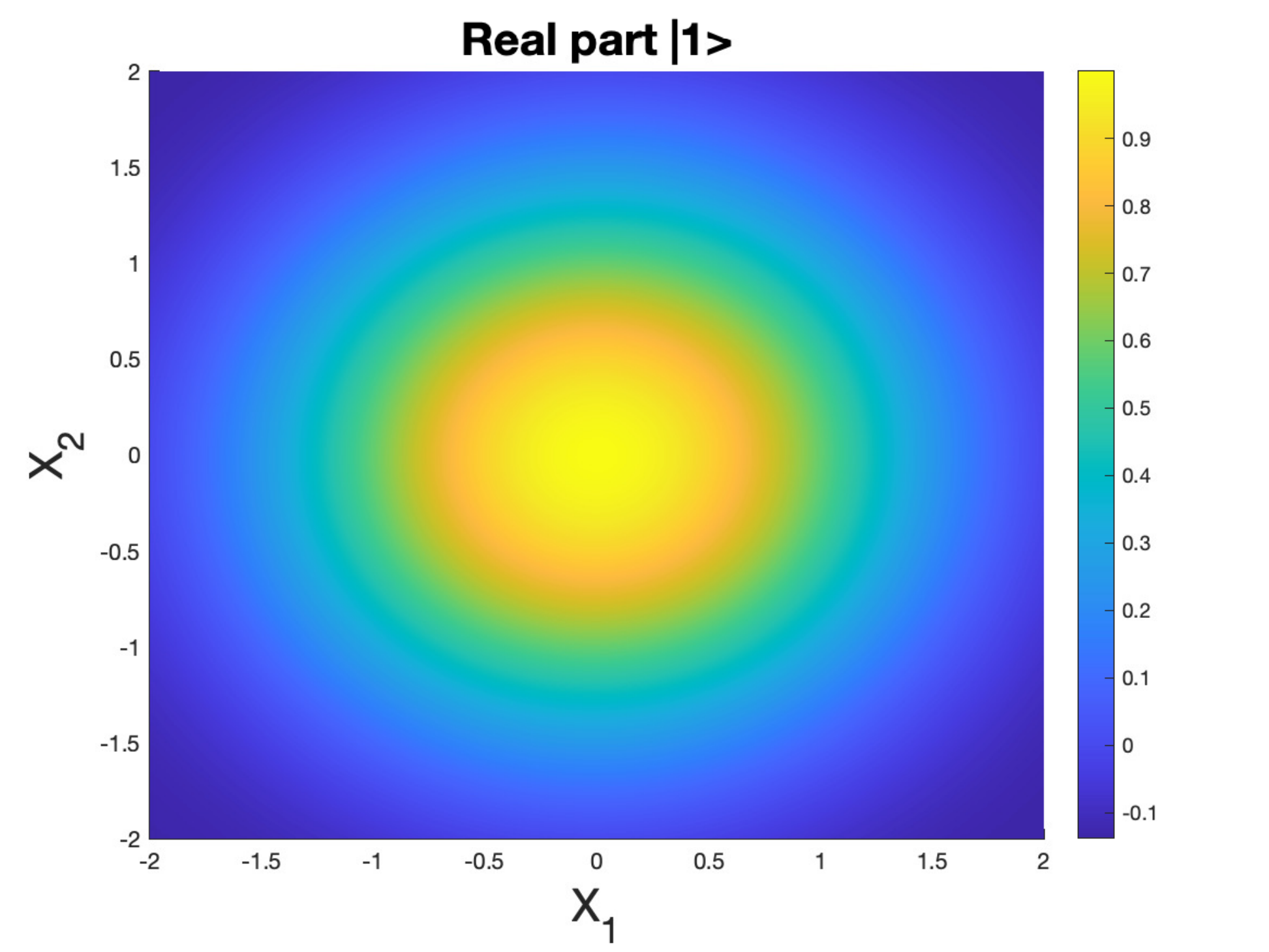} 
\includegraphics[width=4cm]{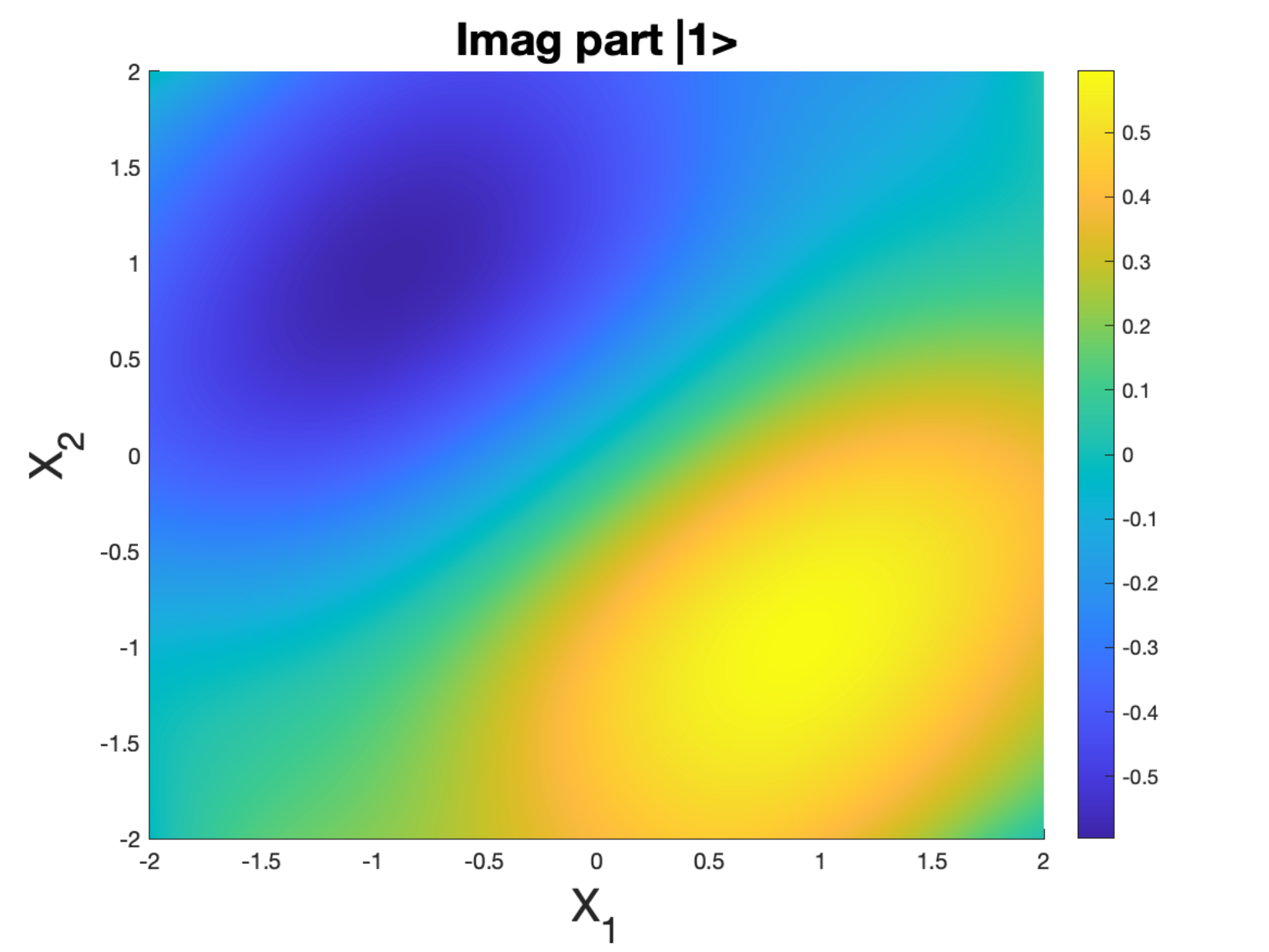}
\includegraphics[width=4cm]{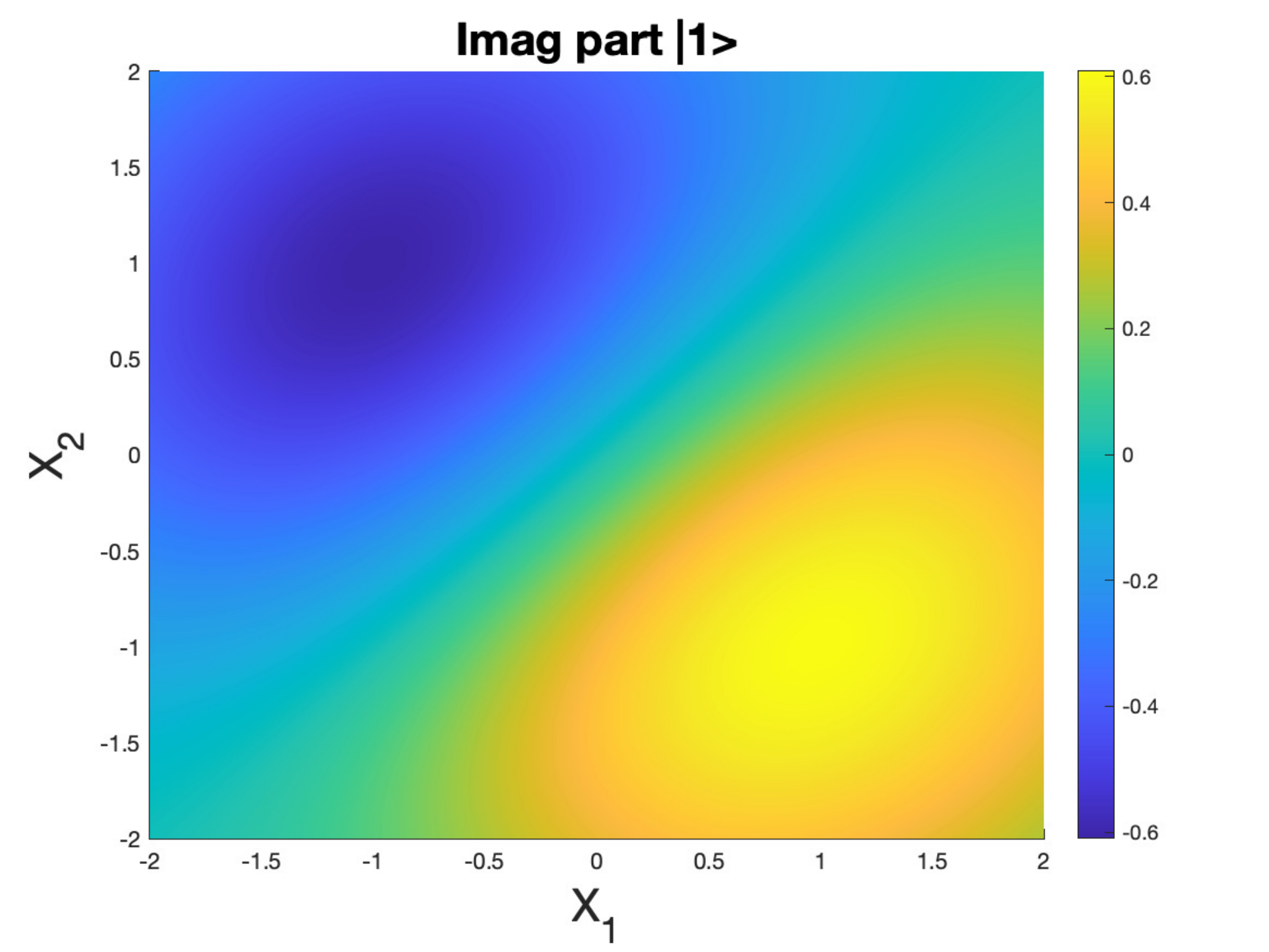}
\caption{Real and imaginary part of the quantum characteristic function of $\ket 0_{\operatorname{Cat}}$(top) with $\alpha = 10\times (1,1)$ and $\ket 1_{\operatorname{Cat}}$ (bottom), with $\alpha = \frac{1}{10}\times (1,1)$, respectively reconstructed with $N=200$ with reconstructed function (left) and true one (right)) on $[-2,2]$.}
\end{figure}

\begin{figure}[h!]
\includegraphics[width=4.2cm]{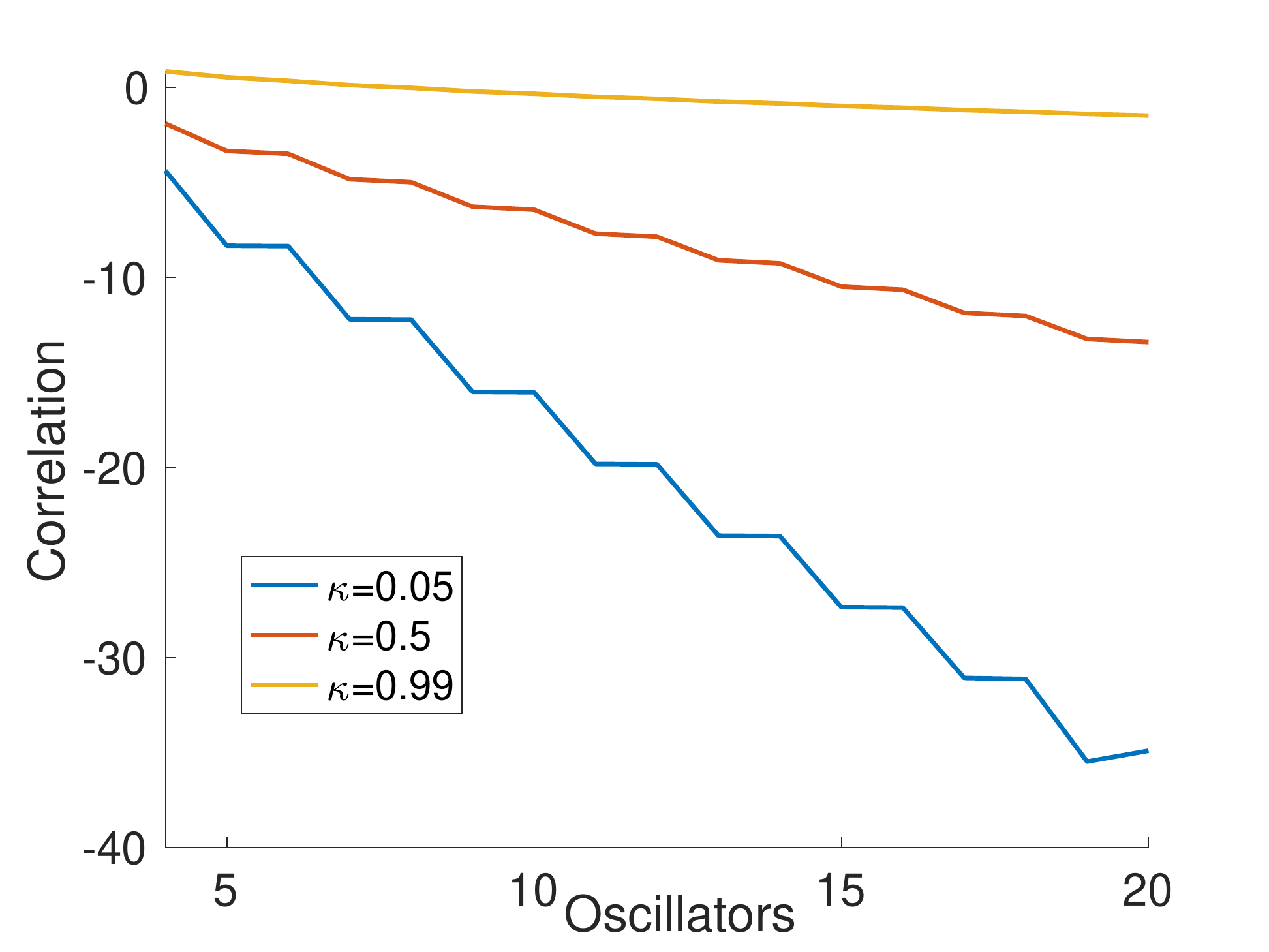} 
\includegraphics[width=4.2cm]{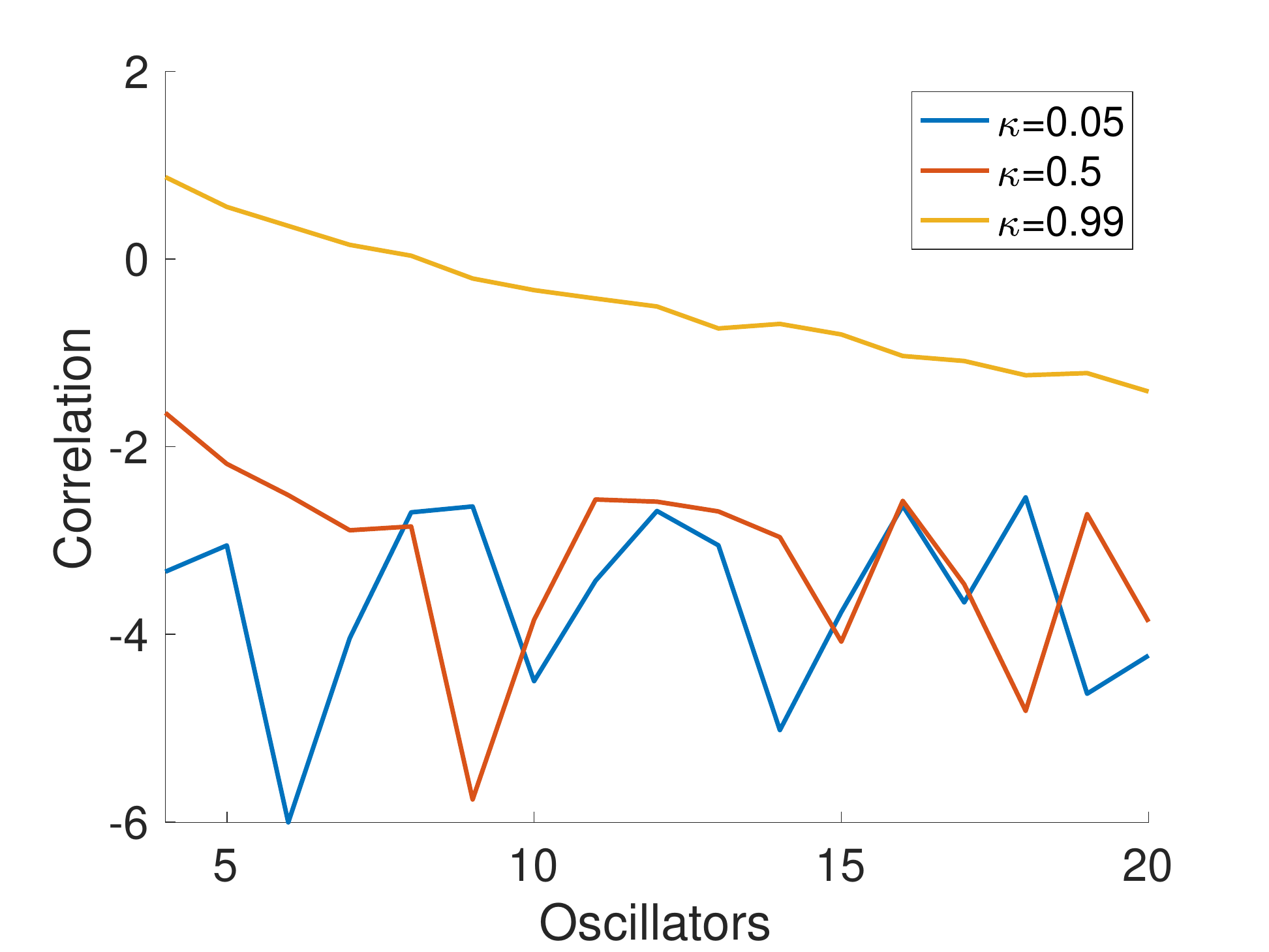}\\
\includegraphics[width=4.2cm]{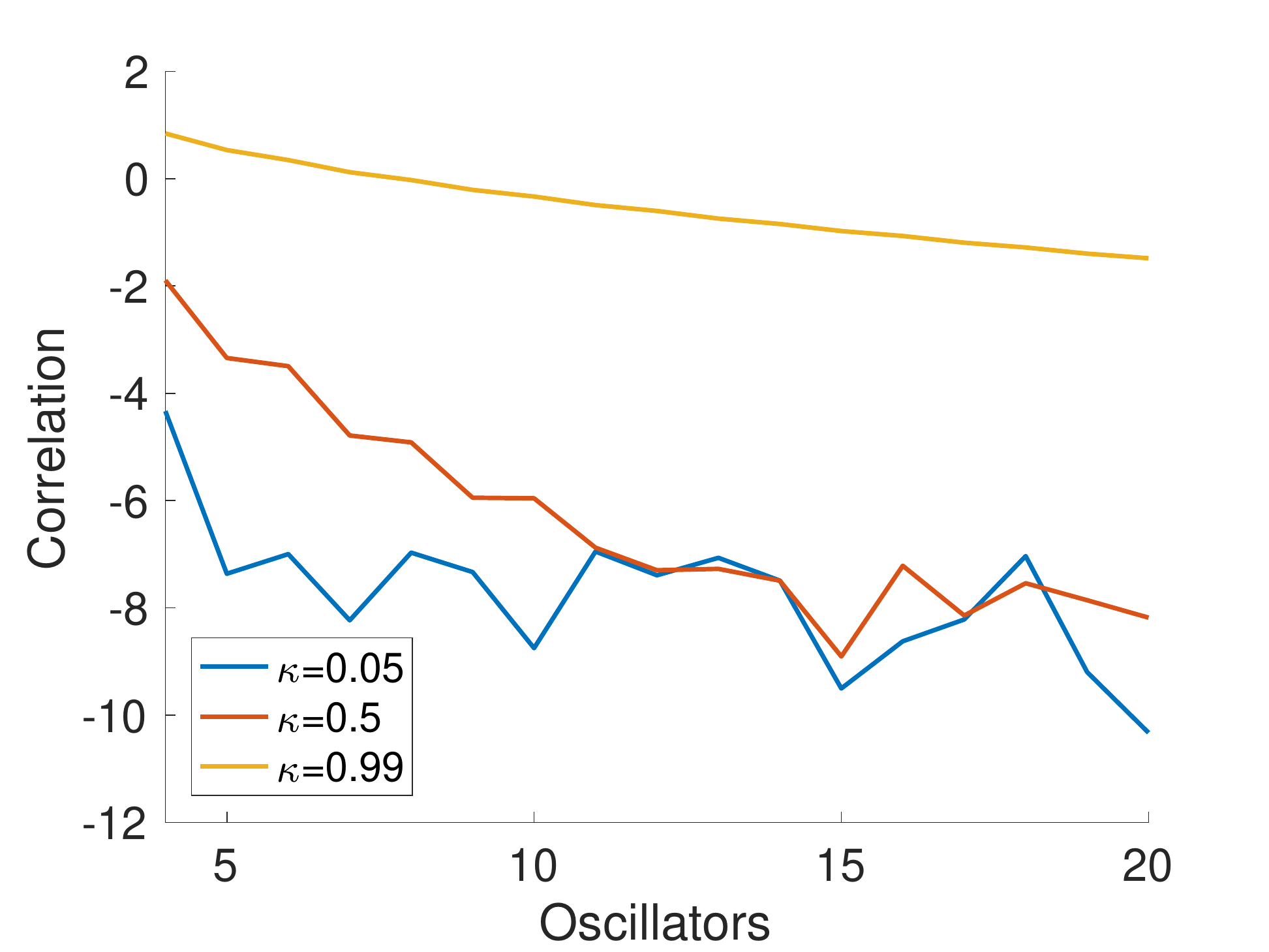}
\includegraphics[width=4.2cm]{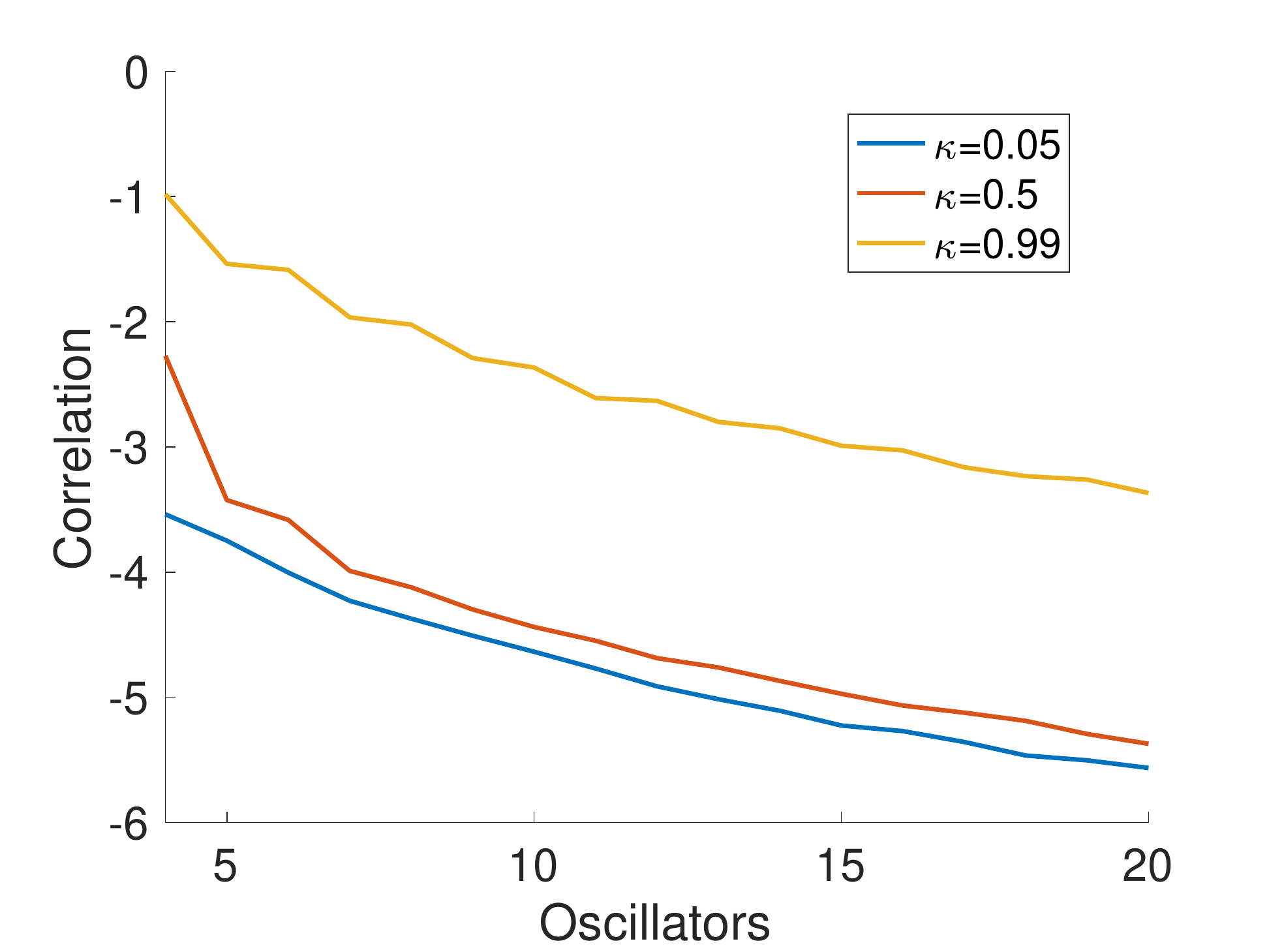}\\
\includegraphics[width=4.2cm]{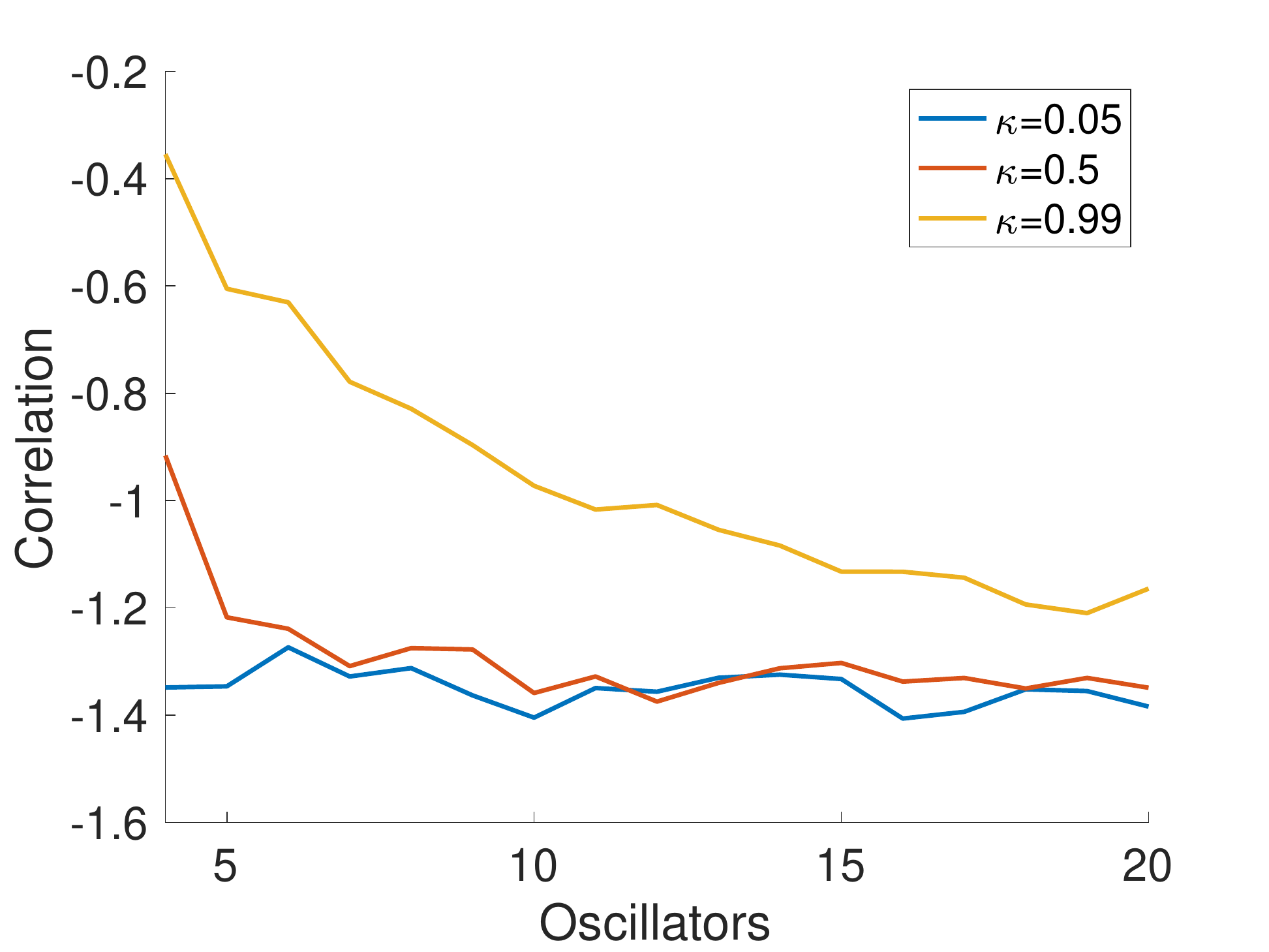}
\includegraphics[width=4.2cm]{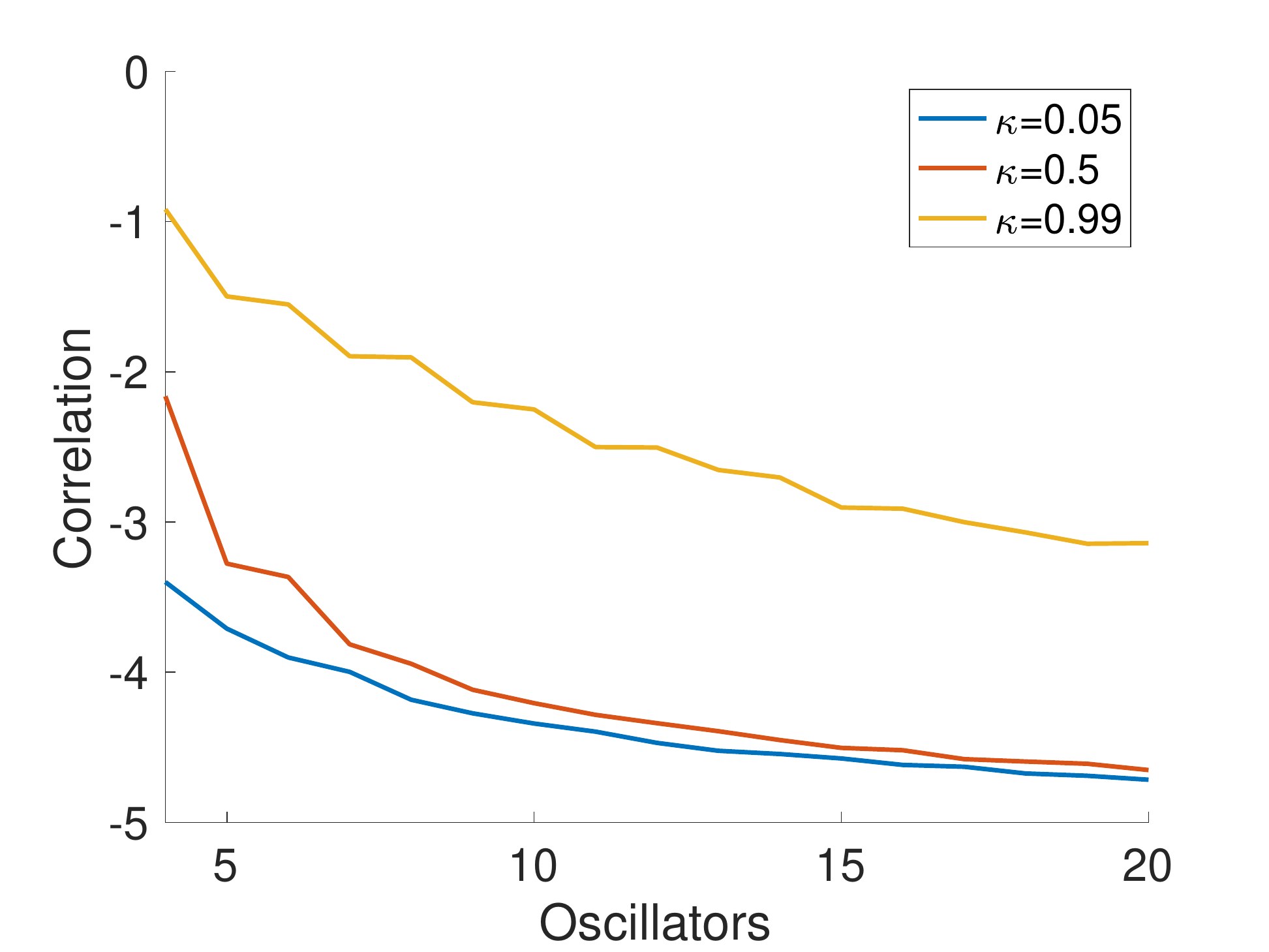}
\caption{On the left, we see the correlation coefficient $\log(\gamma_{1,\lceil m/2 \rceil})$ (in the bottom row with $h_{XX}$ replaced by $h_{XX}+\frac{1}{2m}Q^2$, where $Q$ is a GUE Matrix and $m$ is the number of oscillators). Here $\kappa$ refers to the interaction strength between neighboring oscillators. In the center and on the right, we see the same correlation coefficient recovered from the shadow with 1000 and 100,000 numerical experiments sampled from~\eqref{eq:gaussian}.   }\label{fig:decay_corr}
\end{figure}

\begin{figure}
    \centering
    \includegraphics[width=4.2cm]{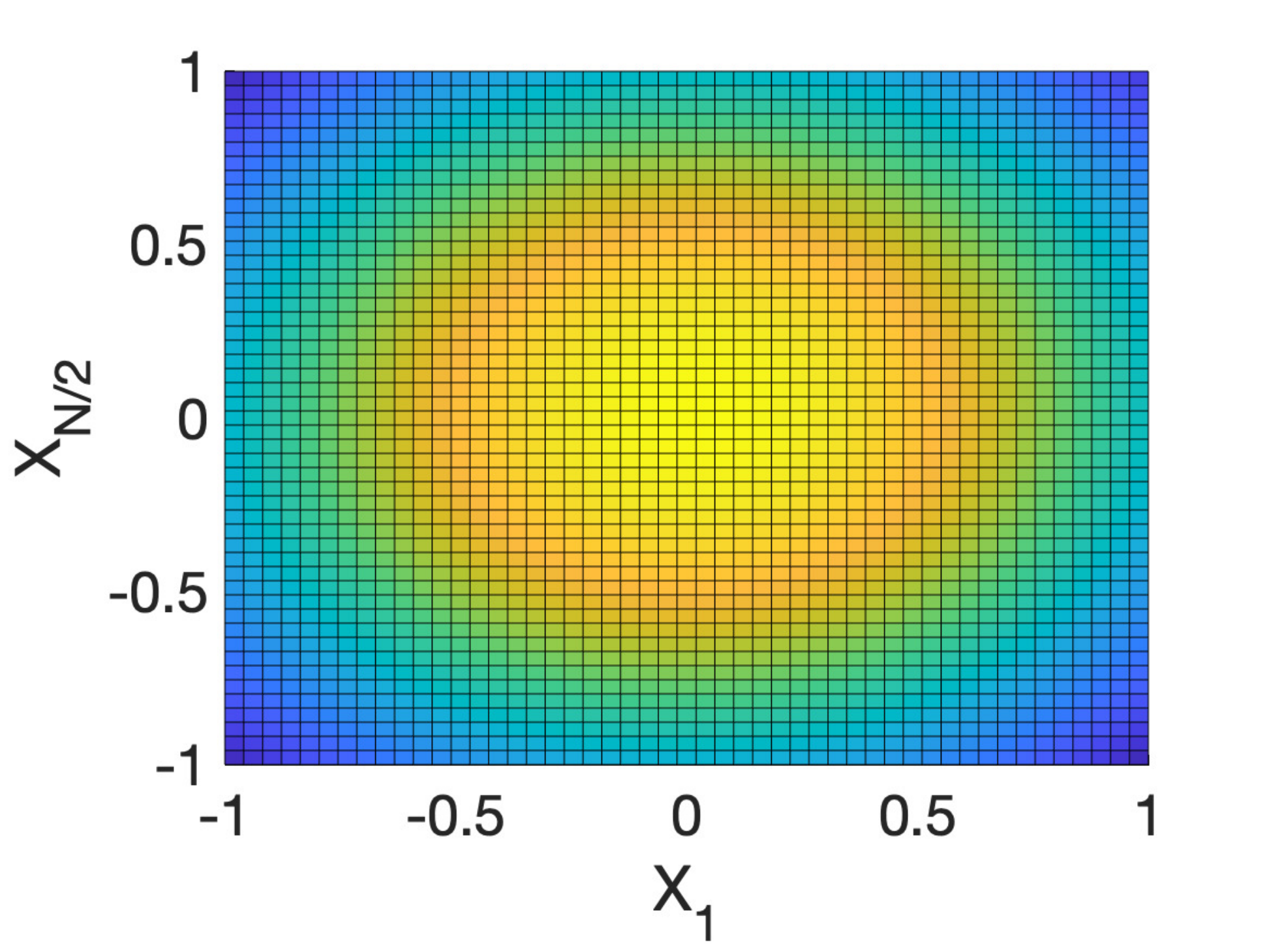}
    \includegraphics[width=4.2cm]{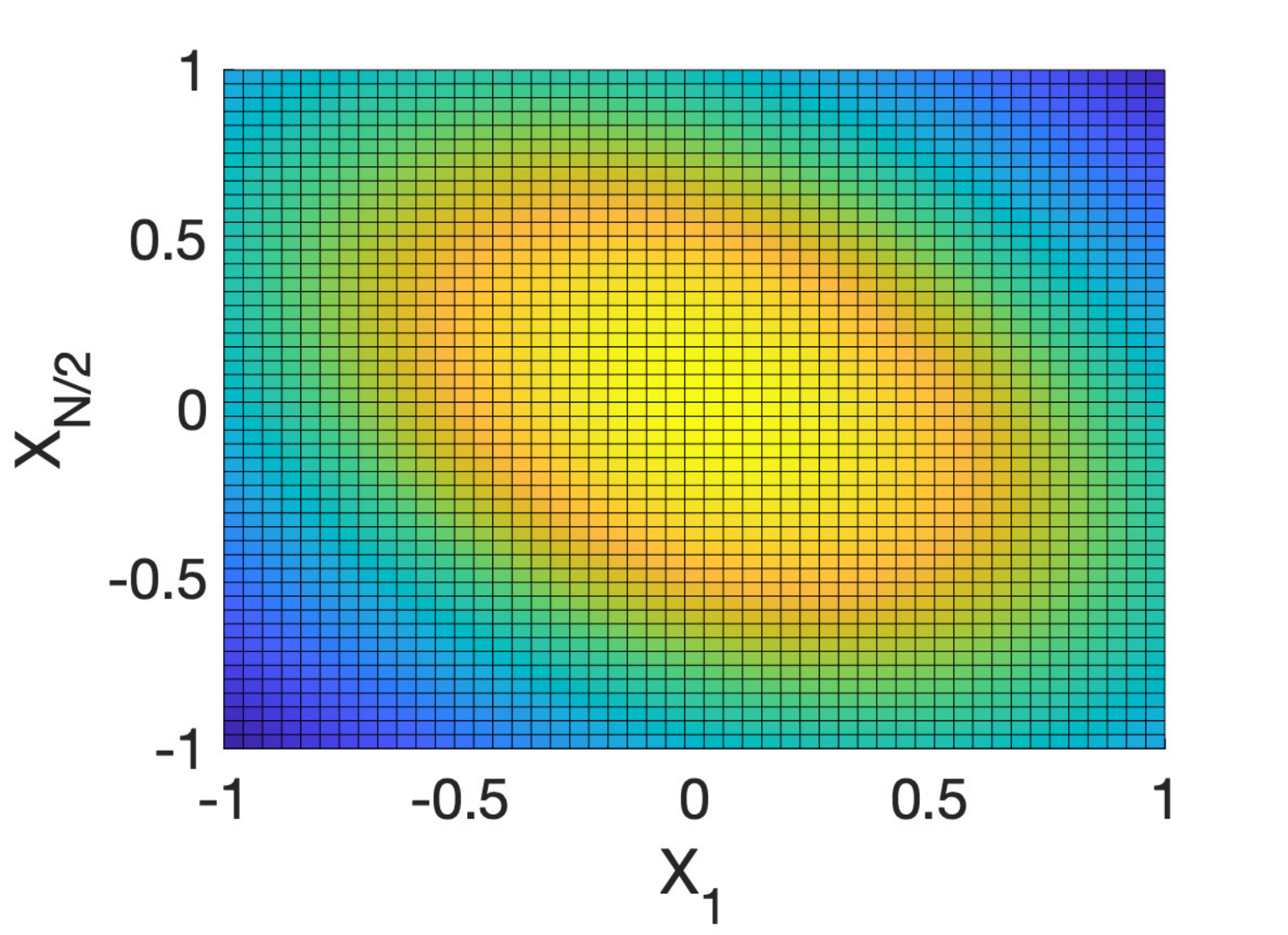}\\
    \includegraphics[width=4.2cm]{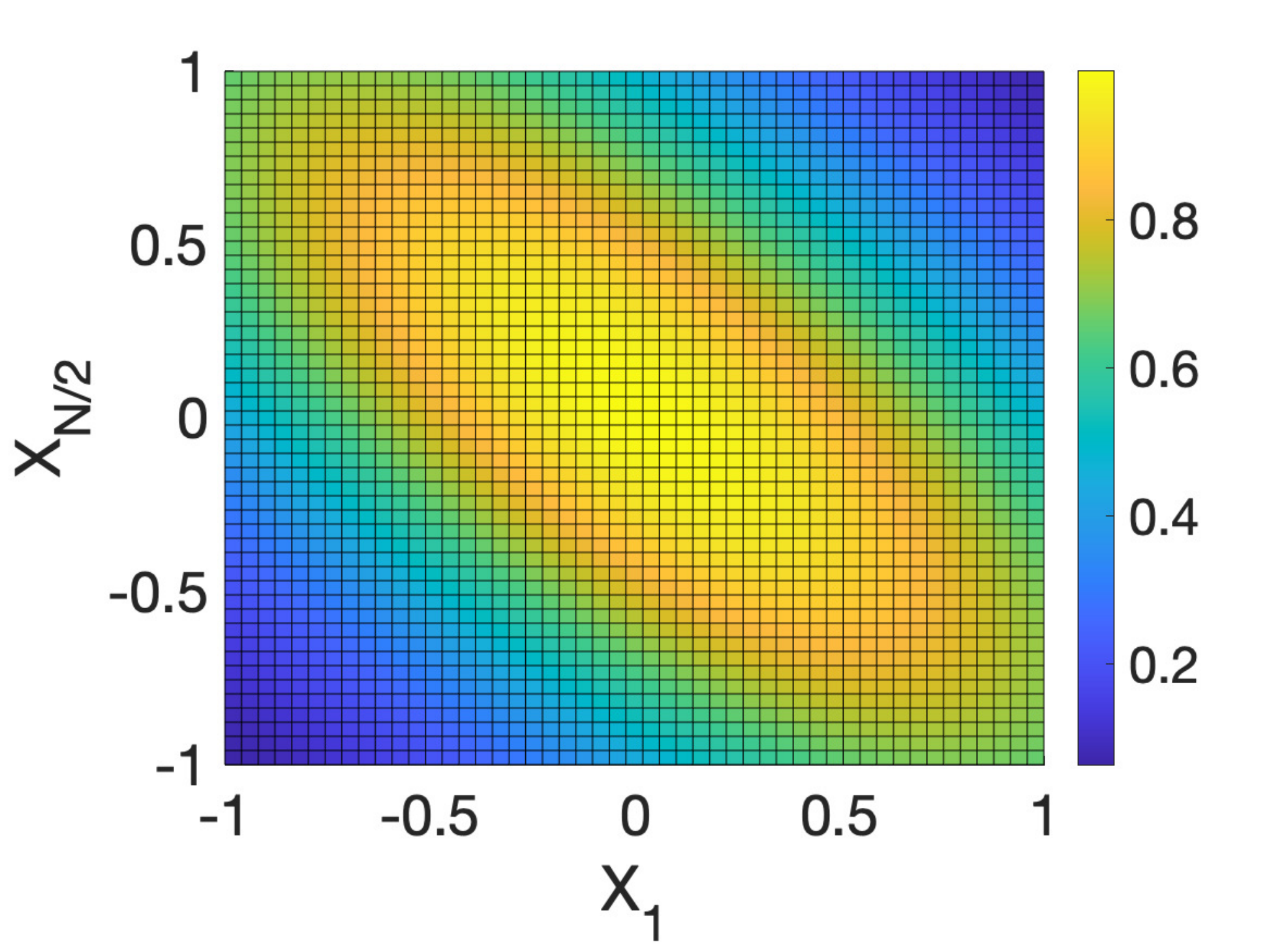}
     \includegraphics[width=4.2cm]{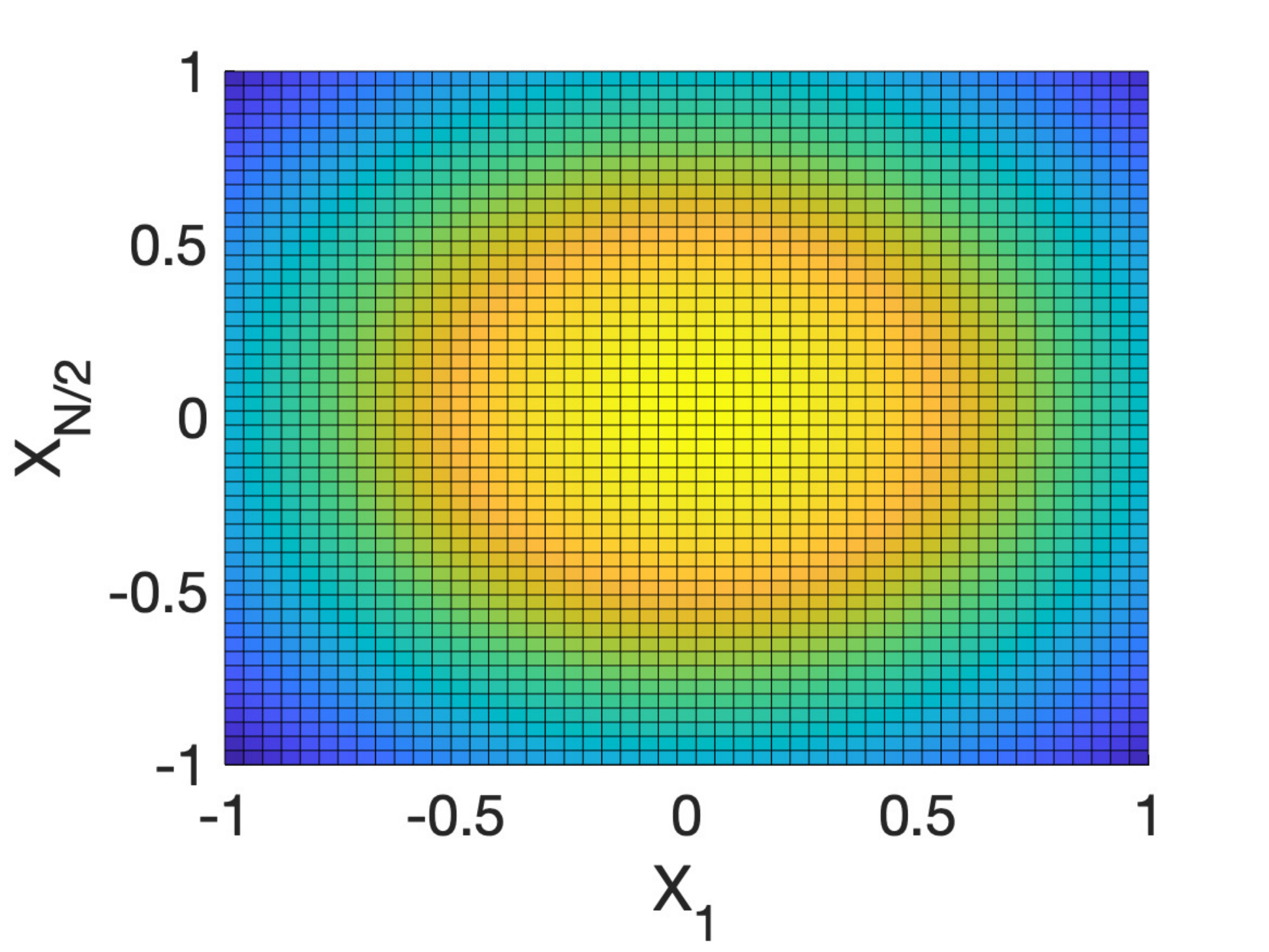}\\
    \includegraphics[width=4.2cm]{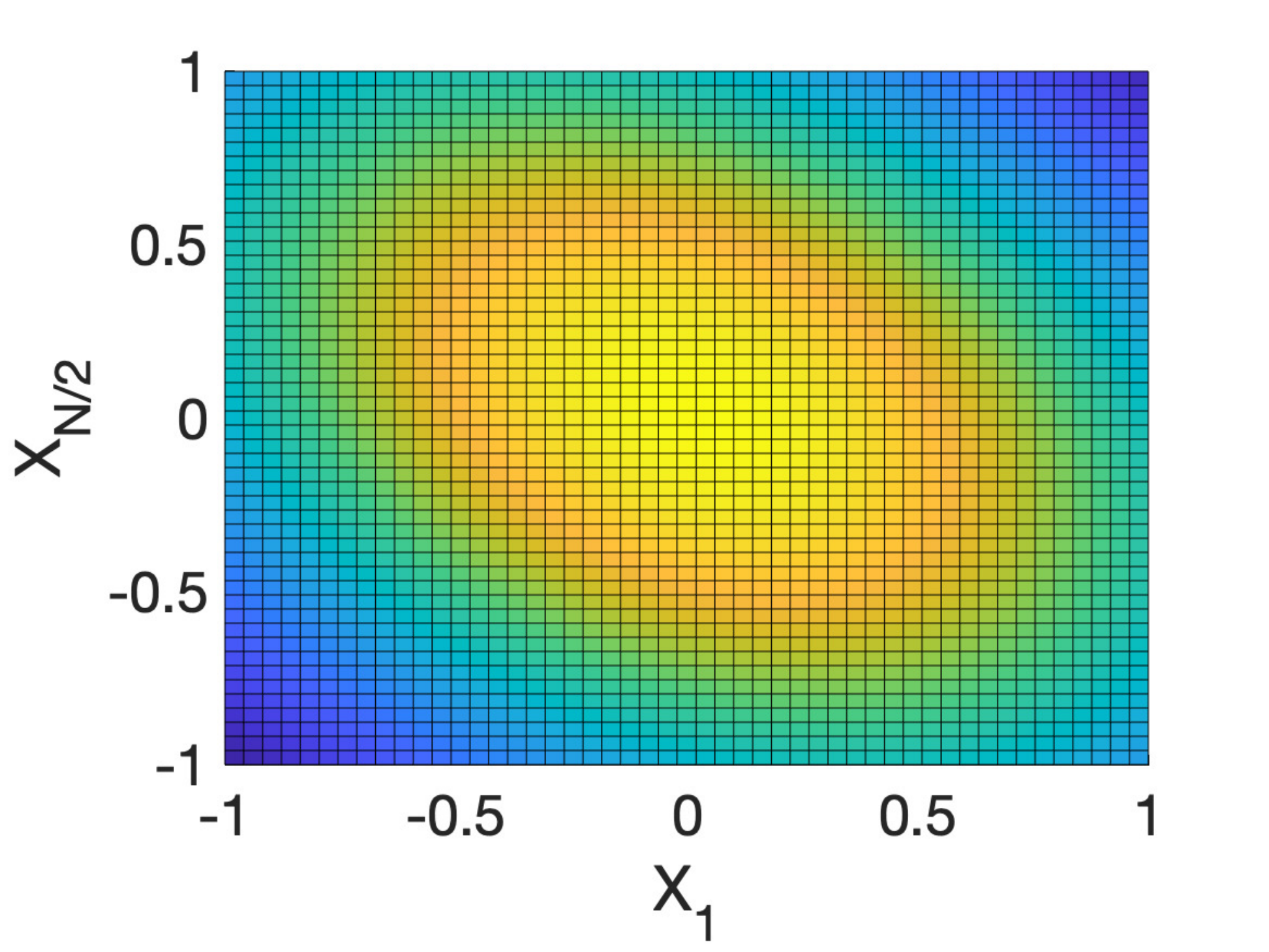}
    \includegraphics[width=4.2cm]{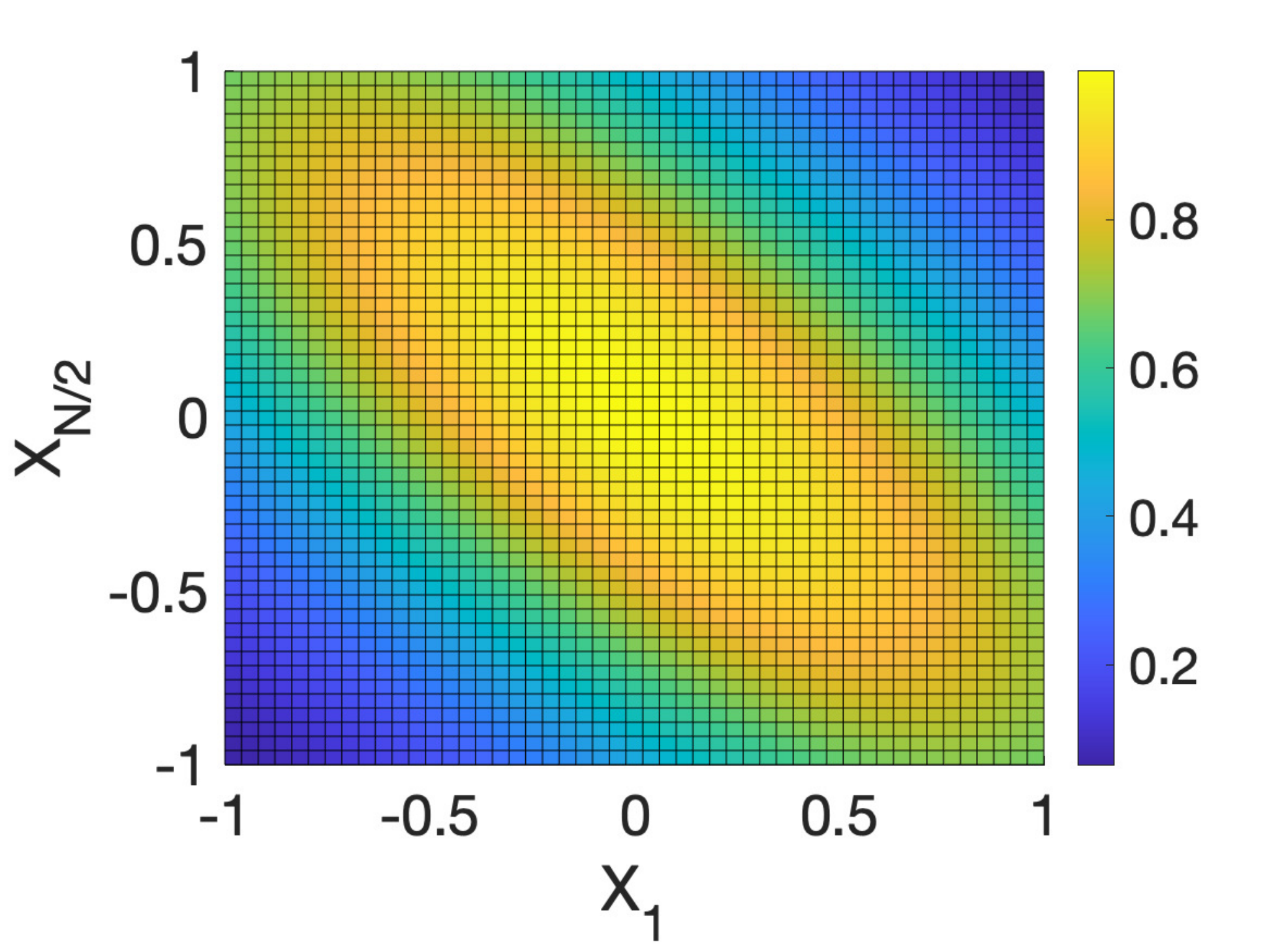}
    \caption{Plot of $\chi_{1,N/2}((x_1,0),(x_{N/2},0))$ with 1000 \emph{numerical measurements} sampled from the multivariate Gaussian distribution~\eqref{eq:gaussian}. From left to right we see simulations for $(\kappa,N) = (0.99,1000), (0.99,10),(0.99,4)$ respectively, with recovered characteristic functions in the top row and true ones in the bottom row. We see that correlations between particles are almost perfectly reproduced.}
    \label{fig:my_label}
\end{figure}
\subsection{Vacuum state}
We start with the simplest example $\rho =\ketbra{0}$ to fix ideas. In this case $\vert \langle 0 | \alpha\rangle  \vert^2 = e^{-\vert \alpha \vert^2/2}$.
Measurements of $\alpha$ are then standard normally distributed $\alpha \sim \mathcal N_{\mathbb R^2}(0,1).$
Considering then $N=50$ experimental realisations, $\alpha_1,...,\alpha_{N} \in \mathbb R^2,$ the trial quantum characteristic function is 
\[ \chi_{\sigma_N}(u)= \frac{e^{\frac{\vert u \vert^2}{4}}}{N}\sum_{i=1}^N e^{-i  u^T \Omega \alpha_i }.\]
The function is illustrated in Figure~\ref{fig:approx}. Reduced particle characteristic functions for up to $1000$ oscillators are illustrated in Figure ~\ref{fig:my_label}. This yields an approximation of the quantum characteristic function on compact sets around zero. Due to the exponentially increasing function $e^{\frac{\vert u \vert^2}{4}}$, this approximation is only local for fixed $N$. We also plot the variance of the true characteristic function $\chi$ vs the reconstructed characteristic function $\tilde \chi$ at $N$ grid points $(x_n)_{i=1}^N$ with $x_n \in D$ in some domain $D$
\[ \mathbf V_{N,D} = \frac{\sum_{i=1}^{N} \vert \chi_i(x_n)-\tilde \chi_i(x_n)\vert^2}{\operatorname{vol}(D) N}.\]

\subsection{Cat qubit states}
For our second example, we consider a non-Gaussian, one-mode pure quantum state called cat state. These states are used in quantum error correction~\cite{mirrahimi2014dynamically,Mazyar2019}. Given a coherent state $|\alpha\rangle$, we denote the cat states 
\begin{align*}
    |+\rangle_{\operatorname{Cat}}=\frac{|\alpha\rangle+|-\alpha\rangle}{\mathcal{N}_+},\quad     |-\rangle_{\operatorname{Cat}}=\frac{|\alpha\rangle-|-\alpha\rangle}{\mathcal{N}_-}\,,
\end{align*}
with normalisation constants \begin{align}
    & \mathcal{N}_{\pm}\coloneqq\sqrt{ (\bra{\alpha} \pm \bra{-\alpha})(\ket{\alpha} \pm \ket{-\alpha})} \nonumber\\
    &= \sqrt{2 \pm (\langle \alpha\vert -\alpha \rangle+\langle-\alpha \vert \alpha \rangle)  } = \sqrt{2(1\pm e^{-2\vert \alpha\vert^2})}.
    \end{align}
Here, $\langle x \vert \pm  \rangle_{\operatorname{Cat}}=\frac{\langle x \vert \alpha \rangle \pm \langle x \vert -\alpha \rangle}{\mathcal{N}_{\pm}},$ where $\langle x \vert \alpha \rangle = e^{-\frac{\vert x \vert^2 + \vert \alpha \vert^2}{4}+\frac{\alpha \bar x}{2}}$ such that $\langle -\alpha \vert \alpha \rangle = \langle \alpha \vert -\alpha \rangle=e^{-\vert \alpha\vert^2}.$ This way, we can define $$|0\rangle_{\operatorname{Cat}}\coloneqq\frac{|+\rangle_{\operatorname{Cat}}+|-\rangle_{\operatorname{Cat}}}{\sqrt{2}} \text{ and }
 |1\rangle_{\operatorname{Cat}}\coloneqq\frac{|+\rangle_{\operatorname{Cat}}-|-\rangle_{\operatorname{Cat}}}{\sqrt{2}}.$$ For $N=100$ we are able to reconstruct the characteristic function of states $\ket 0_{\operatorname{Cat}} $ and $\ket 1_{\operatorname{Cat}}$ respectively on $[-2,2]$, for $\alpha=(1,1)$, which is illustrated in Figure~\ref{fig:qcf}.

 The probability distribution is then 
 %\[\begin{split}
 \begin{align}\big\vert \langle x \vert 0 \rangle_{\operatorname{Cat}} \big\vert^2 = \frac{\vert \langle x \vert + \rangle \vert^2 +\vert \langle x \vert - \rangle \vert^2 +2 \Re ( \overline{\langle x \vert - \rangle} \langle x \vert + \rangle) }{2}
 \end{align}
  and
  \begin{align}
 \big\vert \langle x \vert 1 \rangle_{\operatorname{Cat}} \big\vert^2 = \frac{\vert \langle x \vert + \rangle \vert^2 +\vert \langle x \vert - \rangle \vert^2 -2 \Re ( \overline{\langle x \vert - \rangle} \langle x \vert + \rangle) }{2}
 \end{align}
 %\end{split}. \]

This implies that 
%\[\begin{split} 
\begin{align}
&\chi_{\ket 0_{\operatorname{Cat}} \bra 0_{\operatorname{Cat}}} \nonumber\\
&= \frac12 \left(\chi_{\ket + \bra +}+ \chi_{\ket + \bra -} + \chi_{\ket - \bra +} +\chi_{\ket - \bra -}\right)\\
&= \frac12 %\left(\Bigg
(\frac{1}{\mathcal N_+^2}+\frac{1}{\mathcal N_-^2}\Bigg) (\chi_{\ket \alpha \bra \alpha}+ \chi_{\ket {-\alpha} \bra {-\alpha}}) \nonumber\\
& \quad + \frac12 \Bigg(\frac{1}{\mathcal N_+^2}-\frac{1}{\mathcal N_-^2}\Bigg) (\chi_{\ket \alpha \bra {-\alpha}}+ \chi_{\ket {-\alpha} \bra \alpha})\nonumber\\
&\quad + \frac{\chi_{\ket \alpha \bra \alpha}- \chi_{\ket {-\alpha} \bra {-\alpha}}}{\mathcal N_+ \mathcal N_-} 
\end{align}
%\end{split}
%\]
and 
% \[\begin{split} 
\begin{align}& \chi_{\ket 1_{\operatorname{Cat}} \bra 1_{\operatorname{Cat}}} \nonumber\\
&= \frac12 \left(\chi_{\ket + \bra +} -\chi_{\ket + \bra -} - \chi_{\ket - \bra +} +\chi_{\ket - \bra -}\right)\\
&= \frac12 %\left(
\Bigg(\frac{1}{\mathcal N_+^2}+\frac{1}{\mathcal N_-^2}\Bigg) (\chi_{\ket \alpha \bra \alpha}+ \chi_{\ket {-\alpha} \bra {-\alpha}}) \nonumber\\
& \quad +\frac12 \Bigg(\frac{1}{\mathcal N_+^2}-\frac{1}{\mathcal N_-^2}\Bigg) (\chi_{\ket \alpha \bra {-\alpha}}+ \chi_{\ket {- \alpha} \bra \alpha}) \nonumber\\
& -\frac{\chi_{\ket \alpha \bra \alpha}- \chi_{\ket {-\alpha} \bra {-\alpha}}}{\mathcal N_+ \mathcal N_-} .
\end{align}
%\end{split}\]
The characteristic functions on the right hand side are explicitly given by~\eqref{eq:char_func2}.

\subsection{Chain of quadratic harmonic oscillators}
We consider a hermitian Hamiltonian matrix $ h = \begin{pmatrix} h_{XX} &h_{XP} \\ h_{PX} & h_{PP} \end{pmatrix}$
to which we associate a quadratic chain with Hamiltonian
\[ H =\langle R, h R \rangle \text{ with }R=(X,P), \]
where $h_{XP}=h_{PX}=0$, $h_{PP}= \frac{\operatorname{id}}{2}$ and  $$h_{XX}=\begin{pmatrix}
1/2      & -\frac{\kappa}{4}  & \cdots  & 0    & -\frac{\kappa}{4}     \\
-\frac{\kappa}{4}     & 1/2     & -\frac{\kappa}{4}  &         & 0    \\
\vdots   & -\frac{\kappa}{4}    & 1/2     & \ddots  & \vdots  \\
0  &         & \ddots  & \ddots  & -\frac{\kappa}{4}  \\
-\frac{\kappa}{4}  & 0 & \cdots  & -\frac{\kappa}{4}      & 1/2     \\
\end{pmatrix}$$ for some $\kappa \in [-1,1].$ Here, $\kappa$ refers to the interaction strength between neighboring oscillators: small $\kappa$ means weak correlations whereas large $\kappa$ means strong correlations. Let 
\[ X = h_{XX}^{-1/2} \sqrt{\sqrt{h_{XX}} h_{PP}\sqrt{h_{XX}}} h_{XX}^{-1/2} .\]
The ground state $|\psi\rangle$ of $H$ is a Gaussian state with covariance matrix given by $\gamma= \operatorname{diag}(X,X^{-1})$~\cite{schuch2006quantum}. The decay of spatial correlations is illustrated in Figure~\ref{fig:decay_corr}. We will again use our heterodyne tomography protocol to reconstruct that decay. The distribution of the random variable arising from the heterodyne detection takes the following expression:
\begin{equation}
\label{eq:gaussian}
    |\langle x|\psi\rangle|^2\propto e^{- x^T\Omega^T\Big(\frac{I+\gamma}{2}\Big)^{-1}\Omega x}\,.
\end{equation}
We recall that for any two modes $i,j$, the reduced shadow characteristic function after sampling $(x_k)_{k=1}^N$ from the above distribution is defined as
\begin{align*}
    &\hat{\chi}(u_i,u_j)\coloneqq \frac{1}{N}\sum_{k=1}^N e^{  \frac14(\|u_i\|^2+ \|u_j\|^2)-i(u_i^\intercal \Omega (x_k)_i+u_j^\intercal \Omega (x_k)_j)}\,,
\end{align*}
where $(x_k)_i,(x_k)_j\in\mathbb{R}^2$ are the components of vector $x_k\in\mathbb{R}^{2m}$ corresponding to modes $i$ and $j$. The reduced characteristic functions of the marginal of $|\psi\rangle$ over modes $i$ and $j$ takes the form
\begin{align*}
    \chi(u_i,u_j)=e^{-\frac{1}{4}(u_i,u_j)^\intercal \gamma (u_i,u_j)}
\end{align*}
From which the correlations coefficients corresponding to these two modes can be read-off.

\section{Acknowledgements}

LL acknowledges financial support from the Alexander von Humboldt Foundation.
The authors are also grateful to Pembroke College for funding a workshop in Cambridge in October 2022 where we had some useful discussions.

 \subsection{Application of shadow tomography to photonic crystals and 2D materials}
 Photonic crystals are dielectric solids with a periodically modulated refraction index. One of their many exciting features is that the photonic bandstructure in photonic crystals, which is accessible to optical measurements, provides insights into the electronic bandstructure in periodic solids. Just like electrons in solids, photons are prohibited from propagating at band-gap frequencies inside the medium. While the band structure of photonic crystals is obtained from the classical Maxwell equations, the filling of bands by photons relies on a quantum mechanical description of the electromagnetic field. Such a description is also necessary to understand effects such as spontaneous emission of photons in photonic crystals~\cite{PhysRevLett.84.4341,PhysRevLett.64.2418}. We shall demonstrate now how shadow tomography can be applied to exhibit the dispersion surface of the Bloch-Floquet bands. Our description here is semiclassical, as we treat the photons quantum-mechanically and the crystallic band structure classically.

 In our discussion of applications of shadow tomography, we will then focus on optical analogues of graphene called \emph{photonic graphene}

 We start by using classical electrodynamics to describe the propagation of electromagnetic waves in photonic crystals. We will then quantise the electromagnetic fields to convert the classical picture to a quantum picture involving photons and illustrate how the methods developed in this article can be used in photonic crystals. This semiclassical approach is used to simplify the analysis and get explicit formulas for photonic states and electromagnetic fields.

 \subsubsection{Derivation of the Helmholtz equation}
A mathematical account of the electromagnetic structure of photonic crystals can be found in~\cite{Photon_crystals}.
We shall start by arguing that the propagation of TE polarised light, with $E= (E_1,E_2,0)$ and $H=H_3 \hat{e}_3$, where $\hat{e}_i$ is the $i$-th unit vector, in two dimensional photonic crystals can be reduced to a periodic eigenvalue problem with Helmholtz operator $\mathcal L = \nabla \cdot D_e \nabla$ 
 \begin{equation}
 \label{eq:Bloch-Floquet}
 \mathcal L \psi= \omega^2 \psi
 \end{equation}
 with $\psi=H_3.$ Here, $D_e = \frac{J^T \kappa J}{\mu_3}$ is a periodic function, with $J=\begin{pmatrix} 0 & -1 \\ 1& 0 \end{pmatrix},$ that only depends on the permittivity and permeability of the material and will be specified below. 
 We start from the Maxwell equations
 \[ \begin{split}
     \nabla \times E &= -\partial_t B,\  \nabla \times H = \partial_t D, \\
     \nabla \cdot D &= 0, \  \nabla \cdot B=0. 
 \end{split}\]
 We then find that Maxwell's equation can be written as a Schr\"odinger equation with 
 \[ i \partial_t \Psi =M \Psi \]
 with $\Psi = (E,H)^t$ and 
 \[ M = R^{-1} \begin{pmatrix} 0 &i \nabla \times \\ -i \nabla \times & 0   \end{pmatrix}.\]
 Here, $R= \begin{pmatrix} \varepsilon & \zeta \\ \zeta^* & \mu\end{pmatrix}$ with permittivity tensor $\varepsilon$, permeability tensor $\mu$, and bianisotropy tensor $\zeta.$

 For photonic analogues of 2D material, we shall only consider matrix entries of $R$ that vary in the two-dimensional plane. In addition, we assume that the coupling between the longitudinal and transversal direction is zero. 

 This way, we may decouple the $6 \times 6$ Maxwell equations into two coupled equations with matrices
 \[ R_{\perp} = \begin{pmatrix} \varepsilon_{\perp} & \xi_{\perp} \\ \xi_{\perp} & \mu_{\perp} \end{pmatrix} \text{ and } R_{||} = \operatorname{diag}(\varepsilon_{||},\mu_{||})\]
 that are the form
 \begin{equation}
 \begin{split}
 \label{eq:EM_rel}
 \partial_t(E_{\perp},H_{\perp}) &= R_{\perp}^{-1}(-J \nabla_{\perp} H_3,J \nabla_{\perp} E_3)  \text{ and }\\
  \partial_t(E_3,H_3) &= R_{||}^{-1}(-J \nabla_{\perp} \cdot H_3,J \nabla_{\perp} \cdot  E_3).\end{split}
  \end{equation}
 Here $\nabla_{\perp} = (\partial_{x_1}, \partial_{x_2}).$ Eliminating $(E_{\perp},H_{\perp})$ and using $R_{\perp}^{-1} =: \begin{pmatrix} \kappa& \zeta \\ \zeta^* & \iota \end{pmatrix},$ we find under the assumption of pure TE waves, that is $E_3=0$, the wave equation
 \[  \partial_t^2 H_3(x_1,x_2,t) = -\mathcal L H_3(x_1,x_2,t).\]
 Observe that by~\eqref{eq:EM_rel},  $(E_{\perp},H_{\perp})$ are directly determined from $H_3.$
 By looking at planar waves $H_3(x_1,x_2,t)= H_3(x_1,x_2) e^{\pm i \omega t},$
 we reduce the wave equation to~\eqref{eq:Bloch-Floquet}.

 In magneto-optic materials, one has 
 \[ R_{\perp} =\operatorname{diag}(\varepsilon \sigma_0-\gamma \sigma_2,\mu \sigma_0 ),\]
 where $\gamma$, with $\vert \gamma \vert<\varepsilon,$ is the strength of the Faraday-rotation. This way, $\kappa=\frac{\varepsilon}{\varepsilon^2-\gamma^2} \sigma_0+\frac{\mu}{\varepsilon^2-\gamma^2}  \sigma_2.$

 \subsubsection{Bloch-Floquet theory and band structure of photonic crystals}
 Photonic crystals are periodic with respect to some lattice $\Gamma.$ This is reflected in the periodicity of the operator $\mathcal L$. We can thus apply the standard Gelfand transform and find that that $\mathcal L$ is unitarily equivalent to the direct integral operator 
 \[ \int_{\mathbb C/\Gamma}^{\oplus} \mathcal L(k) \frac{dk}{\vert \mathbb C/\Gamma \vert}
\text{ with }\mathcal L(k)=-(\nabla +ik)\cdot D_e(\nabla +ik)\] for which the eigenvalue problem ~\eqref{eq:Bloch-Floquet} converts into 
 \[ \mathcal L(k)H_3(x,k)=\omega^2(k)H_3(x,k), \quad H_3(x+\Gamma,k)=H_3(x,k) \] 
 for $k \in \CC/\Gamma^*$, the Brillouin zone. By standard arguments, the set of admissible $\omega$ for given $k$ is discrete $(\omega_n(k)^2)_{n \in \mathbb{N}}$ with associated fields $(H_{3,n}(x,k)).$ 

The function $\omega(k)^2$ is also called the dispersion relation and $\partial_k \omega(k)$ is the speed of propagation.

 Thus, any eigenfunction of $\mathcal L(k)$ has a Fourier expansion 
 \[ \phi(x) = \sum_{m \in \mathbb{Z}^2} \phi_m e^{i(m_1k_1+m_2k_2)\cdot x},\]
 where $\Gamma^* = k_1 \mathbb Z +k_2 \mathbb Z.$

 \subsubsection{The quantisation of the EM field}
In order to study quantum effects such as spontaneous emission, it is essential to use a field theoretic description of the electromagnetic field inside a photonic crystal\cite{PhysRevLett.84.4341,PhysRevLett.64.2418}. 

The filling of bands in photonic crystals by photons can be described by quantizing the magnetic field and consider field operators
\[\begin{split}
 \widehat{H_{3}}(x,k,t) &= \sum_{n \in \mathbb N} h_{n}(k)\Big(H_{3,n}(x,k) a_{n}(k)e^{-i\omega_n(k)t}\\
 &+\overline{H_{3,n}(x,k)} a_{n}(k)^*e^{i\omega_n(k)t}\Big)\end{split},\]
 where $h_{n}(k) = \frac{\hbar \omega_n(k)}{2\mu \vert \mathbb C/\Gamma \vert}$ and $a_n(k), a_n(k)^*$ are annihilation and creation operators of photons with wavevector $k$ in band $n$. The macroscopic limit is then obtained by taking the expectation value of the field vector $\langle \widehat{H_{3}}(x,k,t) \rangle$ with respect to the photonic state. 

\subsubsection{Dirac points}
As a simple example, we shall consider the case that $D_e$, appearing in~\eqref{eq:Bloch-Floquet}, has honeycomb-lattice symmetries, i.e. $\Gamma $ is an equilateral triangular lattice with $2\pi/3$-rotational symmetry  In this case, it is well-known, see e.g.~\cite{Photon_crystals}, that the spectrum of $\mathcal L(k)$ at suitable energies can be effectively described, close to some positive energy $E_D>0$, by a two-dimensional Dirac operator 
\[ H(k) = \begin{pmatrix}0 & 2D_{\bar z}+k \\ 2D_{z}+\bar k & 0  \end{pmatrix} \text{ with } z\in \CC/(\mathbb Z+i \mathbb Z) \]
and $k \in \mathbb C$ with periodic boundary conditions where $D_{z} = \frac{1}{i} \partial_z.$
This operator can be diagonalised with Bloch functions
$
\left( \begin{array}{c} 1 \\ \pm \frac{\bar k }{\vert k \vert } \end{array} \right)$ and eigenvalues  $\pm \vert k \vert.
$
This means that the original frequencies satisfy for $k$ small enough $\omega(k)^2 = E_D \pm \vert k \vert,$ where $E_D$ is the energy level of the Dirac cones in the spectrum of $\mathcal L(k).$ 

\subsubsection{Strain-induced pseudomagnetic fields}
Since photons do not directly interact with electromagnetic fields, the effect of electromagnetic fields on the electronic band structure cannot be modelled directly using external fields. However, it has been observed that physical strain can be used to imitate the effect of electromagnetic fields, see e.g.~\cite{PhysRevA.103.013505}. Indeed, consider the displacement field $T(z) = z+ u(z)$ with $u$ a displacement vector. In the case of a honeycomb-lattice symmetries, the effective strain-induced magnetic potential is then given as 
\[ A(z) \propto \tr(U(z)\sigma_3) -i 
\tr(U(z)\sigma_1)\]
where $U(z) = \frac{1}{2}(Du+Du^t),$ with Jacobi matrices $Du.$
This way, by applying suitable strain, we obtain a pseudo-magnetic potential $A(z)= \frac{iBz}{2}$ associated with a constant magnetic field $B>0$ and the effective Hamiltonian is the magnetic Dirac operator
\begin{equation}
\label{eq:Hamiltonian_B}
H = \begin{pmatrix} 0 & 2D_{\bar z}- \frac{1}{2}i B z \\ 
2D_{z} + \frac{1}{2}i B \bar z & 0 \end{pmatrix},
\end{equation}
with constant magnetic field $B>0.$

An example of, up to a change of gauge, pseudomagnetic field inducing strain is $$T(z) = z + \kappa^2 \Im(z)^2.$$

\begin{figure}
\includegraphics[width=2.5cm]{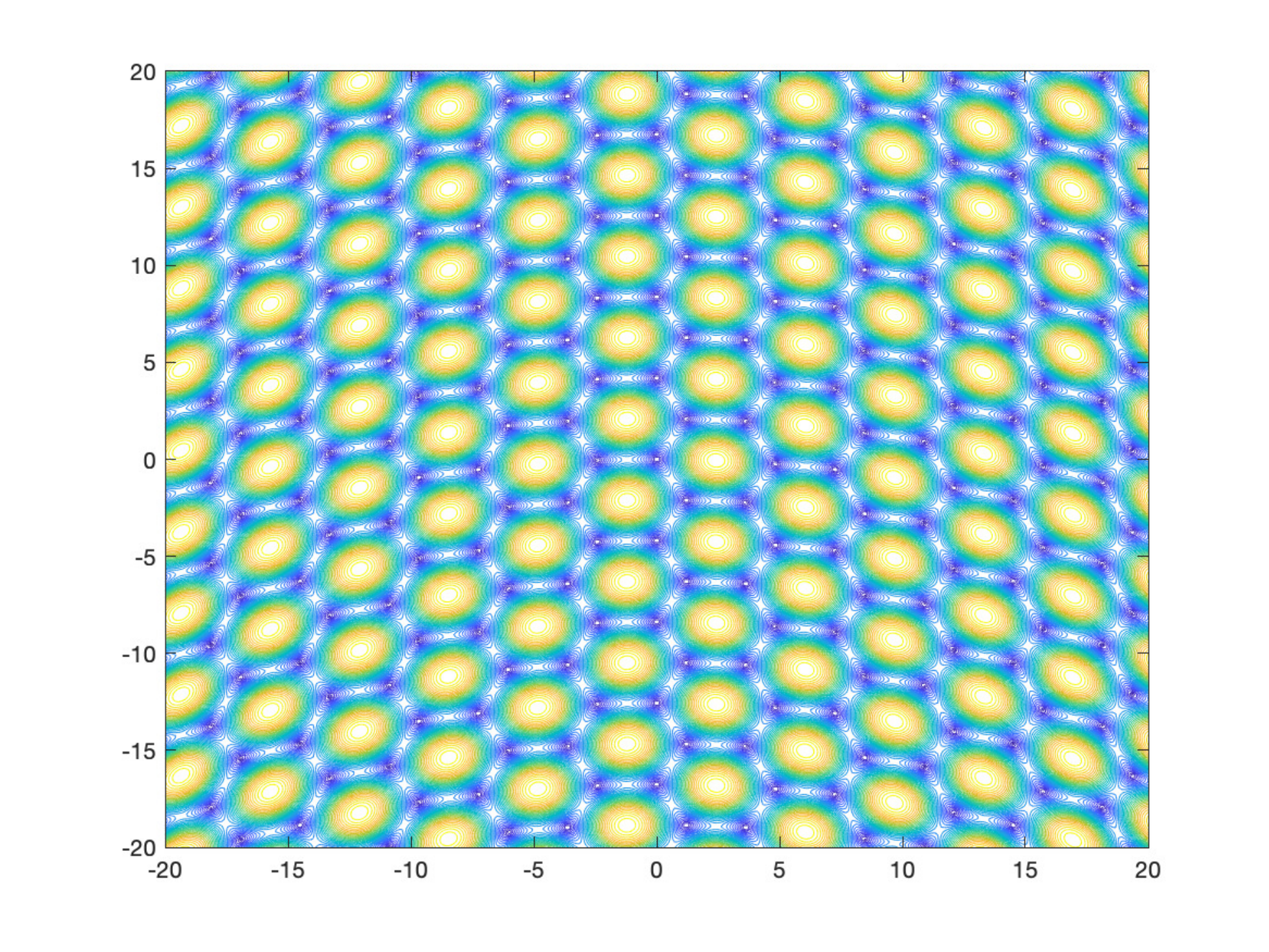}
\includegraphics[width=2.5cm]{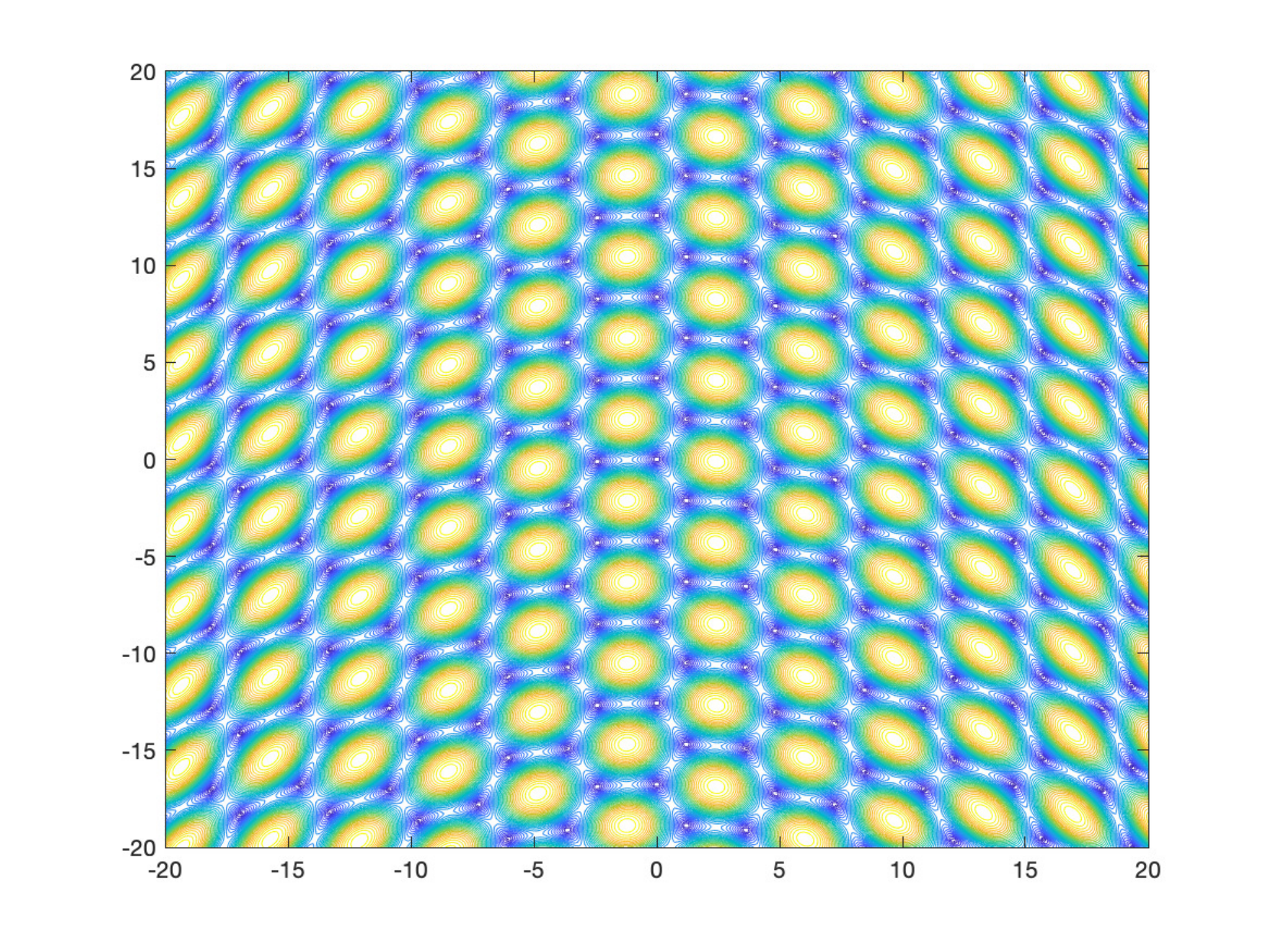}
\includegraphics[width=2.5cm]{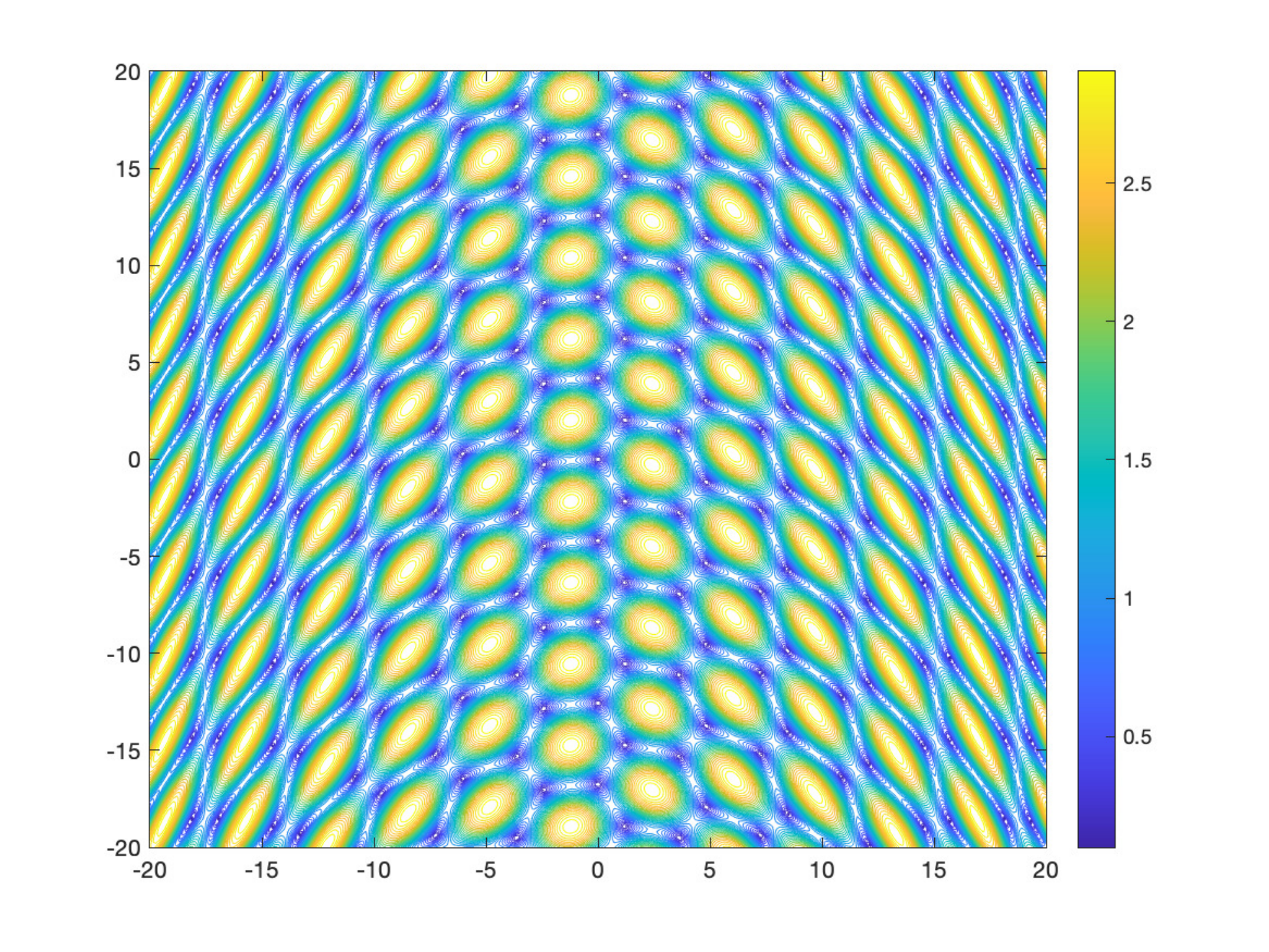}
\caption{Different strain profiles leading to effective constant magnetic fields. Increasing field strenghts from left to right.}
\end{figure}

We then have the conjugation relation 
\[ e^{-B \vert z \vert^2/4} (2D_{\bar z}) e^{B\vert z \vert^2/4} = 2D_{\bar z}-\frac{i}{2}Bz.\]
The infinitely degenerate ground state satisfies then 
\[ u(z) = f(z) e^{-\frac{B\vert z\vert^2}{4}}, f \in \mathscr O(\CC), \int_{\CC} \vert f(z)\vert^2 e^{-B\vert z \vert^2/2} dz<\infty,\]
where $\mathscr O(\CC)$ is the space of entire functions. The infinitely degenerate ground state can be equivalently interpreted as a flat band, see e.g.~\cite{PhysRevB.31.2529}.

Let $a = 2D_{z}-\overline{A(z)} ,$ one then has that $[a^*,a]=2B$ which means that to find the zero energy band for the Hamiltonian~\eqref{eq:Hamiltonian_B}, we have to find zero modes to $a^*$ satisfying Bloch-Floquet boundary conditions.

Let $E_D>0$ be the energy of the Dirac point, then the lowest band is just $\omega(k)^2=E_D$ for all $k$. Correspondingly, by the commutation relation of the ladder operators, we find that the other (also flat) bands are of the form $\omega(k)^2 =E_D^2 \pm \sqrt{2n B}$ with $n \in \{0,...,N\}$ for some large enough $N$ which is in fact independent of $k$, since the band is flat.

As a simple example, we may choose photons in the strained crystal with state $$\vert \psi \rangle  = \prod_{i=1}^n \vert 0 \rangle_{k_i}$$ with harmonic oscillator frequencies $\omega(k_i)\equiv E_D.$ In Figure~\ref{fig:reference} we see that high-energy states are in general harder to reconstruct as their characteristic function is more extended in phase space. 
\begin{figure}
    \centering
    \includegraphics[width=4cm]{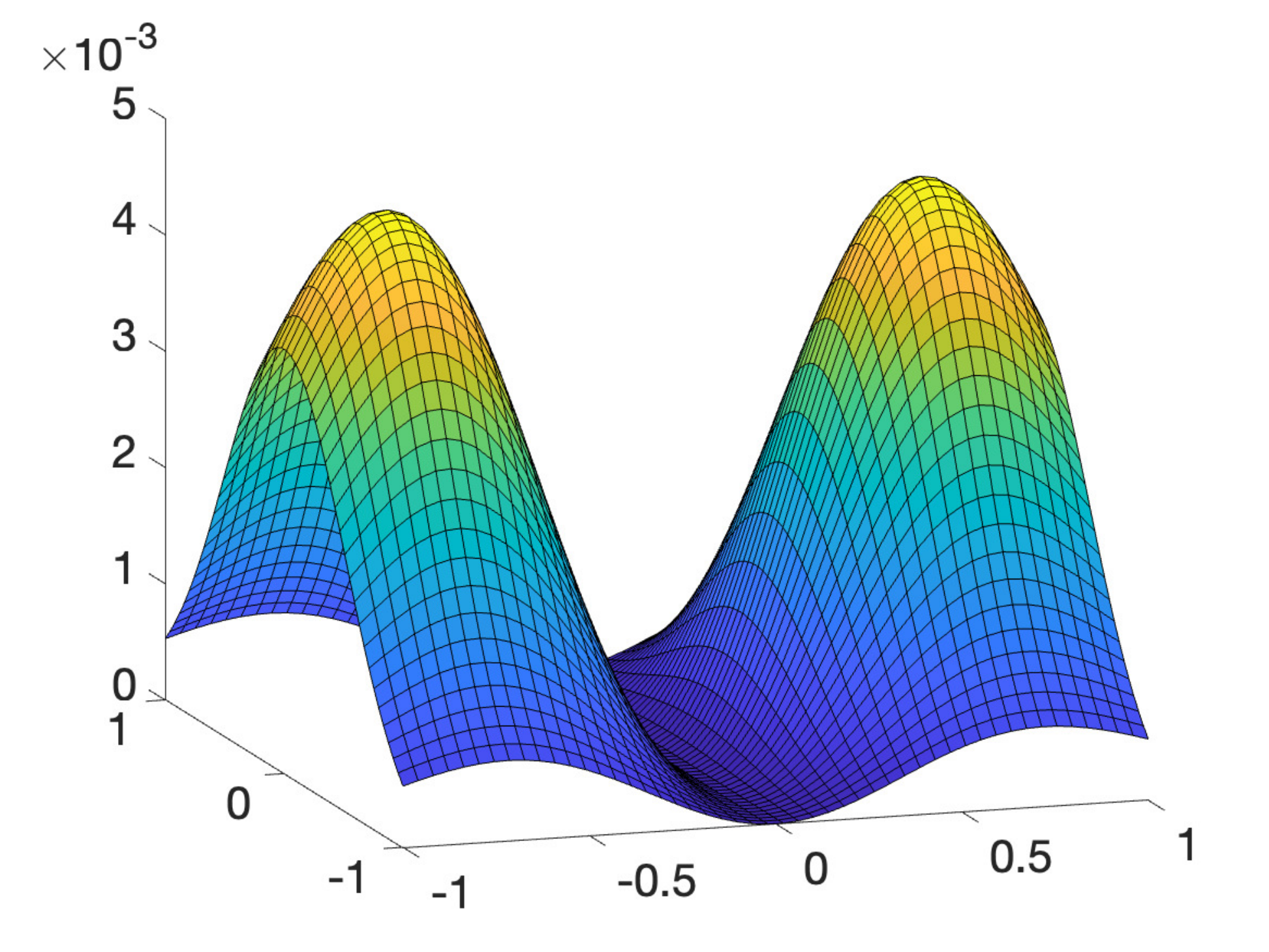}
    \includegraphics[width=4cm]{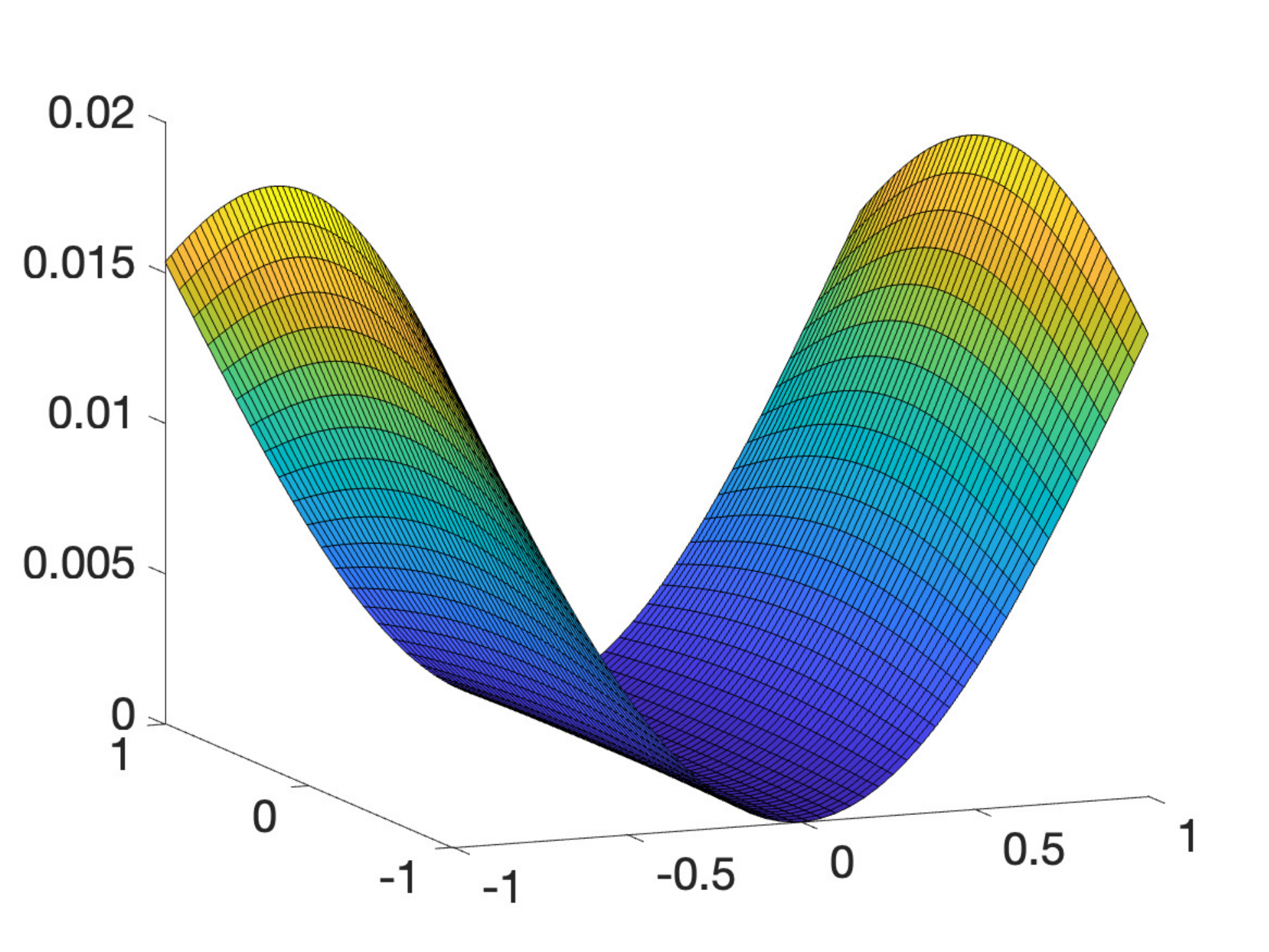}
    \caption{\label{fig:reference} Reconstruction error for single particle characteristic function with $n=5$ particles for $E_D=0.5$ and $E_D=5$ with $N=150$. High energies are harder to reconstruct, as they are more  }
\end{figure}

% %  \textcolor{red}{see conclusion of~\cite{landon2018quantitative}}

% % \section{Left}

% % Comparison to prior work:
% % \begin{itemize}
% % \item~\cite{lvovsky2009continuous,wallentowitz1996unbalanced} homodyne tomography
% % \item For minimax analysis of efficiency of inverse Radon transform, see~\cite{butucea2007minimax}
% \item For old results involving multiple optical modes, see~\cite{raymer1996two,d1999universal} (mostly two-modes, not complex systems, see other references p. 306 of the 2009 review)
% \end{itemize}

\bibliographystyle{unsrt}
\bibliography{Unified-Biblio}

\appendix

\section{Proof of~\eqref{noise_thermal_representation}} \label{folklore_app}

The purpose of this appendix is to provide a self-contained proof of the representation~\eqref{noise_thermal_representation} of the Gaussian white noise channel. We start by observing that it suffices to test~\eqref{noise_thermal_representation} on all projectors $\ketbra{s}$ on the coherent states, where $s\in \R^{2m}$. This is because the linear span of such projectors is dense in the space of Hilbert--Schmidt operators. To prove this latter claim, we can take an arbitrary Hilbert--Schmidt operator $X$ and assume that it is orthogonal to $\Span\{\ketbra{s}\}_{s\in \R^{2m}}$, in formula $\Tr \left[X \ketbra{s} \right] = \braket{s|X|s}=0$ for all $s\in \R^{2m}$. We can express this in words by saying that the Husimi Q-function $Q_X(s)\coloneqq \frac{1}{(2\pi)^m} \braket{s|X|s}$ of $X$ vanishes everywhere. Since the characteristic function~\eqref{chi} of $X$ is a point-wise multiple of the Fourier transform of $Q_X$, in formula~\cite[\S~4.5]{BARNETT-RADMORE}
\bb
\chi_X(t) = e^{\|t\|^2/4} \int d^{2m}s\, Q_X(s)\, e^{i s^\intercal \Omega t}\, ,
\ee
we have that also $\chi_X \equiv 0$ identically. Since the correspondence between Hilbert--Schmid operators and characteristic functions is an isometry and in particular injective~\cite[Theorem~5.3.3]{HOLEVO}, we conclude that $X=0$. Since $X\in \Span\{\ketbra{s}\}_{s\in \R^{2m}}^\perp$ was arbitrary, this entails that $\Span\{\ketbra{s}\}_{s\in \R^{2m}}$ is dense.

Therefore, let us verify~\eqref{noise_thermal_representation} by letting the right-hand side act on an arbitrary $\ketbra{s}$. We obtain that
\begin{align*}
\!\!&\int\!\! \frac{d^{2m} x}{(2\pi \lambda)^m}\, \D(x)\, \tau_{\!\frac{1}{2\lambda}-\frac12}^{\otimes m} \D(x)^\dag \ketbra{s} \D(x)\, \tau_{\!\frac{1}{2\lambda}-\frac12}^{\otimes m} \D(x)^\dag \\
&\ \texteq{(i)} \int\!\! \frac{d^{2m} x}{(2\pi \lambda)^m}\, \D(x)\, \tau_{\!\frac{1}{2\lambda}-\frac12}^{\otimes m} \ketbra{s-x}\, \tau_{\!\frac{1}{2\lambda}-\frac12}^{\otimes m} \D(x)^\dag \\
&\ \texteq{(ii)} \left( \!\frac{2\lambda}{1\!+\!\lambda}\! \right)^{\!\!2m}\!\!\! \int\!\! \frac{d^{2m} x}{(2\pi \lambda)^m}\, e^{-\frac{2\lambda}{(1+\lambda)^2}\, \|s-x\|^2} \\
&\hspace{15ex} \cdot \D(x)\, \Ketbra{\tfrac{1\!-\!\lambda}{1\!+\!\lambda}(s\!-\!x)}\, \D(x)^\dag \\
&\ \texteq{(iii)} \left( \!\frac{2\lambda}{1\!+\!\lambda}\! \right)^{\!\!2m}\!\!\! \int\!\! \frac{d^{2m} y}{(2\pi \lambda)^m}\, e^{-\frac{2\lambda}{(1+\lambda)^2} \|y\|^2} \\
&\hspace{15ex} \cdot \D(s-y)\, \Ketbra{\tfrac{1\!-\!\lambda}{1\!+\!\lambda} y}\, \D(s-y)^\dag \\
&\ \texteq{(iv)} \left( \!\frac{2\lambda}{1\!+\!\lambda}\! \right)^{\!\!2m}\!\!\! \int\!\! \frac{d^{2m} y}{(2\pi \lambda)^m}\, e^{-\frac{2\lambda}{(1+\lambda)^2} \|y\|^2} \Ketbra{s\! -\! \tfrac{2\lambda}{1\!+\!\lambda} y} \\
&\ \texteq{(v)} \int\!\! \frac{d^{2m} z}{(2\pi \lambda)^m}\, e^{-\frac{\|z\|^2}{2\lambda}} \Ketbra{s\! +\! z} \\
&\ \texteq{(vi)} \int\!\! \frac{d^{2m} z}{(2\pi \lambda)^m}\, e^{-\frac{\|z\|^2}{2\lambda}} \D(z) \ketbra{s} \D(z)^\dag \\
&\texteq{(vii)} \cN_\lambda(\ketbra{s})\, .
\end{align*}
Here, (i),~(iv), and~(vi)~follow from~\eqref{coherent} and~\eqref{Weyl}; in~(ii) we employed the identity
\bb
\tau_\nu \ket{t} = \frac{1}{\nu+1}\, e^{-\frac{2\nu+1}{4(\nu+1)^2}\, \|t\|^2} \Ket{\tfrac{\nu}{\nu+1}\,t} ,
\ee
which is readily verified in Fock basis by combining~\eqref{coherent_Fock} and~\eqref{thermal}; in~(iii) we performed the change of variables $y \coloneqq s-x$; in~(v) we set $z \coloneqq -\frac{2\lambda}{1+\lambda} y$; finally, (vii)~is simply~\eqref{noise_random_displacement}. The proof of~\eqref{noise_thermal_representation} is complete.

\section{Fourier representation for characteristic function} \label{Fourier_f_mu_T_app}

Throughout this appendix, we will prove rigorously that the function $f_{\mu,T}$ defined by~\eqref{f_mu_T} has a well-defined Fourier transform, which moreover coincides with itself. This will establish~\eqref{Fourier_f_mu_T}. We formalise these facts as follows:

\begin{lemma}
For every measure $\mu$ on the symplectic group and for every symplectic $T$, the function $f_{\mu,T}$ defined in~\eqref{f_mu_T} is in $L^1(\R^{2m})\cap L^2(\R^{2m})$; more precisely,

\noindent
%\bb
\vspace{-0.35cm}
\begin{align}
& \int d^{2m} x\, \left| f_{\mu,T}(x) \right| = (2\pi)^m\, ,\qquad \int d^{2m} x\, f_{\mu,T}^2(x) \leq \pi^m\, .\\
%\ee
%\end{align}
%{\hbox{Moreover,}} \quad  $\widetilde{f}_{\mu,T}(x) = f_{\mu,T}(x)$.
& {\hbox{Moreover,}} \quad \widetilde{f}_{\mu,T}(x) = f_{\mu,T}(x).\nonumber
\end{align}
\end{lemma}

\begin{proof}
We start with the first identity, which follows by writing
\begin{align*}
\int d^{2m} x\, \left| f_{\mu,T}(x)\right|&= \int d\mu(S) \int d^{2m} x\, e^{- \frac12\, x^\intercal (TS)^{^\intercal} (TS)\, x} \\
&\texteq{(ii)} (2\pi)^m\, .
\end{align*}
Here in~(ii) we combined the Gaussian integral formula
\bb
\int d^N x\, e^{-\frac12 x^\intercal A x + i t^\intercal x} = \sqrt{\frac{(2\pi)^N}{\det A}}\, e^{-\frac12 t^\intercal A^{-1} t}\, ,
\label{Gaussian_integral}
\ee
valid for $N\times N$ positive definite matrices $A>0$, with the observation that $\det (TS)=1$ because both $T$ and $S$ --- and hence $TS$, too --- are symplectic.

Since $f_{\mu,T}$ is in $L^1(\mathbb{R}^{2m})$, positive, and bounded by $1$, it is clear that it must be also in $L^2(\mathbb{R}^{2m})$. However, it requires little effort to prove directly the second inequality, which is anyway tighter than what one would obtain simply by the above observation. It suffices to compute
\bb
&\int d^{2m} x\, f_{\mu,T}^2(x) \\
&\ \texteq{(iii)} \int\!\! d\mu(S)\, d\mu(S') \int\!\! d^{2m}x\, e^{-\frac12 x^{\!\intercal} \left((TS)^{\!^\intercal}\! (TS) + (TS')^{\!^\intercal}\! (TS') \right) x} \\
&\ \texteq{(iv)} \int d\mu(S)\, d\mu(S')\, \frac{(2\pi)^m}{\sqrt{\det \left((TS)^{\!^\intercal}\! (TS) + (TS')^{\!^\intercal}\! (TS') \right)}} \\
&\ \stackrel{\text{(v)}}{\leq} \int\!\! d\mu(S)\, d\mu(S')\, (2\pi)^m \frac{1}{2^m} = \pi^m\, .
\ee
Here, (iii)~is just Tonelli's theorem, (iv)~is again an application of~\eqref{Gaussian_integral}, and~(v) follows from the Minkowski determinant inequality, which states that the function $A\mapsto (\det A)^{1/N}$ is concave on the set of positive semi-definite $N\times N$ matrices~\cite[Eq.~(4.21)]{BHATIA}, combined with the fact that $\det (TS)=1$.

Now that we know that $f_{\mu,T}$ is highly regular, we can manipulate the integrals in its Fourier transform more safely:
\begin{align*}
\widetilde{f}_{\mu,T}(x) &= \int \frac{d^{2m}u}{(2\pi)^m}\, f_{\mu,T}(u) \, e^{i x^\intercal \Omega u} \\
&= \int d\mu(S) \int \frac{d^{2m}u}{(2\pi)^m}\, e^{-\frac12 x^\intercal (TS)^\intercal (TS) x + i x^\intercal \Omega u} \\
&= \int d\mu(S)\, e^{-\frac12 x^\intercal \Omega (TS)^{-1} (TS)^{-\intercal} \Omega^\intercal x} \\
&= \int d\mu(S)\, e^{-\frac12 x^\intercal (TS)^{\intercal} (TS)  x} = f_{\mu,T}(x)\, .
\end{align*}
In the above derivation, we have used~\eqref{Gaussian_integral} together with the fundamental identity $W\Omega = \Omega W^{-\intercal}$, valid for any symplectic matrix $W$.
\end{proof}

\section{The double truncation lemma} \label{double_truncation_app}

\begin{manuallemma}{\ref{double_truncation_lemma}}
For all $\alpha>0$, all non-negative integers $m$ (number of modes) and $M$ (Fock truncation number), and all $\eta\geq 0$, it holds that
\bb
&\left\|\mathcal{L}_{(I+N_m)^\alpha} \mathcal{R}_{(I+N_m)^{\alpha}}\left(\mathcal{P}_M-\widetilde{\mathcal{P}}_{{M}}\right)\right\|_{2\to 1} \\
&\qquad \leq (m M+1)^{2\alpha} (M+1)^m\, 3^{mM} \left(\frac{\Gamma\left(2M+1, \tfrac{\eta^2}{2} \right)}{(2M)!}\right)^{m/2} \\
&\qquad = (m M+1)^{2\alpha} (M+1)^m\, 3^{mM} e^{-\frac{m}{4} \eta^2} \left(\sum_{p=0}^{2M} \frac{\eta^{2p}}{2^p p!} \right)^{m/2} ,
\ee
where in the second line we introduced the incomplete Gamma function, given by~\cite[\S~6.5]{ABRAMOWITZ}
\bb \label{incomplete_Gamma}
\Gamma(N,x)\coloneqq \int_x^\infty ds\, e^{-s} s^{N-1} = (N-1)!\, e^{-x} \sum_{p=0}^{N-1} \frac{x^p}{p!}
\ee
for $N>0$ and $x\geq 0$.
\end{manuallemma}

\begin{proof}
Let $T$ be an arbitrary Hilbert--Schmidt operator acting on $\cH_m$, the Hilbert space of an $m$-mode system. In what follows, we denote by $[M]^m\coloneqq \{0,\ldots, M\}^m$ the set of possible Fock number vectors with all entries bounded by $M$, and for $\mathbf{n}\in [M]^m$ we will use the notation $|\mathbf{n}|\coloneqq \sum_{j=0}^m \mathbf{n}(j)$ to denote the corresponding total photon number. We start by writing 
\bb
&\left\|\mathcal{L}_{(I+N_m)^\alpha} \mathcal{R}_{(I+N_m)^{\alpha}}\left(\mathcal{P}_M(T)-\widetilde{\mathcal{P}}_{{M}}(T)\right)\right\|_1 \\
&\ \le \sum_{\mathbf n_1, \mathbf n_2 \in [M]^m} \hspace{-2ex} (1\!+\!\vert \mathbf n_1 \vert )^{\alpha}(1\!+\!\vert \mathbf n_2 \vert)^{\alpha} \left| \Tr\left[\left(\widetilde{Z}_{\mathbf{n}_1\mathbf{n}_2}-\ketbraa{\mathbf{n}_1}{\mathbf{n}_2}\right) T\right] \right| \\
&\ \le \sum_{\mathbf n_1, \mathbf n_2 \in [M]^m} \hspace{-2ex} (1\!+\!\vert \mathbf n_1 \vert )^{\alpha}(1\!+\!\vert \mathbf n_2 \vert)^{\alpha} \left\|\widetilde{Z}_{\mathbf{n}_1\mathbf{n}_2}-\ketbraa{\mathbf{n}_1}{\mathbf{n}_2}\right\|_2 \|T\|_2 \\
&\ \le (m M\!+\!1)^{2\alpha} (M\!+\!1)^m\!\! \max_{\mathbf n_1, \mathbf n_2 \in [M]^m} \left\|\widetilde{Z}_{\mathbf{n}_1\mathbf{n}_2}-\ketbraa{\mathbf{n}_1}{\mathbf{n}_2}\right\|_2 \|T\|_2
\label{double_truncation_eq1}
\ee
We can then continue by estimating
%Thanks to Plancherel's relation~\eqref{Plancherel} and~\eqref{xi}, we can then write
\bb
&\left\|\widetilde{Z}_{\mathbf{n}_1\mathbf{n}_2} - \ketbraa{\mathbf{n}_1}{\mathbf{n}_2} \right\|^2_2 \\  
&\ \texteq{(i)} \prod_{j=1}^m \int \frac{d^{2}u_j}{2\pi} \left|\chi_{\widetilde{Z}_{\mathbf{n}_1(j) \mathbf{n}_2(j)}}(u_j) - \chi_{\ket{\mathbf{n}_1(j)}\!\bra{\mathbf{n}_2(j)}}\big(u^{(j)}\big) \right|^2 \\
&\ \texteq{(ii)} \prod_{j=1}^m \int \frac{d^{2}u_j}{2\pi} \left| \chi_{\ket{\mathbf{n}_1(j)}\!\bra{\mathbf{n}_2(j)}}(u_j) \right|^2 \left( 1 - \xi_{\eta,R}(u_j) \right)^2 \\
&\ \textleq{(iii)} \prod_{j=1}^m \int_{\|u_j\|\geq \eta} \frac{d^{2}u_j}{2\pi} \left| \chi_{\ket{\mathbf{n}_1(j)}\!\bra{\mathbf{n}_2(j)}}(u_j) \right|^2
\label{double_truncation_eq2}
\ee
Here, (i)~comes from~\eqref{Plancherel}, in~(ii) we remembered~\eqref{tildechin1n2}, and in~(iii) we observed that $1 - \xi_{\eta,R}(u_j)=0$ if $\|u_j\|\leq \eta$, and estimated $1 - \xi_{\eta,R}(u_j)\leq 1$ otherwise. In order to continue, we should upper bound $\chi_{\ket{n}\!\bra{n'}}(u)$ for vectors $u\in \R^2$ of sufficiently large modulus, and for arbitrary integers $n,n'\leq M$. To this end, let us write
\begin{align*}
\left|\chi_{\ket{n}\!\bra{n'}}(u)\right|^2 &= \left|\braket{n'|\D(u)|n}\right|^2 \\
&\texteq{(iv)} e^{-\frac12 \|u\|^2} \left|\braket{n'|e^{\alpha(u)a^\dag} e^{-\alpha(u)^* a}|n} \right|^2 \\
&\textleq{(v)} e^{-\frac12 \|u\|^2} \left\|e^{-\alpha(u)^* a}\ket{n}\right\|^2 \left\|e^{\alpha(u)^* a}\ket{n'}\right\|^2 \\
&\texteq{(vi)} e^{-\frac12 \|u\|^2} \left\| \sum_{k=0}^n \frac{(-\alpha(u))^k}{k!} \sqrt{\frac{n!}{(n-k)!}}\, \ket{n-k}\right\|^2 \\
&\quad \cdot \left\| \sum_{k=0}^{n'} \frac{(\alpha(u)^*)^k}{k!} \sqrt{\frac{n'!}{(n'-k)!}}\, \ket{n'-k}\right\|^2 \\
&= e^{-\frac12 \|u\|^2} \left(\sum_{k=0}^n \frac{|\alpha(u)|^{2k}}{(k!)^2} \frac{n!}{(n-k)!} \right) \\
&\quad \cdot \left(\sum_{k=0}^{n'} \frac{|\alpha(u)|^{2k}}{(k!)^2} \frac{n'!}{(n'-k)!} \right) \\
&\texteq{(vii)} e^{-\frac12 \|u\|^2} L_n\big(-|\alpha(u)|^2\big)\, L_{n'}\big(-|\alpha(u)|^2\big) \\
&= e^{-\frac12 \|u\|^2} L_n\big(-\tfrac12 \|u\|^2\big)\, L_{n'}\big(-\tfrac12 \|u\|^2\big) \\
&\textleq{(viii)} e^{-\frac12 \|u\|^2} L_M\big(-\tfrac12 \|u\|^2\big)^2 .
\end{align*}
The steps of the above derivation can be justified as follows: (iv)~is an application of~\eqref{D_split_complex}, and we recall that $\alpha(u) \coloneqq \frac{1}{\sqrt2} (u_1+u_2)$; (v)~is simply the Cauchy--Schwarz inequality; (vi)~can be verified by a repeated application of~\eqref{Fock_creation}; in~(vii) we introduced the Laguerre polynomials
\bb \label{Laguerre}
L_n(x) \coloneqq \sum_{k=0}^n \binom{n}{k} \frac{(-x)^k}{k!}\, ;
\ee
and in~(viii) we noted that $L_n(-x)$ is monotonically non-decreasing in the integer $n$ for $x\geq 0$, essentially because $\binom{n}{k}$ is monotonically non-decreasing in $n$ for fixed $k$. Continuing, we deduce that
\begin{align*}
&\int_{\|u\|\geq \eta} \frac{d^2u}{2\pi} \left|\chi_{\ket{n}\!\bra{n'}}(u)\right|^2 \\
&\quad \leq\ \int_{\|u\|\geq \eta} \frac{d^2u}{2\pi}\, e^{-\frac12 \|u\|^2} L_M\big(-\tfrac12 \|u\|^2\big)^2 \\
&\quad \texteq{(ix)}\ \int_{s\geq \eta^2/2} ds\, e^{-s} L_M\big(-s\big)^2 \\
&\quad \texteq{(x)}\ \sum_{h,k=0}^M \binom{M}{h} \binom{M}{k} \frac{1}{h! k!} \int_{s\geq \eta^2/2} ds\, e^{-s} s^{h+k} \\
&\quad \texteq{(xi)}\ \sum_{h,k=0}^M \binom{M}{h} \binom{M}{k} \frac{1}{h! k!}\, \Gamma\left(h+k+1, \tfrac{\eta^2}{2} \right) \\
&\quad =\ \sum_{h,k=0}^M \binom{M}{h} \binom{M}{k} \binom{h+k}{h} \frac{\Gamma\left(h+k+1, \tfrac{\eta^2}{2} \right)}{(h+k)!} \\
&\quad \textleq{(xii)}\ \sum_{h,k=0}^M \binom{M}{h} \binom{M}{k}\, 2^{h+k} \frac{\Gamma\left(2M+1, \tfrac{\eta^2}{2} \right)}{(2M)!} \\
&\quad \texteq{(xiii)}\ 3^{2M}\, \frac{\Gamma\left(2M+1, \tfrac{\eta^2}{2} \right)}{(2M)!}
\end{align*}
Here: in~(ix) we introduced the new variable $s\coloneqq \tfrac12 \|u\|^2$; in~(x) we expanded the square thanks to the expression~\eqref{Laguerre} for the Laguerre polynomials; in~(xi) we introduced the incomplete Gamma function~\eqref{incomplete_Gamma}; in~(xii), besides using that $\binom{N}{k}\leq 2^N$ for all non-negative integers $N,k$, we also observed that by virtue of the expansion in~\eqref{incomplete_Gamma} it is easy to verify that $\tfrac{1}{(N-1)!} \Gamma(N,x)$ is a monotonically non-decreasing function of the integer $N\geq 1$ for all fixed $x>0$; finally, (xiii)~is just the binomial theorem, applied twice to $3^M = (1+2)^M$.

Plugging the above estimate into~\eqref{double_truncation_eq2}, we find that
\bb
\left\|\widetilde{Z}_{\mathbf{n}_1\mathbf{n}_2} - \ketbraa{\mathbf{n}_1}{\mathbf{n}_2} \right\|_2 \leq 3^{mM}\, \left(\frac{\Gamma\left(2M+1, \tfrac{\eta^2}{2} \right)}{(2M)!}\right)^{m/2}
\ee
for all $\mathbf{n}_1,\mathbf{n}_2\in [M]^m$. From~\eqref{double_truncation_eq1} we then deduce that
\bb
&\left\|\mathcal{L}_{(I+N_m)^\alpha} \mathcal{R}_{(I+N_m)^{\alpha}}\left(\mathcal{P}_M(T)-\widetilde{\mathcal{P}}_{{M}}(T)\right)\right\|_1 \\
&\quad \leq (m M\!+\!1)^{2\alpha} (M\!+\!1)^m\, 3^{mM} \left(\frac{\Gamma\left(2M\!+\!1, \tfrac{\eta^2}{2} \right)}{(2M)!}\right)^{m/2}\!\! \|T\|_2 \, ,
\ee
concluding the proof.
\end{proof}
%!!!

\end{document}